\documentclass[submission,copyright,creativecommons]{eptcs}
\providecommand{\event}{ITRS 2018}

\usepackage{underscore}           

\usepackage{marginnote}
\usepackage{microtype}
\usepackage{multicol}
\usepackage{mathtools,amsthm}
\usepackage{ragged2e} 
\usepackage{graphicx}
\usepackage{cmll}
\hypersetup{hidelinks}

\newcommand{\macrospath}{./}

\usepackage[utf8]{inputenc}

\usepackage{amssymb}
\usepackage{stmaryrd}
\usepackage{bussproofs}
\usepackage{mathtools}
\usepackage{bbold}

\usepackage{ifthen}
\usepackage{xspace}
\usepackage{fancybox}
\usepackage{hhline}

\usepackage{xcolor}

\newcommand{\sem}[1]{[\![#1]\!]}

\newcommand{\ignore}[1]{}

\newcommand{\myinput}[1]{\ifthenelse{\boolean{withimages}}{\input{#1}}{}}

\newcommand{\reflemma}[1]{Lemma~\ref{l:#1}}
\newcommand{\reflemmas}[2]{Lemmas~\ref{l:#1} and \ref{l:#2}} 
\newcommand{\reflemmap}[2]{Lemma~\ref{l:#1}.\ref{p:#1-#2}}
\newcommand{\reflemmasp}[4]{Lemmas~\ref{l:#1}.\ref{p:#1-#2} and \ref{l:#3}.\ref{p:#3-#4}}

\newcommand{\refpoint}[1]{Point~\ref{p:#1}}
\newcommand{\refpoints}[2]{Points~\ref{p:#1}-\ref{p:#2}}

\newcommand{\refthm}[1]{Thm.~\ref{thm:#1}}

\newcommand{\refprop}[1]{Prop.~\ref{prop:#1}}
\newcommand{\refpropp}[2]{Prop.~\ref{prop:#1}.\ref{p:#1-#2}} 
\newcommand{\refpropps}[3]{Prop.~\ref{prop:#1}.\ref{p:#1-#2}-\ref{p:#1-#3}} 
\newcommand{\refsect}[1]{Sect.~\ref{sect:#1}}

\renewcommand{\refeq}[1]{(\ref{eq:#1})} 
\newcommand{\reffig}[1]{Fig.~\ref{fig:#1}}
\newcommand{\refcoro}[1]{Cor.~\ref{coro:#1}}
\newcommand{\refcor}[1]{Cor.\,\ref{coro:#1}}
\newcommand{\refdef}[1]{Definition~\ref{def:#1}}

\newcommand{\refrmk}[1]{Rmk.~\ref{rmk:#1}} 
\newcommand{\refex}[1]{Ex.~\ref{ex:#1}}

\newcommand{\ie}{\textit{i.e.}\xspace}
\newcommand{\eg}{\textit{e.g.}\xspace}
\newcommand{\ih}{\textit{i.h.}\xspace}

\newcommand{\defeq}{\coloneqq} 
\newcommand{\eqdef}{\eqqcolon} 
\newcommand{\grameq}{\Coloneqq} 

\newcommand{\nat}{\mathbb{N}}
\newcommand{\size}[1]{|#1|}

\renewcommand{\l}{\lambda}
\newcommand{\isub}[2]{\{#1/#2\}}
\renewcommand{\isub}[2]{\{#1{\shortleftarrow}#2\}}

\newcommand{\rootRew}[1]{\mapsto_{#1}}
\newcommand{\Rew}[1]{\rightarrow_{#1}}

\newcommand{\rtobv}{\rootRew{\betav}} 

\newcommand{\slsym}{\sigma_1}
\newcommand{\srsym}{\sigma_3}

\newcommand{\rtosl}{\rootRew{\slsym}}
\newcommand{\rtosr}{\rootRew{\srsym}}

\newcommand{\betav}{{\beta_v}} 

\newcommand{\betavm}{{\bilancia\beta_v}} 
\newcommand{\tobv}{\Rew{\betav}} 

\newcommand{\tobvm}{\Rew{\betavm}} 

\newcommand{\tosig}{\Rew{\sigma}} 
\newcommand{\sigm}{\bilancia{\sigma}}

\newcommand{\tosigm}{\Rew{\sigm}} 

\newcommand{\sigl}{\bilancia{\sigma}_1}
\newcommand{\sigr}{\bilancia{\sigma}_3}
\newcommand{\tosl}{\Rew{\sigl}}
\newcommand{\tosr}{\Rew{\sigr}}

\newcommand{\toshuf}{\tovm} 
\newcommand{\tovm}{\Rew{\vmsym}} 

\newcommand{\bilanciasym}{\flat}
\newcommand{\bilancia}[1]{{#1}^{\bilanciasym\!}}
\newcommand{\vmsym}{\mathsf{shuf}} 
\newcommand{\shuf}{\vmsym} 

\newcommand{\shufeq}{\shufextth} 
\newcommand{\shufeqext}{\shufeqext} 

\newcommand{\tm}{t}
\newcommand{\tmtwo}{u}
\newcommand{\tmthree}{s}
\newcommand{\tmfour}{r}

\newcommand{\tmp}{\tm'}
\newcommand{\tmtwop}{\tmtwo'}
\newcommand{\tmthreep}{\tmthree'}

\newcommand{\var}{x}
\newcommand{\vartwo}{y}
\newcommand{\varthree}{z}

\newcommand{\val}{v}
\newcommand{\valtwo}{\val'}

\newcommand{\ctxholep}[1]{\langle #1\rangle}
\newcommand{\ctxhole}{\ctxholep{\cdot}}

\newcommand{\ctx}{C}

\newcommand{\ctxp}[1]{\ctx\ctxholep{#1}}

\newcommand{\arbctxp}[1]{\arbctxp{#1}}
\newcommand{\arbctxtwop}[1]{\arbctxtwop{#1}}

\newcommand{\deriv}{d}
\newcommand{\derivp}{d'} 

\newcommand{\mctx}{B}

\newcommand{\mctxp}[1]{\mctx\ctxholep{#1}}

\newcommand{\la}[1]{\lambda #1.}

\newcommand{\myproof}[1]{
\ifthenelse{\boolean{omitproofs}}{\begin{IEEEproof} Proof available but omitted for readability. \end{IEEEproof}}{#1}}

\newcommand{\withproofs}[1]{\ifthenelse{\boolean{withproofs}}{#1}{}}

\newcommand{\withoutproofs}[1]{\ifthenelse{\boolean{withproofs}}{}{#1}}

\newcommand{\NoteProof}[1]{\withproofs{\marginpar{\scriptsize \ \ Proof p.\,{\pageref{#1}}}}} 
\newcommand{\NoteState}[1]{\withproofs{\marginpar{\scriptsize \ \ See p.~{\pageref{#1}}}}} 

\newcommand{\firecalc}{\lambda_\mathsf{fire}}

\newcommand{\shufcalc}{\lambda_\shuf}
\newcommand{\shufcalcm}{\lambda_\bilancia{\shuf}}

\newcommand{\doubt}[1]{}

\newcommand{\Rule}{\mathsf{r}}

\newcounter{numberone}
\newcounter{numberoneroman}
\newcounter{numberonealph}

\newcommand{\emptymset}{\zero}

\newboolean{withproofs}
\setboolean{withproofs}{false}

\renewcommand{\NoteProof}[1]{
\marginnote{
\scriptsize{Proof p.\,{\pageref{#1}}}}}
\renewcommand{\NoteState}[1]{
\marginnote{
\scriptsize{See p.\,{\pageref{#1}}}}}

\newcommand{\concl}[4]{#1 \vartriangleright #2 \vdash #3 \colon\! #4}

\newcommand{\Type}[2]{#1 \vartriangleright #2}

\newcommand{\Pair}[2]{#1 \multimap #2}
\renewcommand{\isub}[2]{\{#2/#1\}}
\newcommand{\Dom}[1]{\mathsf{dom}(#1)}
\newcommand{\Fv}[1]{\mathsf{fv}(#1)}
\newcommand{\Ax}{\mathsf{ax}}

\renewcommand{\sem}[2]{\llbracket #1 \rrbracket_{#2}}

\newcommand{\Fin}{\mathrm{f}}

\newcommand{\MultiFin}[1]{\M_\Fin(#1)}
\newcommand{\U}{\mathcal{U}}
\newcommand{\M}{\mathcal{M}}
\newcommand{\Nat}{\mathbb{N}}

\renewcommand{\vmsym}{\mathsf{sh}}
\renewcommand{\shufeq}{\simeq_\shuf}
\newcommand{\betaveq}{\simeq_\betav}
\newcommand{\shufm}{\bilancia{\shuf}}
\newcommand{\toshufm}{\Rew{\shufm}}
\newcommand{\rtoshuf}{\rootRew{\shuf}}
\newcommand{\rtos}{\rootRew{\sigma}}

\newcommand{\RevTo}[1]{{\,}_{#1}\!\!\leftarrow}

\newcommand{\MRevTo}[1]{{\,}_{#1}^*\!\!\leftarrow}

\newcommand{\Balanced}{Balanced}

\newcommand{\anf}{a}
\newcommand{\wnf}{n}
\newcommand{\anfSet}{\Lambda_a}
\newcommand{\wnfSet}{\Lambda_n}
\newcommand{\valSet}{\Lambda_v}

\DeclarePairedDelimiter\abs{\lvert}{\rvert}%

\renewcommand{\size}[1]{\abs{#1}}
\newcommand{\sizeZero}[1]{\abs{#1}_{\bilanciasym\!}}

\newcommand{\Lengsym}{\mathrm{leng}}
\newcommand{\Leng}[1]{\Lengsym(#1)}
\newcommand{\Lengbv}[1]{\Lengsym_{\betavm}(#1)}

\renewcommand{\emptymset}{[\,]}

\theoremstyle{definition}
\newtheorem{definition}{Definition}

\newtheorem{example}[definition]{Example}
\newtheorem{remark}[definition]{Remark} 

\theoremstyle{plain}
\newtheorem{lemma}[definition]{Lemma} 
\newtheorem{theorem}[definition]{Theorem}
\newtheorem{proposition}[definition]{Proposition} 
\newtheorem{corollary}[definition]{Corollary} 
\newtheorem{lemmaAppendix}{Lemma} 
\newtheorem{theoremAppendix}{Theorem} 
\newtheorem{propositionAppendix}{Proposition} 
\newtheorem{corollaryAppendix}{Corollary} 

\title{Towards a Semantic Measure of the Execution Time in Call-by-Value lambda-Calculus (Long Version)}
\author{Giulio Guerrieri
\institute{Dipartimento di Informatica -- Scienza e Ingegneria (DISI), Universit\`a di Bologna, Bologna, Italy}
\email{\href{mailto:giulio.guerrieri@unibo.it}{giulio.guerrieri@unibo.it}}
}
\def\titlerunning{Towards a Semantic Measure of the Execution Time in Call-by-Value lambda-Calculus}
\def\authorrunning{G. Guerrieri}

\begin{document}

\maketitle




\begin{abstract}
  We investigate the possibility of a semantic account of the execution time (i.e. the number of $\betav$-steps leading to the normal form, if any) for the shuffling calculus, an extension of Plotkin's 
  call-by-value $\lambda$-calculus. For this purpose, we use a linear logic based denotational model that can be seen as a non-idempotent intersection type system: relational semantics. 
  Our investigation is inspired by similar ones for linear logic proof-nets and untyped call-by-name $\lambda$-calculus. 
  We first prove a qualitative result: a (possibly open) term is normalizable for weak reduction (which does not reduce under abstractions) if and only if its interpretation is not empty. We then show that the size of type derivations can be used to measure the execution time. Finally, we show that, differently from the case of linear logic and call-by-name $\lambda$-calculus, the quantitative information enclosed in type derivations does not lift to types (i.e.~to the interpretation of terms). 
  To get a truly semantic measure of execution time in a call-by-value setting, we conjecture that a refinement of its \mbox{syntax and operational semantics is needed.}
\end{abstract}

\section{Introduction}
\label{sect:intro}

Type systems enforce properties of programs, such as termination or deadlock-freedom. 
The guarantee provided by most type systems for the $\lambda$-calculus is \emph{termination}. 

\emph{Intersection types} have been introduced as a way of extending simple types for the $\lambda$-calculus to ``finite polymorphism'', by adding a new type constructor $\cap$ and new typing rules governing it. 
Contrary to simple types, intersection types provide a sound and \emph{complete} characterization of termination: not only typed programs terminate, but all terminating programs are typable as well (see \cite{DBLP:journals/aml/CoppoD78,DBLP:journals/ndjfl/CoppoD80,Pottinger80,DBLP:books/daglib/0071545} where different intersection type systems characterize different notions of normalization). 
Intersection types are idempotent, that is, they verify the equation $A \cap A = A$.
This corresponds to an interpretation of a typed term $\tm \colon A \cap B$ as ``$\tm$ can be used both as data of type $A$ and as data of type $B$''. 

More recently \cite{DBLP:conf/tacs/Gardner94,DBLP:journals/logcom/Kfoury00,DBLP:conf/icfp/NeergaardM04,Carvalho07,deCarvalho18} (a survey can be found in \cite{DBLP:journals/igpl/BucciarelliKV17}), \emph{non-idempotent} variants of intersection types have been introduced: they are obtained by dropping the equation $A \cap A = A$.
In a non-idempotent setting, the meaning of the typed term $\tm \colon A \cap A \cap B$ is refined as ``$\tm$ can be used twice as data of type $A$ and once as data of type $B$''.
This could give to programmers a way to keep control on the performance of their code and to count resource consumption.
Finite multisets are the natural setting to interpret the associative, commutative and non-idempotent connective $\cap$: 
if $A$ and $B$ are non-idempotent intersection types, 
the multiset $[A, A, B]$ represents the non-idempotent intersection type $A \cap A \cap B$. 

Non-idempotent intersection types have two main features, both enlightened 
by de Carvalho \cite{Carvalho07,deCarvalho18}:
\begin{enumerate}
  \item \emph{Bounds on the execution time}: they go beyond simply qualitative characterisations of termination, as type derivations provide \emph{quantitative} bounds on the execution time (\ie~on
  the number of $\beta$-steps to reach the $\beta$-normal form). 
  Therefore, non-idempotent intersection types give intensional insights on programs, and seem to provide a tool to reason about complexity of programs.
  The approach is defining a measure for type derivations and showing that the measure gives (a bound to) the length of the evaluation of typed terms.
  \item \emph{Linear logic interpretation}: non-idempotent intersection types are deeply linked to linear logic ($\mathsf{LL}$) \cite{DBLP:journals/tcs/Girard87}. 
  Relational semantics \cite{DBLP:journals/apal/Girard88,DBLP:journals/apal/BucciarelliE01}\,---\,
  the category $\mathbf{Rel}$ of sets and relations endowed with the comonad $\oc$ of finite multisets\,---\,is a sort of ``canonical'' denotational model of $\mathsf{LL}$; 
  the Kleisli category $\mathbf{Rel}_\oc$ of the comonad $\oc$ is a CCC and then provides a denotational model of the ordinary (\ie~call-by-name) $\lambda$-calculus.
  Non-idempotent intersection types can be seen as a syntactic presentation of $\mathbf{Rel}_\oc$: the semantics of a term $\tm$ is the set of conclusions of all type derivations of $\tm$.
\end{enumerate}

These two facts together have a potential, fascinating consequence: 
denotational semantics may provide abstract tools for complexity analysis, that are theoretically solid, being grounded on $\mathsf{LL}$.

Starting from \cite{Carvalho07,deCarvalho18}, research on relational semantics/non-idempotent intersection types has proliferated:
various works in the literature explore their power in bounding the execution time or in characterizing normalization \cite{DBLP:journals/tcs/CarvalhoPF11,DBLP:journals/apal/BucciarelliEM12,DBLP:journals/corr/BernadetL13,DBLP:conf/ictac/KesnerV15,DBLP:journals/iandc/BenedettiR16,DBLP:journals/iandc/CarvalhoF16,DBLP:journals/mscs/PaoliniPR17,DBLP:conf/rta/KesnerV17,DBLP:journals/igpl/BucciarelliKV17,MazzaPellissierVial18}.
All these works study relational semantics/non-idempotent intersection types either in $\mathsf{LL}$ proof-nets (the graphical representation of proofs in $\mathsf{LL}$), or in some variant of ordinary (\ie~call-by-name) $\lambda$-calculus. 
In the second case, the construction of the relational model $\mathbf{Rel}_\oc$ sketched above essentially relies on Girard's call-by-name translation $(\cdot)^n$ of intuitionistic logic into $\mathsf{LL}$, which decomposes the intuitionistic arrow as $(A \Rightarrow B)^n = \oc A^n \multimap B^n$.

Ehrhard \cite{DBLP:conf/csl/Ehrhard12} showed that the relational semantics $\mathbf{Rel}$ of $\mathsf{LL}$ induces also a denotational model for the \emph{call-by-value} $\lambda$-calculus\footnotemark\ that can still be viewed as a non-idempotent intersection type system.
\footnotetext{In call-by-value evaluation $\tobv$, function's arguments are evaluated before being passed to the function, so that $\beta$-redexes can fire only when their arguments are values, \ie~abstractions or variables.
The idea is that only values can be erased or duplicated.	
Call-by-value evaluation is the most common parameter passing mechanism 
used by programming languages.}
The syntactic counterpart of this construction is Girard's (``boring'') call-by-value translation $(\cdot)^v$ of intuitionistic logic into $\mathsf{LL}$ \cite{DBLP:journals/tcs/Girard87}, which decomposes the intuitionistic arrow as $(A \Rightarrow B)^v = \oc(A^v \multimap B^v)$.
Just few works have started the study of relational semantics/non-idempotent intersection types in a call-by-value setting \cite{DBLP:conf/csl/Ehrhard12,DBLP:conf/lfcs/Diaz-CaroMP13,DBLP:conf/fossacs/CarraroG14,DBLP:conf/ppdp/EhrhardG16}, and no one investigates their bounding power on the execution time in such a framework.
Our paper aims to fill this gap and 
study the information enclosed in relational semantics/non-idempotent intersection types concerning the execution time in the call-by-value $\lambda$-calculus.

A difficulty arises immediately in the qualitative characterization of call-by-value normalization 
via the relational model.
One would expect 
that the semantics of a term $\tm$ is non-empty if and only if $\tm$ is (strongly) normalizable for (some restriction of) the call-by-value evaluation $\tobv$, but it is impossible to get this result in Plotkin's original call-by-value $\lambda$-calculus $\lambda_v$ \cite{DBLP:journals/tcs/Plotkin75}.
Indeed, the terms $\tm$ and $\tmtwo$ below are $\betav$-normal but their semantics in the relational model are empty:
\begin{align}\label{eq:premature}
  \tm &\defeq (\lambda y.\Delta)(zI) \Delta & \tmtwo &\defeq \Delta ((\lambda y.\Delta)(zI)) & \text{(where }\Delta &\defeq \lambda x. xx\text{ and } I \defeq \la{\var}{\var} \text{)}
\end{align}
Actually, $\tm$ and $\tmtwo$ should behave like the famous divergent term $\Delta\Delta$, since in $\lambda_v$ they are observationally equivalent to $\Delta\Delta$ with respect all closing contexts and have the same semantics as $\Delta\Delta$ in all non-trivial denotational models of Plotkin's $\lambda_v$. 

The reason of this mismatching is that in $\lambda_v$ there are \emph{stuck $\beta$-redexes} such as $(\lambda y.\Delta)(zI)$ in Eq.~\refeq{premature}, \ie~$\beta$-redexes that $\betav$-reduction will never fire because their argument is normal but not a value (nor will it ever become one).
The real problem with stuck $\beta$-redexes is that they may prevent the creation of other $\beta_v$-redexes, providing \emph{``premature''} \emph{$\beta_v$-normal forms} like $\tm$ and $\tmtwo$ in Eq.~\refeq{premature}.
The issue 
affects termination and thus can impact on the study of 
observational equivalence and other operational properties in $\lambda_v$.

In a call-by-value setting, the issue of stuck $\beta$-redexes and then of premature $\beta_v$-normal forms arises only with \emph{open terms} (in particular, when the reduction under abstractions is allowed, since it forces to deal with ``locally open'' terms). 
Even if to model functional programming languages with a call-by-value parameter passing, such as OCaml, it is usually enough to just consider closed terms and weak evaluation (\ie~not reducing under abstractions: function bodies are evaluated only when all parameters are supplied), the importance to consider open terms in a call-by-value setting can be found, for example, in partial evaluation (which evaluates a function when not all parameters are supplied, see \cite{Jones:1993:PEA:153676}),
in the theory of proof assistants such as Coq (in particular, for type checking in a system based on dependent types, see \cite{DBLP:conf/icfp/GregoireL02}), 
or to reason about (denotational or operational) equivalences of terms in $\lambda_v$ that are congruences, or about other theoretical properties of $\lambda_v$ such as separability or solvability \cite{DBLP:conf/ictcs/Paolini01,parametricBook,DBLP:conf/flops/AccattoliP12,DBLP:conf/fossacs/CarraroG14}.

To overcome the issue of stuck $\beta$-redexes, we study relational semantics/non-idempotent intersection types in the \emph{shuffling calculus} $\shufcalc$, a conservative extension of Plotkin's $\lambda_{v}$ proposed in \cite{DBLP:conf/fossacs/CarraroG14} and further studied in \cite{Guerrieri15,GuerrieriPR15,DBLP:conf/aplas/AccattoliG16,DBLP:journals/lmcs/GuerrieriPR17}. 
It keeps the same term syntax as $\lambda_{v}$ and adds to $\betav$-reduction two commutation rules, $\sigma_{1}$ and $\sigma_{3}$, which ``shuffle'' constructors in order to move stuck $\beta$-redexes: 
they unblock $\beta_v$-redexes that are hidden by the ``hyper-sequential structure'' of terms. 
These commutation rules (referred also as \emph{$\sigma$-reduction rules}) are similar to Regnier's $\sigma$-rules for the call-by-name $\lambda$-calculus \cite{Reg:Thesis:92,Regnier94} and are inspired by the aforementioned 
$(\cdot)^v$ translation of the $\lambda$-calculus into $\mathsf{LL}$ proof-nets. 

Following the same approach used in \cite{deCarvalho18} for the call-by-name $\lambda$-calculus and in \cite{DBLP:journals/tcs/CarvalhoPF11} for $\mathsf{LL}$ proof-nets, we prove that in the shuffling calculus $\shufcalc$:
\begin{enumerate}
  \item (\emph{qualitative result})  relational semantics is adequate for $\shufcalc$, \ie~a possibly open term is normalizable for weak reduction (
  not reducing under $\lambda$'s) if and only if its interpretation in  relational semantics is not empty (\refthm{characterize-normalizable});
  this result was already proven in \cite{DBLP:conf/fossacs/CarraroG14} using different techniques;
  \item (\emph{quantiative result}) the size of type derivations can be used to measure the execution time, \ie~the number of $\betav$-steps (and not $\sigma$-steps) to reach the normal form of the weak reduction (\refprop{number-steps}). 
\end{enumerate}

Finally, we show that, differently from the case of $\mathsf{LL}$ and call-by-name $\lambda$-calculus, we are \emph{not} able to lift the quantitative information enclosed in type derivations 
to types (\ie~to the interpretation of terms) following the same technique used in \cite{deCarvalho18,DBLP:journals/tcs/CarvalhoPF11}, as our \refex{counterexample} shows.
In order to get a genuine semantic measure of execution time in a call-by-value setting, we conjecture that a refinement of its syntax and operational semantics is needed.

Even if our main goal has not yet been achieved, this investigation led to new interesting results:
\begin{enumerate}
  \item all normalizing weak reduction sequences (if any) in $\shufcalc$ from a given term have the same number of $\betav$-steps (\refcoro{same-number}); this is not obvious, as we shall explain in \refex{critical-pair-sigma};
  \item terms whose weak reduction in $\shufcalc$ ends in a value has an elegant semantic characterization (\refprop{semantic-valuable}), and the number of $\betav$-steps needed to reach their normal form can be computed in a simple way from a specific type derivation (\refthm{number-steps-value}). 
  \item all our qualitative and quantitative results for $\shufcalc$ are still valid in Plotkin's $\lambda_v$ restricted to closed terms (which models functional programming languages), see \refthm{characterize-normalizable-Plotkin}, \refcoro{same-number-Plotkin} and \refthm{number-steps-value-Plotkin}.
\end{enumerate}

Omitted proofs are in Appendix~\ref{sect:proofs} together with a list of notations and terminology used here.

\section{The shuffling calculus}
\label{sect:calculus}


\begin{figure}[!htb]
  \centering
  \scalebox{0.85}{\parbox{1.0\linewidth}{
%
%
%
%
  \begin{align*}
    \text{\emph{terms}:}		&& \tm, \tmtwo, \tmthree &\Coloneqq \, \val  \,\mid\,  \tm\tmtwo			&&(\textup{set: } \Lambda)  \\
    \text{\emph{values}:}		&& \val &\Coloneqq \, \var  \,\mid\, \la{\var}\tm 					&&(\textup{set: } \valSet)  \\
    \text{\emph{contexts}:}		&& \ctx	&\Coloneqq \, \ctxhole \mid \la{\var}{\ctx} \mid \ctx\tm \mid \tm\ctx 		&&(\textup{set: }\Lambda_\ctx) \\
    \text{\emph{\Balanced\ contexts}:}	&& \mctx&\Coloneqq \, \ctxhole \mid (\la{\var}{\mctx})\tm \mid \mctx\tm \mid \tm\mctx 	&&(\textup{set: }\Lambda_\mctx)	
    \\[-2\baselineskip]
  \end{align*}
  \begin{align*}
    \text{\emph{Root-steps}:}	&& (\la{\var}{\tm}){\val} &\rtobv \tm\isub{\var}{\val}  \qquad (\l\var.\tm)\tmtwo\tmthree \rtosl (\l\var.\tm \tmthree)\tmtwo,  \ \var \!\notin\! \Fv{\tmthree} \qquad 	    \val ((\l\var.\tmthree)\tmtwo) \rtosr (\l\var.\val \tmthree)\tmtwo,  \ \var \!\notin\! \Fv{\val} \\
    && &\rtos \, \defeq \ \rtosl \!\cup \rtosr \qquad\quad\, \rtoshuf \,\defeq \, \rtobv \!\cup \rtos
    \\[.2\baselineskip]
    \Rule\text{-\emph{reduction}:} &&   \tm \Rew{\Rule}  \tmtwo  &\iff \exists \, \ctx \in \Lambda_\ctx, \, \exists\, \tm'\!, \tmtwo' \!\in \Lambda : \tm = \ctxp{\tm'}, \, \tmtwo = \ctxp{\tmtwo'}, \, \tm' \!\rootRew{\Rule} \tmtwo' \\
    \bilancia{\Rule}\text{-\emph{reduction}:} &&   \tm \Rew{\bilancia{\Rule}}  \tmtwo  &\iff \exists \, \mctx \in \Lambda_\mctx, \, \exists\, \tm'\!, \tmtwo' \!\in \Lambda : \tm = \mctxp{\tm'}, \, \tmtwo = \mctxp{\tmtwo'}, \, \tm' \!\rootRew{\Rule} \tmtwo'
  \end{align*}
%
%
%
%
  }}
  \caption{\label{fig:shuffling-calculus} The shuffling $\l$-calculus $\shufcalc$}
\end{figure}

In this section we introduce the \emph{shuffling calculus} $\shufcalc$, namely the call-by-value $\lambda$-calculus defined in \cite{DBLP:conf/fossacs/CarraroG14} and further studied in \cite{Guerrieri15,GuerrieriPR15,DBLP:conf/aplas/AccattoliG16,DBLP:journals/lmcs/GuerrieriPR17}: 
it adds two commutation rules\,---\,the $\sigma_1$- and $\sigma_3$-reductions\,---\,to 
Plotkin's pure (i.e.~without constants) call-by-value $\lambda$-calculus $\lambda_v$ 
\cite{DBLP:journals/tcs/Plotkin75}.
The syntax for terms of $\shufcalc$ 
is the same as Plotkin's 
$\lambda_v$ and then the same as the ordinary (i.e. call-by-name) $\lambda$-calculus
, see \reffig{shuffling-calculus}.


Clearly, $\valSet \subsetneq \Lambda$.
All terms are considered up to $\alpha$-conversion (\ie~renaming of bound variables
).
%
The set of free variables of a term $\tm$ 
is denoted by $\Fv{\tm}$: $\tm$ is \emph{open} if $\Fv{\tm} \neq \emptyset$, \emph{closed} otherwise.
Given $\val \in \valSet$, $\tm \isub{\var}{\val}$ denotes the term obtained by the \emph{capture-avoiding substitution} of $\val$ for each free occurrence of $\var$ in the term 
$\tm$. 
Note that if $\val, \valtwo \in \Lambda_v$ then $\val\{\valtwo/\var\} \in \valSet$ (values are closed under substitution).


One-hole contexts $\ctx$ are defined as usual, see \reffig{shuffling-calculus}.   
We use $\ctxp{\tm}$ for the term obtained by the capture-allowing substitution of the term $\tm$ for the hole $\ctxhole$ in the context $\ctx$.
In \reffig{shuffling-calculus} we define also a special kind of contexts, \emph{balanced contexts} $\mctx$.

Reductions in the shuffling calculus are defined in \reffig{shuffling-calculus} as follows: 
given a \emph{root-step} rule $\rootRew{\Rule} \, \subseteq \Lambda \times \Lambda$
, we define the \emph{$\Rule$-reduction} $\Rew{\Rule}$ (resp.~\emph{$\bilancia{\Rule}$-reduction} $\Rew{\bilancia{\Rule}}$) as the closure of $\rootRew{\Rule}$ under contexts (resp.~balanced contexts).
The $\bilancia{\Rule}$-reduction is non-deterministic and\,---\,because of balanced contexts\,---\,can reduce under abstractions, but it is ``morally'' \emph{weak}: it reduces under a $\lambda$ only when the $\lambda$ is applied to an argument.
Clearly, $\toshufm \,\subsetneq\, \toshuf$ since $\toshuf$ can freely reduce under $\lambda$'s.

The root-steps used in the shuffling calculus are $\rootRew{\betav}$ (the reduction rule in Plotkin's $\lambda_v$), the commutation rules $\rootRew{\sigma_1}$ and $\rootRew{\sigma_3}$, and $\rootRew{\sigma} \, \defeq \ \rootRew{\sigma_1} \!\cup \rootRew{\sigma_3}$ and $\rootRew{\shuf} \,\defeq \ \rootRew{\beta_v} \!\cup \rootRew{\sigma}$.
The side conditions for $\rootRew{\sigma_1}$ and $ \rootRew{\sigma_3}$ in \reffig{shuffling-calculus} can be always fulfilled by $\alpha$-renaming. 
For any $\Rule \in \{\beta_v, \sigma_1, \sigma_3, \sigma, \shuf\}$, if $\tm \rootRew{\Rule} \tmp$ then $\tm$ is a \emph{$\Rule$-redex} and $\tmp$ is its \emph{$\Rule$-contractum}. 
A term of the shape $(\la{\var}{\tm})\tmtwo$ 
is a \emph{$\beta$-redex}.
Clearly, any $\betav$-redex is a $\beta$-redex but the converse does not hold: $(\la{\var}{\varthree}) (\vartwo I)$ is a $\beta$-redex but not a $\betav$-redex. 
Redexes of different kind may \emph{overlap}: 
for instance, the term $\Delta I \Delta$ is a $\sigma_1$-redex and contains the $\beta_v$-redex $\Delta I$; the term $\Delta (I \Delta) (xI)$ is a $\sigma_1$-redex and contains the $\sigma_3$-redex $\Delta(I\Delta)$, which contains in turn the $\beta_v$-redex $I \Delta$.


From definitions in \reffig{shuffling-calculus} it follows that $\Rew{\shuf} \, = \, \Rew{\betav} \!\cup \Rew{\sigma}$ and $\Rew{\sigma} \, = \, \Rew{\sigma_1} \!\cup \Rew{\sigma_3}$, as well as $\toshufm \, = \, \tobvm \!\cup \tosigm$ and $\tosigm \, = \, \tosl \!\cup \tosr$.
The \emph{shuffling} (resp.~\emph{balanced shuffling}) \emph{calculus} \emph{$\shufcalc$} (resp.~$\shufcalcm$) is the set $\Lambda$ of terms endowed with the reduction $\to_{\shuf}$ (resp.~$\toshufm$). 
The set $\Lambda$ endowed with the reduction $\to_{\beta_v}$ is Plotkin's pure call-by-value $\lambda$-calculus $\lambda_v$ \cite{DBLP:journals/tcs/Plotkin75}, a sub-calculus of $\shufcalc$.

\begin{proposition}[Basic properties of 
reductions, \cite{DBLP:journals/tcs/Plotkin75,DBLP:conf/fossacs/CarraroG14}]\label{prop:general-properties}
  The $\sigma$- and $\sigm$-reductions are confluent and strongly normalizing. 
  The $\betav$-, $\betavm$-, $\shuf$- and $\shufm$-reductions are confluent.
\end{proposition}
\begin{example}\label{ex:reductions}
  Recall the terms $\tm$ and $\tmtwo$ in Eq.~\refeq{premature}: 
  $\tm =\! (\lambda y.\Delta)(xI) \Delta \!\tosl\! (\lambda y.\Delta\Delta) (x I) \!\tobvm\! (\lambda y.\Delta\Delta) (x I) \allowbreak\tobvm \!\dots$ and $\tmtwo = \Delta((\lambda y.\Delta)(x I)) \!\tosr\! (\lambda y.\Delta\Delta) (x I) \!\tobvm\! (\lambda y.\Delta\Delta) (x I) \!\tobvm \!\dots$ are the only possible $\shuf$-reduction paths from $\tm$ and $\tmtwo$ respectively: 
  $\tm$ and $\tmtwo$ are not $\shuf$-normalizable and $\tm \shufeq \tmtwo$.  
  But $\tm$ and $\tmtwo$ are $\betav$-normal ($(\lambda y.\Delta) (x I)$ is a stuck $\beta$-redex) 
  and different, hence $\tm \not\simeq_{\betav} \tmtwo$ by confluence of $\to_{\betav}$ (\refprop{general-properties}).
\end{example}

Example~\ref{ex:reductions} shows how $\sigma$-reduction shuffles constructors and moves stuck $\beta$-redex in order to unblock $\beta_v$-redexes which are hidden by the ``hyper-sequential structure'' of terms, avoiding ``premature'' normal forms.
An alternative approach to circumvent the issue of stuck $\beta$-redexes is given by $\lambda_\mathsf{vsub}$, the call-by-value $\lambda$-calculus with explicit substitutions introduced in \cite{DBLP:conf/flops/AccattoliP12}, 
where hidden $\beta_v$-redexes are reduced using rules acting at a distance. 
In \cite{DBLP:conf/aplas/AccattoliG16} it has been shown that $\lambda_\mathsf{vsub}$ and $\shufcalc$ can be 
embedded in each other preserving termination and divergence.
Interestingly, both calculi are inspired by an analysis of Girard's ``boring'' call-by-value translation of $\lambda$-terms into linear logic proof-nets \cite{DBLP:journals/tcs/Girard87,DBLP:journals/tcs/Accattoli15} according to the linear recursive type $o = \oc o \multimap \oc o$, or equivalently $o = \oc (o \multimap o)$.
In this translation, $\shuf$-reduction corresponds to cut-elimination, more precisely $\betav$-steps (resp.~$\sigma$-steps) correspond to exponential (resp.~multiplicative) cut-elimination steps;
$\shufm$-reduction corresponds to cut-elimination at depth~$0$.

Consider the two subsets of terms defined by mutual induction (notice that $\anfSet \subsetneq \wnfSet \supsetneq \valSet$):
\begin{align*}
 \anf & \grameq \var\val \mid \var\anf    \mid \anf\wnf		\quad \textup{(set: } \anfSet\textup{)} &
 \wnf & \grameq \val \mid \anf \mid (\la{\var}{\wnf})\anf	\quad \textup{(set: } \wnfSet\textup{)}.
\end{align*}
Any $\tm \in \anfSet$ is neither a value nor a $\beta$-redex, but an open applicative term 
with a free ``head variable''.

\newcounter{prop:syntactic-normal}
\addtocounter{prop:syntactic-normal}{\value{definition}}
\begin{proposition}[Syntactic characterization on $\shufm$-normal forms]
\label{prop:syntactic-normal}
  Let $\tm$ be a term:
\NoteProof{propAppendix:syntactic-normal}
  \begin{itemize}
    \item $\tm$ is $\shufm$-normal iff $\tm \in \wnfSet$;
    \item $\tm$ is $\shufm$-normal and is neither a value nor a $\beta$-redex iff $\tm \in \anfSet$.
  \end{itemize}
\end{proposition}

Stuck $\beta$-redexes correspond to $\shufm$-normal forms of the shape $(\la{\var}{\wnf})\anf$.
As a consequence of \refprop{syntactic-normal}, the behaviour of \emph{closed} terms with respect to $\shufm$-reduction (resp.~$\betavm$-reduction) is quite simple: either they diverge or they $\shufm$-normalize (resp.~$\betavm$-normalize) to a closed value. Indeed:

\newcounter{coro:syntactic-normal-closed}
\addtocounter{coro:syntactic-normal-closed}{\value{definition}}
\begin{corollary}[Syntactic characterization of closed $\shufm$- and $\betavm$-normal forms]
	\label{coro:syntactic-normal-closed}
	\NoteProof{coroAppendix:syntactic-normal-closed}
	Let $\tm$ be a closed term:
	 $\tm$ is $\shufm$-normal iff $\tm$ is $\betavm$-normal iff $\tm$ is a value iff $\tm = \la{\var}{\tmtwo}$ for some term $\tmtwo$ with $\Fv{\tmtwo} \subseteq\{\var\}$.
\end{corollary}


\section{A non-idempotent intersection type system}
\label{sect:type}

We aim to define a non-idempotent intersection type system in order to characterize the (strong) normalizable terms for the reduction $\toshuf$.
\emph{Types} are \emph{positive} or \emph{negative}, defined by mutual induction: 
\begin{align*}
  &\mbox{Negative Types:} & M,N &\grameq \Pair{P}{Q} &\qquad&&
  &\mbox{Positive Types:} & P,Q &\grameq [N_1, \dots, N_n] \ \mbox{ (with $n \in \nat$)}
\end{align*}
where $[N_1, \dots, N_n]$ is a (possibly empty) finite multiset of negative types;
in particular the \emph{empty multiset} $\emptymset$ (obtained for $n = 0$) is the only atomic (positive) type.
A positive type $[N_1, \dots, N_n]$ has to be intended as a conjunction $N_1 \land \dots \land N_n$ of negative types $N_1, \dots, N_n$, for a commutative and associative conjunction connective $\land$ that is not idempotent and whose neutral element is $\emptymset$.

The derivation rules for the non-idempotent intersection type system are in \reffig{types}.
In this typing system, \emph{judgments} have the shape $\Gamma \vdash \tm : P$ where $\tm$ is a term, $P$ is a positive type and $\Gamma$ is an \emph{environment} (\ie~a total function from variables to positive types whose domain $\Dom{\Gamma} = \{\var \mid \Gamma(\var) \neq \emptymset\}$ is finite
). 
The \emph{sum of environments} $\Gamma \uplus \Delta$ is defined pointwise via multiset sum: $(\Gamma \uplus \Delta)(\var) = \Gamma(\var) \uplus \Delta(\var)$. 
An environment $\Gamma$ such that $\Dom{\Gamma} \subseteq \{\var_1, \dots, \var_n\}$ with $\var_i \neq \var_j$ and $\Gamma(\var_i) = P_i$ for all $1 \leq i \neq j \leq k$ is often written as $\Gamma = \var_1 \colon\! P_1,\dots, \var_n \colon\! P_k$.
In particular, $\Gamma$ and $\Gamma, \var \colon\! \emptymset$ (where $\var \notin \Dom{\Gamma}$) are the same environment; 
and $\vdash \tm \colon\! P$ stands for the judgment $\Gamma \vdash \tm \colon\! P$ where $\Gamma$ is the \emph{empty environment}, \ie~$\Dom{\Gamma} = \emptyset$ (that is, $\Gamma(\var) = \emptymset$ for any variable $\var$).
Note that the sum of environments $\uplus$ is commutative, associative and its neutral element is the empty environment: given an environment $\Gamma$, one has $\Gamma \uplus \Delta = \Gamma$ iff $\Dom{\Delta} = \emptyset$. 
The notation $\concl{\pi}{\Gamma}{\tm}{P}$ means that $\pi$ is a derivation with conclusion the judgment $\Gamma \vdash \tm \colon P$.
We write $\Type{\pi}{\tm}$ if $\pi$ is such that $\concl{\pi}{\Gamma}{\tm}{P}$ for some environment $\Gamma$ and positive type $P$.

\begin{figure}[!t]
  \centering
      \AxiomC{}
      \RightLabel{\footnotesize$\Ax$}
      \UnaryInfC{$\var \!:\! P \vdash \var \!:\! P$}
      \DisplayProof
      \ \ 
      \AxiomC{$\Gamma \vdash \tm \!:\! [\Pair{P}{Q}]$}
      \AxiomC{$\Gamma' \vdash \tmtwo \!:\! P$}
      \RightLabel{\footnotesize$@$}
      \BinaryInfC{$\Gamma \uplus \Gamma' \vdash \tm\tmtwo \!:\! Q$}
      \DisplayProof
      \ \ 
      \AxiomC{$\Gamma_1, \var \!:\! P_1 \vdash \tm \!:\! Q_1$}
      \AxiomC{$\overset{n \in \nat}{\dots}$}
      \AxiomC{$\Gamma_n, \var \!:\! P_n \vdash \tm \!:\! Q_n$}
      \RightLabel{\footnotesize$\mathsf{\lambda}$}
      \TrinaryInfC{$\Gamma_1 \uplus \dots \uplus \Gamma_n \vdash \la{\var}{\tm} \!:\! [\Pair{P_1}{Q_1}, \dots, \Pair{P_n}{Q_n}]$}
      \DisplayProof
  \caption{Non-idempotent intersection type system for the shuffling calculus.
  }
  \label{fig:types}
\end{figure}

It is worth noticing that the type system in \reffig{types} is \emph{syntax oriented}: 
for each type judgment $J$ there is a \emph{unique} derivation rule whose conclusion matches the judgment $J$.

The \emph{size} $\size{\pi}$ of a type derivation $\pi$ is just the the number of $@$ rules in $\pi$. 
Note that judgments play no role in in the size of a derivation.

\begin{example}
\label{ex:size}
  Let $I = \la{\var}{\var}$. The derivations (typing $II$ and $I$ with same type and same environment)
  \begin{align*}
    \pi_{II} = 
    \AxiomC{}
    \RightLabel{\footnotesize{$\mathsf{ax}$}}
    \UnaryInfC{$\var \colon\! [\,] \vdash \var \colon\! [\,]$}
    \RightLabel{\footnotesize{$\lambda$}}
    \UnaryInfC{$\vdash I \colon\! [\Pair{[\,]}{[\,]}]$}
    \AxiomC{}
    \RightLabel{\footnotesize{$\lambda$}}
    \UnaryInfC{$\vdash I \colon\! [\,]$}
    \RightLabel{\footnotesize{$@$}}
    \BinaryInfC{$\vdash II \colon\! [\,]$}
    \DisplayProof    
    &&
    \pi_I =
    \AxiomC{}
    \RightLabel{\footnotesize{$\lambda$}}
    \UnaryInfC{$ \vdash I \colon\! [\,]$}
    \DisplayProof
  \end{align*}
  are such that $\size{\pi_{II}} = 1$ and $\size{\pi_I} = 0$.
  Note that $II \toshufm I$ and $\size{\pi_{II}} = \size{\pi_{I}} + 1$.
\end{example}

The following lemma (whose proof is quite technical) will play a crucial role 
to prove the substitution lemma (\reflemma{substitution}) and the subject reduction (\refprop{quant-subject-reduction}) and expansion (\refprop{quant-subject-expansion}).
\newcounter{l:value}
\addtocounter{l:value}{\value{definition}}
\begin{lemma}[Judgment decomposition for values]
\label{l:value}
  Let 
\NoteProof{lappendix:value}
  $\val \in \Lambda_v$, $\Delta$ be an environment, and $P_1, \dots, P_p$ be positive types (for some $p \in \nat$).
  There is a derivation $\concl{\pi}{\Delta}{\val}{P_1 \uplus \dots \uplus P_p}$ iff for all $1 \leq i \leq p$ there are an environment $\Delta_i$ and a derivation $\concl{\pi_i}{\Delta_i}{\val}{P_i}$ such that $\Delta = \biguplus_{i=1}^p \Delta_i$.
  Moreover, $\size{\pi} = \sum_{i=1}^p\size{\pi_i}$.
\end{lemma}

The left-to-right direction of \reflemma{value} means that, given 
$\concl{\pi}{\Delta}{\val}{P}$, for \emph{every} $p \in \nat$ and \emph{every} decomposition of the positive type $P$ into a multiset sum of positive types $P_1, \dots, P_p$
, there are environments $\Delta_1, \dots, \Delta_p$ such that $\Delta_i \vdash \val \colon\! P_i$ is derivable for all $1 \leq i \leq p$.

\newcounter{l:substitution}
\addtocounter{l:substitution}{\value{definition}}
\begin{lemma}[Substitution]
\label{l:substitution}
  Let 
\NoteProof{lappendix:substitution}
  $\tm \in \Lambda$ and $\val \in \Lambda_\val$.
  If $\concl{\pi}{\Gamma, \var \colon\! P}{\tm}{Q}$ and $\concl{\pi'}{\Delta}{\val}{P}$, then there exists $\concl{\pi''}{\Gamma \uplus \Delta}{\tm\isub\var\val}{Q}$ such that $\size{\pi''} = \size{\pi} + \size{\pi'}$.
\end{lemma}

We can now prove the subject reduction, with a quantitative flavour about the size of type derivations in order to extract information about the execution time.

\newcounter{prop:quant-subject-reduction}
\addtocounter{prop:quant-subject-reduction}{\value{definition}}
\begin{proposition}[Quantitative balanced subject reduction]
\label{prop:quant-subject-reduction}
  Let 
\NoteProof{propappendix:quant-subject-reduction}
  $\tm, \tmp \in \Lambda$ and $\concl{\pi}{\Gamma}{\tm}{Q}$.
  \begin{enumerate}
    \item\label{p:quant-subject-reduction-betav} \emph{Shrinkage under $\betavm$-step:} If $\tm \tobvm \tmp$ then $\size{\pi} > 0$ and there exists a derivation $\pi'$ with conclusion $\Gamma \vdash \tmp \colon\! Q$ such that $\size{\pi'} = \size{\pi} - 1$.
    \item\label{p:quant-subject-reduction-sigma} \emph{Size invariance under $\sigm$-step:} If $\tm \tosigm \tmp$ then $\size{\pi} > 0$ and there exists a derivation $\pi'$ with conclusion $\Gamma \vdash \tmp \colon\! Q$ such that $\size{\pi'} = \size{\pi}$.
  \end{enumerate}
\end{proposition}

In \refprop{quant-subject-reduction}, the fact that $\toshuf$ does not reduce under $\lambda$'s is crucial to get the quantitative information, otherwise one can have a term $\tm$ such that every derivation $\concl{\pi}{\Gamma}{\tm}{P}$ is such that $\size{\pi} = 0$ (and then there is no derivation $\pi'$ with conclusion $\Gamma \vdash \tmp \colon P$ such that $\size{\pi} = \size{\pi'} -1$):
this is the case, for example, for $\tm = \la\var\delta\delta \tobvm \tm$.

In order to prove the quantitative subject expansion (\refprop{quant-subject-expansion}), we first need the following technical lemma stating the commutation of abstraction with
abstraction and application.

\newcounter{l:commutation}
\addtocounter{l:commutation}{\value{definition}}
\begin{lemma}[Abstraction commutation]\hfill
\label{l:commutation}
\NoteProof{lappendix:commutation}
  \begin{enumerate}
    \item\label{p:commutation-abstraction} \emph{Abstraction vs.~abstraction:}
    Let $k \in \nat$.
    If $\concl{\pi}{\Delta}{\la{\vartwo}{(\la{\var}{\tm})\val}}{\biguplus_{i=1}^k [\Pair{P_i'}{P_i}]}$ and $\vartwo \notin \Fv{\val}$, then there is $\concl{\pi'}{\Delta}{(\la{\var}{\la{\vartwo}{\tm}})\val}{\biguplus_{i=1}^k [\Pair{P_i'}{P_i}]}$ such that $\size{\pi'} = \size{\pi} + 1 - k$.

    \item \label{p:commutation-application}\emph{Application vs.~abstraction:}
    If $\concl{\pi}{\Delta}{((\la{\var}{\tm})\val)((\la{\var}{\tmtwo})\val)}{P}$ then there exists a derivation $\concl{\pi'}{\Delta}{(\la{\var}{\tm\tmtwo})\val}{P}$ such that $\size{\pi'} = \size{\pi} - 1$.
  \end{enumerate}
\end{lemma}

\newcounter{prop:quant-subject-expansion}
\addtocounter{prop:quant-subject-expansion}{\value{definition}}
\begin{proposition}[Quantitative balanced subject expansion]
\label{prop:quant-subject-expansion}
\NoteProof{propappendix:quant-subject-expansion}
  Let $\tm, \tmp \in \Lambda$ and $\concl{\pi'}{\Gamma}{\tmp}{Q}$.
  \begin{enumerate}
    \item\label{p:quant-subject-expansion-betav}\emph{Enlargement under anti-$\betavm$-step:} If $\tm \tobvm \tmp$ then there is $\concl{\pi}{\Gamma}{\tm}{Q}$ with $\size{\pi} = \size{\pi'} + 1$.
    \item\label{p:quant-subject-expansion-sigma}\emph{Size invariance under anti-$\sigm$-step:} If $\tm \tosigm \tmp$ then $\size{\pi'} > 0$ and there is $\concl{\pi}{\Gamma}{\tm}{Q}$ with $\size{\pi} = \size{\pi'}$.
  \end{enumerate}
\end{proposition}

Actually, subject reduction and expansion hold for the whole $\shuf$-reduction $\toshuf$, not only for the balanced $\shuf$-reduction $\toshufm$.
The drawback for $\toshuf$ is that the quantitative information about the size of the derivation is lost in the case of a $\betav$-step.

\newcounter{l:subject-reduction}
\addtocounter{l:subject-reduction}{\value{definition}}
\begin{lemma}[Subject reduction]
\label{l:subject-reduction}
\NoteProof{lappendix:subject-reduction}
  Let $\tm, \tmp \in \Lambda$ and $\concl{\pi}{\Gamma}{\tm}{Q}$.
  \begin{enumerate}
    \item\label{p:subject-reduction-betav}\emph{Shrinkage under $\betav$-step:} If $\tm \tobv \tmp$ then there is $\concl{\pi'}{\Gamma}{\tmp}{Q}$ with $\size{\pi} \geq \size{\pi'}$.
    \item\label{p:subject-reduction-sigma}\emph{Size invariance under $\sigma$-step:} If $\tm \tosig \tmp$ then there is $\concl{\pi'}{\Gamma}{\tmp}{Q}$ such that $\size{\pi} = \size{\pi'}$.
  \end{enumerate}
\end{lemma}

\newcounter{l:subject-expansion}
\addtocounter{l:subject-expansion}{\value{definition}}
\begin{lemma}[Subject expansion]
\label{l:subject-expansion}
\NoteProof{lappendix:subject-expansion}
  Let $\tm, \tmp \in \Lambda$ and $\concl{\pi'}{\Gamma}{\tmp}{Q}$.
  \begin{enumerate}
    \item\label{p:subject-expansion-betav}\emph{Enlargement under anti-$\betav$-step:} If $\tm \tobv \tmp$ then there is $\concl{\pi}{\Gamma}{\tm}{Q}$ with $\size{\pi} \geq \size{\pi'}$.
    \item\label{p:subject-expansion-sigma}\emph{Size invariance under anti-$\sigma$-step:} If $\tm \tosig \tmp$ then there is $\concl{\pi}{\Gamma}{\tm}{Q}$ such that $\size{\pi} = \size{\pi'}$.
  \end{enumerate}
\end{lemma}

In \reflemmasp{subject-reduction}{betav}{subject-expansion}{betav} it is impossible to estimate more precisely the relationship between $\size{\pi}$ and $\size{\pi'}$.
Indeed, \refex{size} has shown that there are $\concl{\pi_I}{\vartwo \colon [\,]}{I}{[\,]}$ and $\concl{\pi_{II}}{\vartwo \colon\! [\,]}{II}{[\,]}$ such that $\size{\pi_I} = 0$ and $\size{\pi_{II}} = 1$ (where $I = \la{\var}{\var}$).
So, given $k \in \nat$, consider the derivations $\concl{\pi_k}{\,}{\la{\vartwo}{II}}{[\Pair{[\,]}{[\,]}, \,\overset{k}{\dots}\,, \Pair{[\,]}{[\,]}]}$ and $\concl{\pi_k'}{\,}{\la{\vartwo}{I}}{[\Pair{[\,]}{[\,]}, \,\overset{k}{\dots}\,, \Pair{[\,]}{[\,]}]}$ below:
\begin{align*}
  \pi_n =
  \AxiomC{$\ \vdots\, \pi_{II}$}
  \noLine
  \UnaryInfC{$\vartwo \colon\! [\,] \vdash II \colon\! [\,]$}
  \AxiomC{$\overset{k}{\ldots}$}
  \AxiomC{$\ \vdots\, \pi_{II}$}
  \noLine
  \UnaryInfC{$\vartwo \colon\! [\,] \vdash II \colon\! [\,]$}
  \RightLabel{\footnotesize$\lambda$}
  \TrinaryInfC{$\vdash \la{\vartwo}{II} \colon\! [\Pair{[\,]}{[\,]}, \,\overset{k}{\dots}\,, \Pair{[\,]}{[\,]}]$}
  \DisplayProof
  &&
  \pi_n' =
  \AxiomC{$\ \vdots\, \pi_{I}$}
  \noLine
  \UnaryInfC{$\vartwo \colon\! [\,] \vdash I \colon\! [\,]$}
  \AxiomC{$\overset{k}{\ldots}$}
  \AxiomC{$\ \vdots\, \pi_{I}$}
  \noLine
  \UnaryInfC{$\vartwo \colon\! [\,] \vdash I \colon\! [\,]$}
  \RightLabel{\footnotesize$\lambda$}
  \TrinaryInfC{$\vdash \la{\vartwo}{I} \colon\! \big[ \Pair{[\,]}{[\,]}, \,\overset{k}{\dots}\,, \Pair{[\,]}{[\,]} \big]$}
  \DisplayProof
\end{align*}
Clearly, $\la{\vartwo}{II} \toshuf \la{\vartwo}{I}$ (but $\la{\vartwo}{II} \not\toshufm \la{\vartwo}{I}$) and the $\pi_k'$ (resp.~$\pi_k$) is the only derivation typing $\la{\vartwo}{I}$ (resp.~$\la{\vartwo}{II}$) with the same type and environment as $\pi_k$ (resp.~$\pi_k'$).
One has $\size{\pi_k} = k \cdot \size{\pi_{II}} = k$ and $\size{\pi_k'} = k \cdot \size{\pi_I} = 0$, thus 
 the difference of size of the derivations $\pi_k$ and $\pi_k'$ can be arbitrarely large (since $k \in \nat$); in particular $\size{\pi_0} = \size{\pi_0'}$, so for $k = 0$ the size of derivations does not even strictly decrease.

\section{Relational semantics: qualitative results}
\label{sect:qualitative-semantics}

\reflemmas{subject-reduction}{subject-expansion} have an important consequence: the non-idempotent intersection type system of \reffig{types} defines a denotational model for the shuffling calculus $\shufcalc$ (\refthm{invariance} below).

\begin{definition}[Suitable list of variables for a term, semantics of a term]
\label{def:semantics}
  Let $\tm \in \Lambda$ and let $\var_1, \dots, \var_k$ be pairwise distinct variables, for some $k \in \nat$.
  
  If $\Fv{\tm} \subseteq \{\var_1, \dots, \var_k\}$, then we say that the list $\vec{\var} = (\var_1, \dots, \var_k)$ is \emph{suitable for} $\tm$.
  
  If $\vec{\var} = (\var_1, \dots, \var_k)$ is suitable for $\tm$, the (\emph{relational}) \emph{semantics}, or \emph{interpretation}, \emph{of} $\tm$ \emph{for} $\vec{\var}$ is
  \begin{equation*}
    \sem{\tm}{\vec{\var}} = \{((P_1,\dots, P_k),Q) \mid \exists \, \concl{\pi}{\var_1 \colon\! P_1, \dots, \var_k \colon\! P_k}{\tm}{Q}\} \,.
  \end{equation*}
\end{definition}

Essentially, the semantics of a term $\tm$ for a suitable list $\vec{\var}$ of variables is the set of 
judgments for $\vec{\var}$ and $\tm$ that can be derived
in the non-idempotent intersection type system of \reffig{types}.

If we identify the negative type $\Pair{P}{Q}$ with the pair $(P,Q)$ and if we set $\U \defeq \bigcup_{k \in \Nat} \U_k$ where:
\begin{align*}
  \U_0 &\defeq \emptyset & &\U_{k+1} \defeq \MultiFin{\U_k} \times \MultiFin{\U_k} && \textup{($\MultiFin{X}$ is the set of finite multisets over the set $X$)}
\end{align*}
then, for any $\tm \in \Lambda$ and any suitable list $\vec{\var} = (\var_1, \dots, \var_k)$ for $\tm$, one has $\sem{\tm}{\vec{\var}} \subseteq \MultiFin{\U}^k \times \MultiFin{\U}$;
in particular, if $\tm$ is closed and $\vec{\var} = (\,)$, then $\sem{\tm}{} = \{Q \mid \exists\, \concl{\pi}{\ }{\tm}{Q}\} \subseteq \MultiFin{\U}$ (up to an obvious isomorphism).
Note that $\U = \MultiFin{\U} \times \MultiFin{\U}$:
\cite{DBLP:conf/csl/Ehrhard12,DBLP:conf/fossacs/CarraroG14} proved that the latter identity is enough to have a denotational model for $\shufcalc$. 
We can also prove it explicitly using \reflemmas{subject-reduction}{subject-expansion}.

\newcounter{thm:invariance}
\addtocounter{thm:invariance}{\value{definition}}
\begin{theorem}[Invariance under $\shuf$-equivalence]
\label{thm:invariance}
\NoteProof{thmappendix:invariance}
  Let $\tm, \tmtwo \in \Lambda$, let $k \in \nat$ and let $\vec{\var} = (\var_1, \dots, \var_k)$ be a suitable list of variables for $\tm$ and $\tmtwo$.
  If $\tm \shufeq \tmtwo$ then $\sem{\tm}{\vec{\var}} = \sem{\tmtwo}{\vec{\var}}$.
\end{theorem}


An interesting property of relational semantics is that all $\shufm$-normal forms have a non-empty interpretation (\reflemma{semantics-normal}).
To prove that we use the syntactic characterization of $\shufm$-normal forms (\refprop{syntactic-normal}).
Note that a stronger statement (\reflemmap{semantics-normal}{anf}) is required for $\shufm$-normal forms belonging to $\anfSet$, in order to handle the case where the $\shufm$-normal form is a $\beta$-redex.

\newcounter{l:semantics-normal}
\addtocounter{l:semantics-normal}{\value{definition}}
\begin{lemma}[Semantics and typability of $\shufm$-normal forms]
\label{l:semantics-normal}
\NoteProof{lappendix:semantics-normal}
  Let $\tm$ be a term, let $k \in \nat$ and let $\vec{\var} = (\var_1, \dots, \var_k)$ be a list of variables suitable for $\tm$.
  \begin{enumerate}
    \item\label{p:semantics-normal-anf} If $\tm \in \anfSet$ then for every positive type $Q$ there exist positive types $P_1, \dots, P_k$ and a derivation $\concl{\pi}{\var_1 \colon\! P_1, \dots, \var_k \colon\!  P_k}{\tm}{Q}$.
    \item\label{p:semantics-normal-wnf} If $\tm \in \wnfSet$ then there are positive types $Q, P_1, \dots, P_k$ and a derivation $\concl{\pi}{\var_1 \colon\! P_1, \dots, \var_k \colon\!  P_k\allowbreak}{\tm}{Q}$.
    \item\label{p:semantics-normal-nonempty} If $\tm$ is $\shufm$-normal then $\sem{\tm}{\vec{\var}} \neq \emptyset$.
  \end{enumerate}
\end{lemma}

A consequence of \refprop{quant-subject-reduction} (and \refthm{invariance} and \reflemma{semantics-normal}) is a qualitative result: a semantic and logical (if we consider our non-idempotent type system as a logical framework) characterization of (strong) $\shufm$-normalizable terms (\refthm{characterize-normalizable}).
In this theorem, the main equivalences are between Points \ref{p:characterize-normalizable-normalizable}, \ref{p:characterize-normalizable-nonempty} and \ref{p:characterize-normalizable-strongly-normalizable}
, already proven in \cite{DBLP:conf/fossacs/CarraroG14} using different techniques.
Points \ref{p:characterize-normalizable-equivalen-normal} and \ref{p:characterize-normalizable-derivable} can be seen as ``intermediate stages'' in the proof of the main equivalences, which are informative enough to deserve to be explicitely stated.

\newcounter{thm:characterize-normalizable}
\addtocounter{thm:characterize-normalizable}{\value{definition}}
\begin{theorem}[Semantic and logical characterization of $\shufm$-normalization]
\label{thm:characterize-normalizable}
\NoteProof{thmappendix:characterize-normalizable}
  Let $\tm \in \Lambda$ and let $\vec{\var} = (\var_1, \dots, \var_k)$ be a suitable list of variables for $\tm$.
  The following are equivalent:
  \begin{enumerate}
    \item\label{p:characterize-normalizable-normalizable} \emph{Normalizability:} $\tm$ is $\shufm$-normalizable;
    \item\label{p:characterize-normalizable-equivalen-normal} \emph{Completeness:} $\tm \shufeq \tmtwo$ for some $\shufm$-normal $\tmtwo \in \Lambda$;
    \item\label{p:characterize-normalizable-nonempty}\emph{Adequacy:} $\sem{\tm}{\vec{\var}} \neq \emptyset$;
    \item\label{p:characterize-normalizable-derivable}\emph{Derivability:} there is a derivation $\concl{\pi}{\var_1 \colon\! P_1, \dots, \var_k \colon\! P_k}{\tm}{Q}$ for some positive types $P_1, \dots, P_k, Q$;
    \item\label{p:characterize-normalizable-strongly-normalizable}\emph{Strong normalizabilty:} $\tm$ is strongly $\shufm$-normalizable.
  \end{enumerate}
\end{theorem}

Equivalence \eqref{p:characterize-normalizable-strongly-normalizable}$\Leftrightarrow$\eqref{p:characterize-normalizable-normalizable} means that normalization and strong normalization are equivalent for $\shufm$-reduction, thus in studying the termination of $\shufm$-reduction no intricacy arises from its non-determinism.
Equivalence \eqref{p:characterize-normalizable-normalizable}$\Leftrightarrow$\eqref{p:characterize-normalizable-equivalen-normal} says that $\shufm$-reduction is complete to get $\shufm$-normal forms;
in particular, this entails that every $\shuf$-normalizable term is $\shufm$-normalizable.
Equivalence \eqref{p:characterize-normalizable-normalizable}$\Leftrightarrow$\eqref{p:characterize-normalizable-equivalen-normal} is the analogue of a well-known theorem \cite[Thm.~8.3.11]{Barendregt84} for ordinary (\ie~call-by-name) $\lambda$-calculus relating head $\beta$-reduction and $\beta$-equivalence: this corroborates the idea that $\shufm$-reduction is the ``head reduction'' in a call-by-value setting, despite its non-determinism.
The equivalence \eqref{p:characterize-normalizable-nonempty}$\Leftrightarrow$\eqref{p:characterize-normalizable-derivable} holds by definition of relational semantics.

Implication \eqref{p:characterize-normalizable-normalizable}$\Rightarrow$\eqref{p:characterize-normalizable-nonempty} (or equivalently \eqref{p:characterize-normalizable-normalizable}$\Rightarrow$\eqref{p:characterize-normalizable-derivable}, \ie~``normalizable $\Rightarrow$ typable'') does not hold in Plotkin's $\lambda_v$: indeed, the (open) terms $\tm$ and $\tmtwo$ in Eq.~\refeq{premature} (see also \refex{reductions}) are $\betav$-normal (because of a stuck $\beta$-redex) but 
$\sem{\tm}{\var} = \emptyset = \sem{\tmtwo}{\var}$. 
Equivalences such as the ones in \refthm{characterize-normalizable} hold in a call-by-value setting provided that $\betav$-reduction is extended, \eg~by adding $\sigma$-reduction. 
In \cite{DBLP:conf/aplas/AccattoliG16}, $\shufcalc$ is proved to be termination equivalent to other extensions of $\lambda_v$ (in the framework Open Call-by-Value, where evaluation is call-by-value and weak, on possibly open terms) such as the fireball calculus \cite{parametricBook,DBLP:conf/icfp/GregoireL02,fireballs} and the value substitution calculus \cite{DBLP:conf/flops/AccattoliP12}, so \refthm{characterize-normalizable} is a general result \emph{characterizing termination in those calculi} as well.

\newcounter{l:semantic-value}
\addtocounter{l:semantic-value}{\value{definition}}
\begin{lemma}[Uniqueness of the derivation with empty types; Semantic and logical characterization of values]
  \label{l:semantic-value}
  \NoteProof{lappendix:semantic-value}
  Let $\tm \in \Lambda$ be $\shufm$-normal.
  
  \begin{enumerate}
    \item\label{p:semantic-value-uniqueness} If $\concl{\pi}{\,}{\tm}{\emptymset}$ and $\concl{\pi'}{\Gamma}{\tm}{\emptymset}$, then $\tm \in \valSet$, $\size{\pi} = 0$, $\Dom{\Gamma} = \emptyset$ and $\pi = \pi'$.
    More precisely, $\pi$ consists 
    of a rule $\Ax$ if $\tm$ is a variable, otherwise $\tm$ is an abstraction and $\pi$ consists 
    of a 0-ary~rule~$\lambda$.  

    \item\label{p:semantic-value-characterization} Given a list $\vec{\var} = (\var_1, \dots, \var_k)$ of variables suitable for $\tm$, the following are equivalent:
    \begin{multicols}{2}
    \begin{enumerate}
	    \item\label{p:semantic-value-value} $\tm$ is a value;
	    \item\label{p:semantic-value-semantic} $((\emptymset, \overset{k}{\dots\,}, \emptymset), \emptymset) \in \sem{\tm}{\vec{\var}}$\,;
	    \item\label{p:semantic-value-logic} there exists $\concl{\pi}{\,}{\tm}{\emptymset}$\,;
	    \item\label{p:semantic-value-size} there exists $\Type{\pi}{\tm}$ such that $\size{\pi} = 0$.
    \end{enumerate} 
    \end{multicols}	
  \end{enumerate}
\end{lemma}
Qualitatively, \reflemma{semantic-value} allows us to refine the semantic and logical characterization given by \refthm{characterize-normalizable} for a specific class of terms: the \emph{valuable} ones, \ie~the terms that $\shufm$-normalize to a value.
Valuable terms are all only the terms whose semantics contains a specific element: the point with only empty types.

\newcounter{prop:semantic-valuable}
\addtocounter{prop:semantic-valuable}{\value{definition}}
\begin{proposition}[Logical and semantic characterization of valuability]
	\label{prop:semantic-valuable}
\NoteProof{propappendix:semantic-valuable}
	Let $\tm$ be a 
	term and $\vec{\var} = (\var_1, \dots, \var_k)$ be a suitable list of variables for $\tm$.
	The following are equivalent:
		\begin{enumerate}
			\item\label{p:semantic-valuable-value}\emph{Valuability:} $\tm$ is $\shufm$-normalizable and the $\shufm$-normal form of $\tm$ is a value;
			\item\label{p:semantic-valuable-semantic}\emph{Empty point in the semantics:} $((\emptymset, \overset{k}{\dots}\,, \emptymset), \emptymset) \in \sem{\tm}{\vec{\var}}$;
			\item\label{p:semantic-valuable-logic}\emph{Derivability with empty types:} there exists a derivation $\concl{\pi}{\,}{\tm}{\emptymset}$.
		\end{enumerate} 
\end{proposition}

%

\section{The quantitative side of type derivations}
\label{sect:quantitative-semantics}

By the quantitative subject reduction (\refprop{quant-subject-reduction}),  
the size of \emph{any} derivation typing a ($\shufm$-normalizable) term $\tm$ is an upper bound on the number of $\betavm$-steps in \emph{any} $\shufm$-normalizing reduction sequence from $\tm$, since the size of a type derivation decreases by $1$ after each $\betavm$-step, and does not change after each $\sigm$-step.
%

\begin{corollary}[Upper bound on the number of $\betavm$-steps]
	Let $\tm$ be a $\shufm$-normalizable term and $\tm_0$ be its $\shufm$-normal form. 
	For any reduction sequence $\deriv \colon \tm \toshufm^* \tm_0$ and any $\Type{\pi}{\tm}$, $\Lengbv{\deriv} \leq \size{\pi}$.
\end{corollary}

In order to extract from a type derivation the \emph{exact} number of $\betavm$-steps to reach the $\shufm$-normal form, we have 
to take into account also the size of derivations of $\shufm$-normal forms.
Indeed, by \reflemmap{semantic-value}{characterization}, $\shufm$-normal forms that are not values admit only derivations with sizes greater than $0$.  
The sizes of type derivations of a $\shufm$-normal form $\tm$ are related to a special kind of size of $\tm$ that we now define.

%

The \emph{balanced size} of a term $\tm$, denoted by $\sizeZero{\tm}$, is defined by induction on $\tm$ as follows ($\val \in \valSet$):
\begin{align*}
\sizeZero{\val} &= 0 	& 
\sizeZero{\tm\tmtwo} &=
\begin{cases}
\sizeZero{\tmthree} + \sizeZero{\tmtwo} + 1	&\textup{if } \tm = \la{\var}{\tmthree} \\
\sizeZero{\tm} + \sizeZero{\tmtwo} + 1	&\textup{otherwise}
\end{cases}
\end{align*}
So, the balanced size of a term $\tm$ is the number of applications occurring in $\tm$ under a balanced context, \ie the number of pairs $(\tmtwo,\tmthree)$ such that $\tm = \mctxp{\tmtwo\tmthree}$ for some balanced context $\mctx$. 
For instance, $\sizeZero{(\la{\var}{\vartwo\vartwo})(\varthree\varthree)} = 3$ and $\sizeZero{(\la{\var}{\la{\var'\!}{\vartwo\vartwo}})(\varthree\varthree)} = 2$.
The following lemma can be seen as a quantitative version of \reflemma{semantics-normal}.

\newcounter{l:sizes}
\addtocounter{l:sizes}{\value{definition}}
\begin{lemma}[Relationship between sizes of normal forms and derivations]
	\label{l:sizes}
	\NoteProof{lappendix:sizes}
	Let $\tm \in \Lambda$.
	\begin{enumerate}
		\item\label{p:sizes-normal} If $\tm$ is $\shufm$-normal then $\sizeZero{\tm} = \min{\{\size{\pi} \mid \Type{\pi}{\tm}\}}$.
		\item\label{p:sizes-value} If $\tm$ is a value then $\sizeZero{\tm} = \min{\{\size{\pi} \mid \Type{\pi}{\tm}\}} = 0 
		$.
	\end{enumerate}
\end{lemma}

Thus, the balanced size of a $\shufm$-normal form $\wnf$ equals the minimal size of the type derivation of $\wnf$.

\newcounter{prop:number-steps}
\addtocounter{prop:number-steps}{\value{definition}}
\begin{proposition}[Exact number of $\betavm$-steps]
	\label{prop:number-steps}
\NoteProof{propappendix:number-steps}
	Let $\tm$ be a $\shufm$-normalizable term and $\tm_0$ be its $\shufm$-normal form. 
	For every reduction sequence $\deriv \colon \tm \toshufm^* \tm_0$ and every $\Type{\pi}{\tm}$ and $\Type{\pi_0}{\tm_0}$ such that $\size{\pi} = \min \{\size{\pi'} \mid \Type{\pi'}{\tm}\}$ and $\size{\pi_0} = \min \{\size{\pi_0'} \mid \Type{\pi_0'}{\tm_0} \}$, one has
	\begin{equation}\label{eq:number-steps}
		\Lengbv{\deriv} = \size{\pi} - \sizeZero{\tm_0} = \size{\pi} - \size{\pi_0}\,.
	\end{equation}
	If moreover $\tm_0$ is a value, then $\Lengbv{\deriv} = \size{\pi}$.
\end{proposition}

In particular, Eq.~\refeq{number-steps} implies that for \emph{any} reduction sequence $\deriv \colon \tm \toshufm^* \tm_0$ and \emph{any} $\Type{\pi}{\tm}$ and $\Type{\pi_0}{\tm_0}$ such that $\size{\pi_0} = \min \{\size{\pi_0'} \mid \Type{\pi_0'}{\tm_0} \}$, one has $\Lengbv{\deriv} \leq \size{\pi} - \sizeZero{\tm_0} = \size{\pi} - \size{\pi_0}\,$, since $\size{\pi} \geq \min \{\size{\pi'} \mid \Type{\pi'}{\tm}\}$.

\refprop{number-steps} could seem slightly disappointinig: it allows us to know the exact number of $\betavm$-steps of a $\shufm$-normalizing reduction sequence from $\tm$ only if we already know the $\shufm$-normal form $\tm_0$ of $\tm$ (or the minimal derivation of $\tm_0$), which essentially means that we have to perform the reduction sequence in order to know the exact number of its $\betavm$-steps.
However, \refprop{number-steps} says also that this limitation is circumvented in the case $\tm $ $\shufm$-reduces to a value. 
Moreover, a notable and immediate consequence of \refprop{number-steps} is: 

\begin{corollary}[Same number of $\betavm$-steps]
\label{coro:same-number}
	Let $\tm$ be a $\shufm$-normalizable term and $\tm_0$ be its $\shufm$-normal form.
	For all reduction sequences $\deriv \colon \tm \toshufm^* \tm_0$ and $\derivp \colon \tm \toshufm^* \tm_0$, one has $\Lengbv{\deriv} = \Lengbv{\derivp}$.
\end{corollary}
 
Even if $\shufm$-reduction is weak, in the sense that it does not reduce under $\lambda$'s, \refcoro{same-number} is not obvious at all, since the rewriting theory of $\shufm$-reduction is not quite elegant, in particular it does not enjoy any form of (quasi-)diamond property because of $\sigma$-reduction, as shown by the following example.

\begin{example}
\label{ex:critical-pair-sigma}
  Let $\tm \defeq (\la{\vartwo}{\vartwo'})(\Delta(\var I))I$: one has $\tmtwo \defeq (\la{\vartwo}{\vartwo'})(\Delta(\var I)) \RevTo{\sigl} \tm \tosr (\la{\varthree}{(\la{\vartwo}{\vartwo'})(\varthree\varthree)})(\var I)I \eqdef \tmthree$ and the only way to join this critical pair is by performing \emph{one} $\sigr$-step from $\tmtwo$ and \emph{two} $\sigl$-steps from $\tmthree$, so that $\tmtwo \tosr (\la{\varthree}{(\la{\vartwo}{\vartwo'I})(\varthree\varthree)})(\var I)  \RevTo{\sigl} (\la{\varthree}{(\la{\vartwo}{\vartwo'})(\varthree\varthree)I})(\var I) \RevTo{\sigl} \tmthree$.
  Since each $\sigm$-step can create a new $\betav$-redex in a balanced context (as shown in \refex{reductions}), \emph{a priori} there is no evidence that \refcoro{same-number} should hold.
\end{example}

\refcor{same-number} allows us to define the following function $\Lengsym_\betav \colon \Lambda \to \nat \cup \{\infty\}$
\begin{equation*}
  \Lengbv{\tm} =
  \begin{cases}
    \Lengbv{\deriv} &\text{if there is a $\shufm$-normalizing reduction sequence $\deriv$ from $\tm$;} \\
    \infty &\text{otherwise.}
  \end{cases}
\end{equation*}
In other words, in $\shufcalc$ we can univocally associate with every term the number of $\betavm$-steps needed to reach its $\shufm$-normal form, if any (the infinity $\infty$ is associated with 
non-$\shufm$-normalizable terms).
The characterization  of $\shufm$-normalization given in \refthm{characterize-normalizable} allows us to determine through semantic or logical means if the value of $\Lengbv{\tm}$ is a finite number or not.

Quantitatively, via \reflemma{semantic-value} 
we can simplify the way to compute the number of $\betavm$-steps to reach the $\shufm$-normal form of a valuable (\ie~that reduces to a value) term $\tm$, using only a specific type derivation of $\tm$.

\newcounter{thm:number-steps-value}
\addtocounter{thm:number-steps-value}{\value{definition}}
\begin{theorem}[Exact number of $\betavm$-steps for valuables]
  \label{thm:number-steps-value}
  \NoteProof{thmappendix:number-steps-value}
  If $\tm \toshufm^*\! \val \in \valSet$ then $\Lengbv{\tm} = \size{\pi}$ for $\concl{\pi}{\,}{\tm}{\emptymset}$.
\end{theorem}


\refprop{semantic-valuable} and \refthm{number-steps-value} provide a procedure to determine 
if a term $\tm$ $\shufm$-normalizes to a value 
and, in case, how many $\betavm$-steps are needed to reach its $\shufm$-normal form (this number does not depend on the reduction strategy according to \refcor{same-number}), considering only the term $\tm$ and without performing any $\shufm$-
step:
\begin{enumerate}
  \item check if there is a derivation $\pi$ with empty types, \ie~$\concl{\pi}{\,}{\tm}{\emptymset}$;
  \item if it is so (\ie~if $\tm$ $\shufm$-normalize to a value, according to \refprop{semantic-valuable}), compute the size $\size{\pi}$.
\end{enumerate}

Remind that, according to \refcoro{syntactic-normal-closed}, any closed term either is not $\shufm$-normalizable, or it $\shufm$-normalizes to a (closed) value. 
So, this procedure completely determines (qualitatively and quantitatively) the behavior of closed terms with respect to $\shufm$-reduction (and to $\betavm$-reduction, as we will see in \refsect{conclusions}).


%
%


\section{Conclusions}
\label{sect:conclusions}

\paragraph{Back to Plotkin's $\lambda_v$.}
%
%
The shuffling calculus $\shufcalc$ can be used to prove some properties of Plotkin's call-by-value $\lambda$-calculus $\lambda_v$ (whose only reduction rule is $\tobv$) \emph{restricted to closed terms}.
This is an example of how 
the study of some properties of a framework (in this case, $\lambda_v$) can be naturally 
done in a more general framework (in this case, $\shufcalc$).
It is worth noting that $\lambda_v$ with only closed terms is an interesting fragment: it represents the core of many functional programming languages, such as OCaml.

The starting point is \refcoro{syntactic-normal-closed}, which says that, in the closed setting with weak reduction, normal forms for $\shufcalc$ and $\lambda_v$ coincide: they are all and only closed values.
We can then reformulate \refthm{characterize-normalizable} and \refprop{semantic-valuable} as a semantic and logical characterization of $\betavm$-normalization in  Plotkin's $\lambda_v$ \emph{restricted to closed terms}.

\newcounter{thm:characterize-normalizable-Plotkin}
\addtocounter{thm:characterize-normalizable-Plotkin}{\value{definition}}
\begin{theorem}[Semantic and logical characterization of $\betavm$-normalization in the closed case]
	\label{thm:characterize-normalizable-Plotkin}
	\NoteProof{thmappendix:characterize-normalizable-Plotkin}
	Let $\tm$ be a closed term.
	The following are equivalent:
	\begin{enumerate}
	\item\label{p:characterize-normalizable-normalizable-Plotkin} \emph{Normalizability:} $\tm$ is $\betavm$-normalizable;
	\item\label{p:characterize-normalizable-valuable-Plotkin} \emph{Valuability:} $\tm \tobvm^* \val$ for some closed value $\val$;
	\item\label{p:characterize-normalizable-equivalen-normal-Plotkin} \emph{Completeness:} $\tm \betaveq \val$ for some closed  value $\val$;
	\item\label{p:characterize-normalizable-nonempty-Plotkin}\emph{Adequacy:} $\sem{\tm}{\vec{\var}} \neq \emptyset$ for any list $\vec{\var} = (\var_1, \dots, \var_k)$ (with $k \in \nat$) of pairwise distinct variables;
	\item\label{p:characterize-normalizable-semantic-Plotkin} \emph{Empty point:} $((\emptymset, \overset{k}{\dots}\,, \emptymset), \emptymset) \in \sem{\tm}{\vec{\var}}$ for any list $\vec{\var} = (\var_1, \dots, \var_k)$ ($k \in \nat$) of pairwise distinct variables;
	
	\item\label{p:characterize-normalizable-logic-Plotkin}\emph{Derivability with empty types:} there exists a derivation $\concl{\pi}{\,}{\tm}{\emptymset}$;
	
	\item\label{p:characterize-normalizable-derivable-Plotkin}\emph{Derivability:} there exists a derivation $\concl{\pi}{\,}{\tm}{Q}$ for some positive type $ Q$;
	\item\label{p:characterize-normalizable-strongly-normalizable-Plotkin}\emph{Strong normalizabilty:} $\tm$ is strongly $\betavm$-normalizable.
\end{enumerate}
\end{theorem}

We have already seen on p.~\pageref{thm:characterize-normalizable} that \refthm{characterize-normalizable-Plotkin} does not hold in $\lambda_v$ with open terms: closure is crucial.

\refthm{characterize-normalizable-Plotkin} entails that a closed term is $\shufm$-normalizable iff it is $\betavm$-normalizable iff it $\betavm$-reduces to a closed value.
Thus, \refcoro{same-number} and \refthm{number-steps-value} can be reformulated for $\lambda_v$ \emph{restricted to closed terms} as follows.

\newcounter{coro:same-number-Plotkin}
\addtocounter{coro:same-number-Plotkin}{\value{definition}}
\begin{corollary}[Same number of $\betavm$-steps]
	\label{coro:same-number-Plotkin}
	\NoteProof{coroappendix:same-number-Plotkin}
	Let $\tm$ be a closed $\betavm$-normalizable term and $\tm_0$ be its $\betavm$-normal form.
	For all reduction sequences $\deriv \colon \tm \tobvm^* \tm_0$ and $\derivp \colon \tm \tobvm^* \tm_0$, one has $\Lengbv{\deriv} = \Lengbv{\derivp}$.
\end{corollary}

\newcounter{thm:number-steps-value-Plotkin}
\addtocounter{thm:number-steps-value-Plotkin}{\value{definition}}
\begin{theorem}[
	Number of $\betavm$-steps]
	\label{thm:number-steps-value-Plotkin}
	\NoteProof{thmappendix:number-steps-value-Plotkin}
	If $\tm $ is closed and $\betavm$-normalizable, then $\Lengbv{\tm} = \size{\pi}$ for $\concl{\pi}{\,}{\tm}{\emptymset}$.
\end{theorem}

Clearly, the procedure sketched on p.~\pageref{sect:conclusions}, when applied to a \emph{closed} term $\tm$, determines if $\tm$ $\betavm$-normalizes
and, in case, how many $\betavm$-steps are needed to reach its $\betavm$-normal form.

\paragraph{Towards a semantic measure.}
In order to get a truly semantic measure of the execution time in the shuffling calculus $\shufcalc$, we should first be able to give an \emph{upper bound} to the number of $\betavm$-steps in a $\shufm$-reduction looking only at the semantics of terms. 
Therefore, we need to define a notion of size for the elements of the semantics of terms.
The most natural approach is the following.
For any positive type $P = [\Pair{P_1}{Q_1}, \dots, \Pair{P_k}{Q_k}] \in \MultiFin{\U}$ (with $k \in \nat$), the \emph{size of} $P$ is $\size{P} = k + \sum_{i=1}^k (\size{P_i} + \size{Q_i})$.
So, the size of a positive type $P$ is the number of occurrences of $\Pair{}{}$ in $P$; in particular, $\size{\emptymset} = 0$.
For any  $((P_1, \dots, P_n), Q) \in \MultiFin{\U}^k \times \MultiFin{\U}$ (with $k \in \nat$), the \emph{size of} $((P_1, \dots, P_k), Q)$ is \mbox{$\size{((P_1, \dots, P_k), Q)} = \size{Q} + \sum_{i=1}^k \size{P_i}$.}

The approach of \cite{deCarvalho18,DBLP:journals/tcs/CarvalhoPF11} relies on a crucial lemma to find an upper bound (and hence the exact length) of the execution time: 
it relates the size of a type derivation to the size of its conclusion, for a \emph{normal} term/proof-net.
In $\shufcalc$ this lemma should claim that ``For every \emph{$\shuf$-normal} form $\tm$, if $\concl{\pi}{\var_1\colon\! P_1, \dots, \var_k \colon\! P_k}{\tm}{Q}$ then $\size{\pi} \leq \size{((P_1, \dots, P_k), Q)}$''.
Unfortunately, in $\shufcalc$ this property is false!

\begin{example}
	\label{ex:counterexample}
	Let $\tm \defeq (\la{\var}{\var})(\vartwo\vartwo)$, which is a $\shuf$-normal form.
	Consider the derivation
	{\small
		\begin{align*}
		\pi \defeq
		\AxiomC{}
		\RightLabel{\footnotesize{$\mathsf{ax}$}}
		\UnaryInfC{$\var \colon\! \emptymset \vdash \var \colon\! \emptymset$}
		\RightLabel{\footnotesize{$\lambda$}}
		\UnaryInfC{$ \vdash \la{\var}{\var} \colon\! [\Pair{\emptymset}{\emptymset}]$}
		\AxiomC{}
		\RightLabel{\footnotesize{$\mathsf{ax}$}}
		\UnaryInfC{$\var \colon\! [\Pair{\emptymset}{\emptymset}] \vdash \var \colon\! [\Pair{\emptymset}{\emptymset}]$}
		\AxiomC{}
		\RightLabel{\footnotesize{$\mathsf{ax}$}}
		\UnaryInfC{$\var \colon\! \emptymset \vdash \var \colon\! \emptymset$}
		\RightLabel{\footnotesize{$@$}}
		\BinaryInfC{$\vartwo \colon\! [\Pair{\emptymset}{\emptymset}] \vdash \vartwo\vartwo \colon\! \emptymset$}
		\RightLabel{\footnotesize{$@$}}
		\BinaryInfC{$\vartwo \colon\! [\Pair{\emptymset}{\emptymset}] \vdash (\la{\var}{\var})(\vartwo\vartwo) \colon\! \emptymset$}
		\DisplayProof \,.		
		\end{align*}
	}
	\noindent Then, $\size{\pi} = 2 > 1 = \size{([\Pair{\emptymset}{\emptymset}],\emptymset)}$, which provides a counterexample to the property demanded above.
\end{example}

We conjecture that in order to overcome this counterexample (and to successfully follow the method of \cite{deCarvalho18,DBLP:journals/tcs/CarvalhoPF11} to get a purely semantic measure of the execution time) we should change the syntax and the operational semantics of our calculus, always remaining in a call-by-value setting equivalent (from the termination point of view) to $\shufcalc$ and the other calculi studied in \cite{DBLP:conf/aplas/AccattoliG16}.
Intuitively, in \refex{counterexample} $\tm$ contains 
one application\,---\,$(\la{\var}{\var})(\vartwo\vartwo)$\,---\,that is a stuck $\beta$-redex and is the source of one ``useless'' instance of the rule $@$ in $\pi$.
The idea for the new calculus is to ``fire'' a stuck $\beta$-redex $(\la{\var}{\tm})\tmtwo$ without performing the substitution $\tm\isub{\var}{\tmtwo}$ (as $\tmtwo$ might not be a value), but just creating an explicit substitution $\tm[\tmtwo/\var]$ that removes the application but ``stores'' the stuck $\beta$-redex. 
Such a calculus has been recently introduced~in~\cite{AccattoliGuerrieri18b}.

\paragraph{Related work.} 
This work has been presented at the workshop ITRS 2018. 
Later, the author further investigated this topic with Beniamino Accattoli in  \cite{AccattoliGuerrieri18b},
where we applied the same type system (and hence the same relational semantics) to a different call-by-value calculus with weak evaluation, $\firecalc$. 
The techniques used in both papers are similar (but not identical), some differences are due to the distinct calculi the type system is applied to. 
Some results are analogous: semantic and logical characterization of termination, extraction of quantitative information from type derivations. 
In \cite{AccattoliGuerrieri18b} we focused on an abstract characterization of the type derivations that provide an exact bound on the number of steps to reach the normal form.
Here, the semantic and logical characterization of termination is more informative than in \cite{AccattoliGuerrieri18b} because the reduction in $\shufcalc$ is not deterministic, contrary to $\firecalc$.
Moreover here, unlike \cite{AccattoliGuerrieri18b}, we investigate in detail the case of terms reducing to values and how the general results for $\shufcalc$ can be 
applied to 
analyze qualitative and quantitative properties of Plotkin's $\lambda_v$ restricted to closed terms (
see above).

Recently, Mazza, Pellissier and Vial \cite{MazzaPellissierVial18} introduced a general, elegant and abstract framework for building intersection (idempotent and non-idempotent) type systems characterizing normalization in different calculi.
However, such a work contains a wrong claim in one of its applications to concrete calculi and type systems, confirmed by a personal communication with the authors: they affirm that the same type system as the one used here characterizes normalization in Plotkin's $\lambda_v$ (endowed with the reduction $\tobvm$), but we have shown on p.~\pageref{thm:characterize-normalizable} that this is false for open terms.
Indeed, the property called full expansiveness in \cite{MazzaPellissierVial18} (which entails that ``normalizable $\Rightarrow$ typable'') actually does not hold in $\lambda_v$.
It is still true that their approach can be applied to characterize termination in Plotkin's $\lambda_v$ restricted to closed terms and in the shuffling calculus $\shufcalc$.
Proving that the abstract properties described in \cite{MazzaPellissierVial18} to characterize normalization hold in closed $\lambda_v$ or in $\shufcalc$ amounts essentially to show that subject reduction (our \refprop{quant-subject-reduction}), subject expansion (our \refprop{quant-subject-expansion}) and typability of normal forms (our \reflemma{semantics-normal}) hold.

The shuffling calculus $\shufcalc$ is 
compatible with Girard's call-by-value translation of $\lambda$-terms into linear logic ($\mathsf{LL}$) proof-nets: according to that, $\lambda$-values (which are the only duplicable and erasable $\lambda$-terms) are the only $\lambda$-terms translated as boxes; also, $\shuf$-reduction corresponds to cut-elimination and $\shufm$-reduction corresponds to cut-elimination at depth $0$ (i.e.~outside exponential boxes). 
The exact correspondence has many technical intricacies, which are outside the scope of this paper, anyway it can be recovered by composing the translation of the value substitution calculus (another extension of Plotkin's $\lambda_v$) into $\mathsf{LL}$ proof-nets (see \cite{DBLP:journals/tcs/Accattoli15}), and the encoding (studied in \cite{DBLP:conf/aplas/AccattoliG16}) of $\shufcalc$ into the value substitution calculus.
The relational semantics studied here is nothing but the relational semantics for $\mathsf{LL}$ (see \cite{DBLP:journals/tcs/CarvalhoPF11}) restricted to fragment of $\mathsf{LL}$ that is the image of Girard's call-by-value translation.
The notion of ``experiment'' in \cite{DBLP:journals/tcs/CarvalhoPF11} corresponds to our type derivation, and the ``result'' of an experiment there corresponds to the conclusion of a type derivation here.  
The main results of \cite{DBLP:journals/tcs/CarvalhoPF11} are similar to ours: characterization of normalization for $\mathsf{LL}$ proof-nets, extraction of quantitative information from (results of) experiments.
Nonetheless, the properties shown here for $\shufcalc$ cannot be derived by simply analyzing the analogous results for $\mathsf{LL}$ proof-nets (proven in \cite{DBLP:journals/tcs/CarvalhoPF11}) within its call-by-value fragment.
Indeed, \refex{counterexample} shows that some property, which holds in the\,---\,apparently\,---\,more general case of untyped $\mathsf{LL}$ proof-nets (as proven in \cite{DBLP:journals/tcs/CarvalhoPF11}), does not hold in the\,---\,apparently\,---\,special case of terms in $\shufcalc$.
It could seem surprising but, actually, there is no contradiction because $\mathsf{LL}$ proof-nets in \cite{DBLP:journals/tcs/CarvalhoPF11} always require an explicit constructor for dereliction, whereas $\shufcalc$ is outside of this fragment since variables correspond in $\mathsf{LL}$ proof-nets to exponential axioms (which keep implicit the dereliction).

\phantomsection
\addcontentsline{toc}{section}{References}
\bibliographystyle{eptcs}
\bibliography{\macrospath/biblio}

\begin{thebibliography}{10}
\providecommand{\bibitemdeclare}[2]{}
\providecommand{\surnamestart}{}
\providecommand{\surnameend}{}
\providecommand{\urlprefix}{Available at }
\providecommand{\url}[1]{\texttt{#1}}
\providecommand{\href}[2]{\texttt{#2}}
\providecommand{\urlalt}[2]{\href{#1}{#2}}
\providecommand{\doi}[1]{doi:\urlalt{http://dx.doi.org/#1}{#1}}
\providecommand{\bibinfo}[2]{#2}

\bibitemdeclare{article}{DBLP:journals/tcs/Accattoli15}
\bibitem{DBLP:journals/tcs/Accattoli15}
\bibinfo{author}{Beniamino \surnamestart Accattoli\surnameend}
  (\bibinfo{year}{2015}): \emph{\bibinfo{title}{{Proof nets and the
  call-by-value {\(\lambda\)}-calculus}}}.
\newblock {\sl \bibinfo{journal}{Theor. Comput. Sci.}} \bibinfo{volume}{606},
  pp. \bibinfo{pages}{2--24}.

\bibitemdeclare{inproceedings}{DBLP:conf/aplas/AccattoliG16}
\bibitem{DBLP:conf/aplas/AccattoliG16}
\bibinfo{author}{Beniamino \surnamestart Accattoli\surnameend} \&
  \bibinfo{author}{Giulio \surnamestart Guerrieri\surnameend}
  (\bibinfo{year}{2016}): \emph{\bibinfo{title}{{Open Call-by-Value}}}.
\newblock In: {\sl \bibinfo{booktitle}{{{APLAS} 2016}}}, pp.
  \bibinfo{pages}{206--226}.

\bibitemdeclare{inproceedings}{AccattoliGuerrieri18b}
\bibitem{AccattoliGuerrieri18b}
\bibinfo{author}{Beniamino \surnamestart Accattoli\surnameend} \&
  \bibinfo{author}{Giulio \surnamestart Guerrieri\surnameend}
  (\bibinfo{year}{2018}): \emph{\bibinfo{title}{Types of Fireballs}}.
\newblock In: {\sl \bibinfo{booktitle}{{APLAS 2018}}}, pp.
  \bibinfo{pages}{45--66}.

\bibitemdeclare{inproceedings}{DBLP:conf/flops/AccattoliP12}
\bibitem{DBLP:conf/flops/AccattoliP12}
\bibinfo{author}{Beniamino \surnamestart Accattoli\surnameend} \&
  \bibinfo{author}{Luca \surnamestart Paolini\surnameend}
  (\bibinfo{year}{2012}): \emph{\bibinfo{title}{{Call-by-Value Solvability,
  revisited}}}.
\newblock In: {\sl \bibinfo{booktitle}{{FLOPS}}}, pp. \bibinfo{pages}{4--16}.

\bibitemdeclare{inproceedings}{fireballs}
\bibitem{fireballs}
\bibinfo{author}{Beniamino \surnamestart Accattoli\surnameend} \&
  \bibinfo{author}{Claudio \surnamestart {Sacerdoti Coen}\surnameend}
  (\bibinfo{year}{2015}): \emph{\bibinfo{title}{{On the Relative Usefulness of
  Fireballs}}}.
\newblock In: {\sl \bibinfo{booktitle}{{{LICS} 2015}}}, pp.
  \bibinfo{pages}{141--155}.

\bibitemdeclare{book}{Barendregt84}
\bibitem{Barendregt84}
\bibinfo{author}{Hendrik~Pieter \surnamestart Barendregt\surnameend}
  (\bibinfo{year}{1984}): \emph{\bibinfo{title}{{The Lambda Calculus -- Its
  Syntax and Semantics}}}.
\newblock \bibinfo{volume}{103}, \bibinfo{publisher}{North-Holland}.

\bibitemdeclare{article}{DBLP:journals/iandc/BenedettiR16}
\bibitem{DBLP:journals/iandc/BenedettiR16}
\bibinfo{author}{Erika~De \surnamestart Benedetti\surnameend} \&
  \bibinfo{author}{Simona \surnamestart {Ronchi Della Rocca}\surnameend}
  (\bibinfo{year}{2016}): \emph{\bibinfo{title}{A type assignment for
  {\(\lambda\)}-calculus complete both for {FPTIME} and strong normalization}}.
\newblock {\sl \bibinfo{journal}{Inf. Comput.}} \bibinfo{volume}{248}, pp.
  \bibinfo{pages}{195--214}.

\bibitemdeclare{article}{DBLP:journals/corr/BernadetL13}
\bibitem{DBLP:journals/corr/BernadetL13}
\bibinfo{author}{Alexis \surnamestart Bernadet\surnameend} \&
  \bibinfo{author}{St{\'{e}}phane \surnamestart Lengrand\surnameend}
  (\bibinfo{year}{2013}): \emph{\bibinfo{title}{Non-idempotent intersection
  types and strong normalisation}}.
\newblock {\sl \bibinfo{journal}{Logical Methods in Computer Science}}
  \bibinfo{volume}{9}(\bibinfo{number}{4}).

\bibitemdeclare{article}{DBLP:journals/apal/BucciarelliE01}
\bibitem{DBLP:journals/apal/BucciarelliE01}
\bibinfo{author}{Antonio \surnamestart Bucciarelli\surnameend} \&
  \bibinfo{author}{Thomas \surnamestart Ehrhard\surnameend}
  (\bibinfo{year}{2001}): \emph{\bibinfo{title}{On phase semantics and
  denotational semantics: the exponentials}}.
\newblock {\sl \bibinfo{journal}{Ann. Pure Appl. Logic}}
  \bibinfo{volume}{109}(\bibinfo{number}{3}), pp. \bibinfo{pages}{205--241}.

\bibitemdeclare{article}{DBLP:journals/apal/BucciarelliEM12}
\bibitem{DBLP:journals/apal/BucciarelliEM12}
\bibinfo{author}{Antonio \surnamestart Bucciarelli\surnameend},
  \bibinfo{author}{Thomas \surnamestart Ehrhard\surnameend} \&
  \bibinfo{author}{Giulio \surnamestart Manzonetto\surnameend}
  (\bibinfo{year}{2012}): \emph{\bibinfo{title}{A relational semantics for
  parallelism and non-determinism in a functional setting}}.
\newblock {\sl \bibinfo{journal}{Ann. Pure Appl. Logic}}
  \bibinfo{volume}{163}(\bibinfo{number}{7}), pp. \bibinfo{pages}{918--934}.

\bibitemdeclare{article}{DBLP:journals/igpl/BucciarelliKV17}
\bibitem{DBLP:journals/igpl/BucciarelliKV17}
\bibinfo{author}{Antonio \surnamestart Bucciarelli\surnameend},
  \bibinfo{author}{Delia \surnamestart Kesner\surnameend} \&
  \bibinfo{author}{Daniel \surnamestart Ventura\surnameend}
  (\bibinfo{year}{2017}): \emph{\bibinfo{title}{Non-idempotent intersection
  types for the Lambda-Calculus}}.
\newblock {\sl \bibinfo{journal}{Logic Journal of the {IGPL}}}
  \bibinfo{volume}{25}(\bibinfo{number}{4}), pp. \bibinfo{pages}{431--464}.

\bibitemdeclare{inproceedings}{DBLP:conf/fossacs/CarraroG14}
\bibitem{DBLP:conf/fossacs/CarraroG14}
\bibinfo{author}{Alberto \surnamestart Carraro\surnameend} \&
  \bibinfo{author}{Giulio \surnamestart Guerrieri\surnameend}
  (\bibinfo{year}{2014}): \emph{\bibinfo{title}{{A Semantical and Operational
  Account of Call-by-Value Solvability}}}.
\newblock In: {\sl \bibinfo{booktitle}{{{FOSSACS} 2014}}}, pp.
  \bibinfo{pages}{103--118}.

\bibitemdeclare{phdthesis}{Carvalho07}
\bibitem{Carvalho07}
\bibinfo{author}{Daniel \surnamestart de~Carvalho\surnameend}
  (\bibinfo{year}{2007}): \emph{\bibinfo{title}{S\'emantiques de la logique
  lin\'eaire et temps de calcul}}.
\newblock \bibinfo{type}{{T}h\`ese de doctorat}, \bibinfo{school}{Universit\'e
  Aix-Marseille II}.

\bibitemdeclare{article}{deCarvalho18}
\bibitem{deCarvalho18}
\bibinfo{author}{Daniel \surnamestart de~Carvalho\surnameend}
  (\bibinfo{year}{2018}): \emph{\bibinfo{title}{Execution time of
  {\(\lambda\)}-terms via denotational semantics and intersection types}}.
\newblock {\sl \bibinfo{journal}{Math. Str. in Comput. Sci.}}
  \bibinfo{volume}{28}(\bibinfo{number}{7}), pp. \bibinfo{pages}{1169--1203}.

\bibitemdeclare{article}{DBLP:journals/tcs/CarvalhoPF11}
\bibitem{DBLP:journals/tcs/CarvalhoPF11}
\bibinfo{author}{Daniel \surnamestart de~Carvalho\surnameend},
  \bibinfo{author}{Michele \surnamestart Pagani\surnameend} \&
  \bibinfo{author}{Lorenzo \surnamestart {Tortora de Falco}\surnameend}
  (\bibinfo{year}{2011}): \emph{\bibinfo{title}{{A semantic measure of the
  execution time in linear logic}}}.
\newblock {\sl \bibinfo{journal}{Theor. Comput. Sci.}}
  \bibinfo{volume}{412}(\bibinfo{number}{20}), pp. \bibinfo{pages}{1884--1902}.

\bibitemdeclare{article}{DBLP:journals/iandc/CarvalhoF16}
\bibitem{DBLP:journals/iandc/CarvalhoF16}
\bibinfo{author}{Daniel \surnamestart de~Carvalho\surnameend} \&
  \bibinfo{author}{Lorenzo \surnamestart {Tortora de Falco}\surnameend}
  (\bibinfo{year}{2016}): \emph{\bibinfo{title}{A semantic account of strong
  normalization in linear logic}}.
\newblock {\sl \bibinfo{journal}{Inf. Comput.}} \bibinfo{volume}{248}, pp.
  \bibinfo{pages}{104--129}.

\bibitemdeclare{article}{DBLP:journals/aml/CoppoD78}
\bibitem{DBLP:journals/aml/CoppoD78}
\bibinfo{author}{Mario \surnamestart Coppo\surnameend} \&
  \bibinfo{author}{Mariangiola \surnamestart Dezani{-}Ciancaglini\surnameend}
  (\bibinfo{year}{1978}): \emph{\bibinfo{title}{A new type assignment for
  {\(\lambda\)}-terms}}.
\newblock {\sl \bibinfo{journal}{Arch. Math. Log.}}
  \bibinfo{volume}{19}(\bibinfo{number}{1}), pp. \bibinfo{pages}{139--156}.

\bibitemdeclare{article}{DBLP:journals/ndjfl/CoppoD80}
\bibitem{DBLP:journals/ndjfl/CoppoD80}
\bibinfo{author}{Mario \surnamestart Coppo\surnameend} \&
  \bibinfo{author}{Mariangiola \surnamestart Dezani{-}Ciancaglini\surnameend}
  (\bibinfo{year}{1980}): \emph{\bibinfo{title}{An extension of the basic
  functionality theory for the {\(\lambda\)}-calculus}}.
\newblock {\sl \bibinfo{journal}{Notre Dame Journal of Formal Logic}}
  \bibinfo{volume}{21}(\bibinfo{number}{4}), pp. \bibinfo{pages}{685--693}.

\bibitemdeclare{inproceedings}{DBLP:conf/lfcs/Diaz-CaroMP13}
\bibitem{DBLP:conf/lfcs/Diaz-CaroMP13}
\bibinfo{author}{Alejandro \surnamestart D{\'{\i}}az{-}Caro\surnameend},
  \bibinfo{author}{Giulio \surnamestart Manzonetto\surnameend} \&
  \bibinfo{author}{Michele \surnamestart Pagani\surnameend}
  (\bibinfo{year}{2013}): \emph{\bibinfo{title}{Call-by-Value Non-determinism
  in a Linear Logic Type Discipline}}.
\newblock In: {\sl \bibinfo{booktitle}{{LFCS} 2013}}, pp.
  \bibinfo{pages}{164--178}.

\bibitemdeclare{inproceedings}{DBLP:conf/csl/Ehrhard12}
\bibitem{DBLP:conf/csl/Ehrhard12}
\bibinfo{author}{Thomas \surnamestart Ehrhard\surnameend}
  (\bibinfo{year}{2012}): \emph{\bibinfo{title}{{Collapsing non-idempotent
  intersection types}}}.
\newblock In: {\sl \bibinfo{booktitle}{{CSL}}}, pp. \bibinfo{pages}{259--273}.

\bibitemdeclare{inproceedings}{DBLP:conf/ppdp/EhrhardG16}
\bibitem{DBLP:conf/ppdp/EhrhardG16}
\bibinfo{author}{Thomas \surnamestart Ehrhard\surnameend} \&
  \bibinfo{author}{Giulio \surnamestart Guerrieri\surnameend}
  (\bibinfo{year}{2016}): \emph{\bibinfo{title}{The Bang Calculus: an untyped
  lambda-calculus generalizing call-by-name and call-by-value}}.
\newblock In: {\sl \bibinfo{booktitle}{PPDP 2016}}, \bibinfo{publisher}{{ACM}},
  pp. \bibinfo{pages}{174--187}.

\bibitemdeclare{inproceedings}{DBLP:conf/tacs/Gardner94}
\bibitem{DBLP:conf/tacs/Gardner94}
\bibinfo{author}{Philippa \surnamestart Gardner\surnameend}
  (\bibinfo{year}{1994}): \emph{\bibinfo{title}{Discovering Needed Reductions
  Using Type Theory}}.
\newblock In: {\sl \bibinfo{booktitle}{TACS '94}}, {\sl
  \bibinfo{series}{Lecture Notes in Computer Science}} \bibinfo{volume}{789},
  \bibinfo{publisher}{Springer}, pp. \bibinfo{pages}{555--574}.

\bibitemdeclare{article}{DBLP:journals/tcs/Girard87}
\bibitem{DBLP:journals/tcs/Girard87}
\bibinfo{author}{Jean-Yves \surnamestart Girard\surnameend}
  (\bibinfo{year}{1987}): \emph{\bibinfo{title}{{Linear Logic}}}.
\newblock {\sl \bibinfo{journal}{Theoretical Computer Science}}
  \bibinfo{volume}{50}, pp. \bibinfo{pages}{1--102}.

\bibitemdeclare{article}{DBLP:journals/apal/Girard88}
\bibitem{DBLP:journals/apal/Girard88}
\bibinfo{author}{Jean{-}Yves \surnamestart Girard\surnameend}
  (\bibinfo{year}{1988}): \emph{\bibinfo{title}{Normal functors, power series
  and {\(\lambda\)}-calculus}}.
\newblock {\sl \bibinfo{journal}{Ann. Pure Appl. Logic}}
  \bibinfo{volume}{37}(\bibinfo{number}{2}), pp. \bibinfo{pages}{129--177}.

\bibitemdeclare{inproceedings}{DBLP:conf/icfp/GregoireL02}
\bibitem{DBLP:conf/icfp/GregoireL02}
\bibinfo{author}{Benjamin \surnamestart Gr{\'e}goire\surnameend} \&
  \bibinfo{author}{Xavier \surnamestart Leroy\surnameend}
  (\bibinfo{year}{2002}): \emph{\bibinfo{title}{{A compiled implementation of
  strong reduction}}}.
\newblock In: {\sl \bibinfo{booktitle}{{{ICFP} '02}}}, pp.
  \bibinfo{pages}{235--246}.

\bibitemdeclare{inproceedings}{Guerrieri15}
\bibitem{Guerrieri15}
\bibinfo{author}{Giulio \surnamestart Guerrieri\surnameend}
  (\bibinfo{year}{2015}): \emph{\bibinfo{title}{{Head reduction and
  normalization in a call-by-value lambda-calculus}}}.
\newblock In: {\sl \bibinfo{booktitle}{{{WPTE} 2015}}}, pp.
  \bibinfo{pages}{3--17}.

\bibitemdeclare{inproceedings}{GuerrieriPR15}
\bibitem{GuerrieriPR15}
\bibinfo{author}{Giulio \surnamestart Guerrieri\surnameend},
  \bibinfo{author}{Luca \surnamestart Paolini\surnameend} \&
  \bibinfo{author}{Simona \surnamestart {Ronchi Della Rocca}\surnameend}
  (\bibinfo{year}{2015}): \emph{\bibinfo{title}{{Standardization of a
  Call-By-Value Lambda-Calculus}}}.
\newblock In: {\sl \bibinfo{booktitle}{{{TLCA} 2015}}}, pp.
  \bibinfo{pages}{211--225}.

\bibitemdeclare{article}{DBLP:journals/lmcs/GuerrieriPR17}
\bibitem{DBLP:journals/lmcs/GuerrieriPR17}
\bibinfo{author}{Giulio \surnamestart Guerrieri\surnameend},
  \bibinfo{author}{Luca \surnamestart Paolini\surnameend} \&
  \bibinfo{author}{Simona \surnamestart {Ronchi Della Rocca}\surnameend}
  (\bibinfo{year}{2017}): \emph{\bibinfo{title}{{Standardization and
  Conservativity of a Refined Call-by-Value lambda-Calculus}}}.
\newblock {\sl \bibinfo{journal}{Logical Methods in Computer Science}}
  \bibinfo{volume}{13}(\bibinfo{number}{4}).

\bibitemdeclare{book}{Jones:1993:PEA:153676}
\bibitem{Jones:1993:PEA:153676}
\bibinfo{author}{Neil~D. \surnamestart Jones\surnameend},
  \bibinfo{author}{Carsten~K. \surnamestart Gomard\surnameend} \&
  \bibinfo{author}{Peter \surnamestart Sestoft\surnameend}
  (\bibinfo{year}{1993}): \emph{\bibinfo{title}{{Partial Evaluation and
  Automatic Program Generation}}}.
\newblock \bibinfo{publisher}{Prentice-Hall, Inc.}, \bibinfo{address}{Upper
  Saddle River, NJ, USA}.

\bibitemdeclare{inproceedings}{DBLP:conf/ictac/KesnerV15}
\bibitem{DBLP:conf/ictac/KesnerV15}
\bibinfo{author}{Delia \surnamestart Kesner\surnameend} \&
  \bibinfo{author}{Daniel \surnamestart Ventura\surnameend}
  (\bibinfo{year}{2015}): \emph{\bibinfo{title}{A Resource Aware Computational
  Interpretation for Herbelin's Syntax}}.
\newblock In: {\sl \bibinfo{booktitle}{{ICTAC} 2015}}, pp.
  \bibinfo{pages}{388--403}.

\bibitemdeclare{inproceedings}{DBLP:conf/rta/KesnerV17}
\bibitem{DBLP:conf/rta/KesnerV17}
\bibinfo{author}{Delia \surnamestart Kesner\surnameend} \&
  \bibinfo{author}{Pierre \surnamestart Vial\surnameend}
  (\bibinfo{year}{2017}): \emph{\bibinfo{title}{Types as Resources for
  Classical Natural Deduction}}.
\newblock In: {\sl \bibinfo{booktitle}{{FSCD} 2017}}, {\sl
  \bibinfo{series}{LIPIcs}}~\bibinfo{volume}{84}, pp.
  \bibinfo{pages}{24:1--24:17}.

\bibitemdeclare{article}{DBLP:journals/logcom/Kfoury00}
\bibitem{DBLP:journals/logcom/Kfoury00}
\bibinfo{author}{Assaf~J. \surnamestart Kfoury\surnameend}
  (\bibinfo{year}{2000}): \emph{\bibinfo{title}{A linearization of the
  Lambda-calculus and consequences}}.
\newblock {\sl \bibinfo{journal}{J. Log. Comput.}}
  \bibinfo{volume}{10}(\bibinfo{number}{3}), pp. \bibinfo{pages}{411--436}.

\bibitemdeclare{book}{DBLP:books/daglib/0071545}
\bibitem{DBLP:books/daglib/0071545}
\bibinfo{author}{Jean{-}Louis \surnamestart Krivine\surnameend}
  (\bibinfo{year}{1993}): \emph{\bibinfo{title}{Lambda-calculus, types and
  models}}.
\newblock \bibinfo{series}{Ellis Horwood series}, \bibinfo{publisher}{Masson}.

\bibitemdeclare{article}{MazzaPellissierVial18}
\bibitem{MazzaPellissierVial18}
\bibinfo{author}{Damiano \surnamestart Mazza\surnameend}, \bibinfo{author}{Luc
  \surnamestart Pellissier\surnameend} \& \bibinfo{author}{Pierre \surnamestart
  Vial\surnameend} (\bibinfo{year}{2018}): \emph{\bibinfo{title}{Polyadic
  Approximations, Fibrations and Intersection Types}}.
\newblock {\sl \bibinfo{journal}{Proceedings of the ACM on Programming
  Languages}} \bibinfo{volume}{2}(\bibinfo{number}{{POPL}}), pp.
  \bibinfo{pages}{6:1--6:28}.

\bibitemdeclare{inproceedings}{DBLP:conf/icfp/NeergaardM04}
\bibitem{DBLP:conf/icfp/NeergaardM04}
\bibinfo{author}{Peter~M{\o}ller \surnamestart Neergaard\surnameend} \&
  \bibinfo{author}{Harry~G. \surnamestart Mairson\surnameend}
  (\bibinfo{year}{2004}): \emph{\bibinfo{title}{Types, potency, and
  idempotency: why nonlinearity and amnesia make a type system work}}.
\newblock In: {\sl \bibinfo{booktitle}{{ICFP} 2004}}, pp.
  \bibinfo{pages}{138--149}.

\bibitemdeclare{inproceedings}{DBLP:conf/ictcs/Paolini01}
\bibitem{DBLP:conf/ictcs/Paolini01}
\bibinfo{author}{Luca \surnamestart Paolini\surnameend} (\bibinfo{year}{2002}):
  \emph{\bibinfo{title}{{Call-by-Value Separability and Computability}}}.
\newblock In: {\sl \bibinfo{booktitle}{{ICTCS}}}, pp. \bibinfo{pages}{74--89}.

\bibitemdeclare{article}{DBLP:journals/mscs/PaoliniPR17}
\bibitem{DBLP:journals/mscs/PaoliniPR17}
\bibinfo{author}{Luca \surnamestart Paolini\surnameend}, \bibinfo{author}{Mauro
  \surnamestart Piccolo\surnameend} \& \bibinfo{author}{Simona \surnamestart
  {Ronchi Della Rocca}\surnameend} (\bibinfo{year}{2017}):
  \emph{\bibinfo{title}{Essential and relational models}}.
\newblock {\sl \bibinfo{journal}{Mathematical Structures in Computer Science}}
  \bibinfo{volume}{27}(\bibinfo{number}{5}), pp. \bibinfo{pages}{626--650}.

\bibitemdeclare{article}{DBLP:journals/tcs/Plotkin75}
\bibitem{DBLP:journals/tcs/Plotkin75}
\bibinfo{author}{Gordon~D. \surnamestart Plotkin\surnameend}
  (\bibinfo{year}{1975}): \emph{\bibinfo{title}{{Call-by-Name, Call-by-Value
  and the lambda-Calculus}}}.
\newblock {\sl \bibinfo{journal}{Theor. Comput. Sci.}}
  \bibinfo{volume}{1}(\bibinfo{number}{2}), pp. \bibinfo{pages}{125--159}.

\bibitemdeclare{inproceedings}{Pottinger80}
\bibitem{Pottinger80}
\bibinfo{author}{Garrel \surnamestart Pottinger\surnameend}
  (\bibinfo{year}{1980}): \emph{\bibinfo{title}{A type assignment for the
  strongly normalizable $\lambda$-terms}}.
\newblock In: {\sl \bibinfo{booktitle}{To HB Curry: essays on combinatory
  logic, $\lambda$-calculus and formalism}}, pp. \bibinfo{pages}{561--577}.

\bibitemdeclare{phdthesis}{Reg:Thesis:92}
\bibitem{Reg:Thesis:92}
\bibinfo{author}{Laurent \surnamestart Regnier\surnameend}
  (\bibinfo{year}{1992}): \emph{\bibinfo{title}{{Lambda-calcul et
  r{\'e}seaux}}}.
\newblock \bibinfo{type}{Ph{D} thesis}, \bibinfo{school}{Univ. Paris VII}.

\bibitemdeclare{article}{Regnier94}
\bibitem{Regnier94}
\bibinfo{author}{Laurent \surnamestart Regnier\surnameend}
  (\bibinfo{year}{1994}): \emph{\bibinfo{title}{{Une {\'e}quivalence sur les
  lambda-termes}}}.
\newblock {\sl \bibinfo{journal}{Theoretical Computer Science}}
  \bibinfo{volume}{2}(\bibinfo{number}{126}), pp. \bibinfo{pages}{281--292}.

\bibitemdeclare{book}{parametricBook}
\bibitem{parametricBook}
\bibinfo{author}{Simona \surnamestart {Ronchi Della Rocca}\surnameend} \&
  \bibinfo{author}{Luca \surnamestart Paolini\surnameend}
  (\bibinfo{year}{2004}): \emph{\bibinfo{title}{{The Parametric
  $\l$-Calculus}}}.
\newblock \bibinfo{publisher}{Springer}.

\end{thebibliography}

\newpage
\appendix

\section{Technical appendix: omitted proofs}
\label{sect:proofs}

The enumeration of propositions, theorems, lemmas already stated in the body of the article is unchanged

\subsection{Preliminaries and notations}
  The set of $\lambda$-terms is denoted by $\Lambda$.
  We set $I \defeq \lambda x.x$ and $\Delta \defeq \lambda x. xx$.
  Let $\Rew{\Rule} \, \subseteq \Lambda \times \Lambda$.
  \begin{itemize}
    \item The reflexive-transitive closure of $\Rew{\Rule}$ is denoted by $\Rew{\Rule}^*$.
    The \emph{$\Rule$-equivalence} $\simeq_\Rule$ is the 
    reflexive-transitive and symmetric closure of $\to_\Rule$.

    \item  Let $\tm$ be a term: $\tm$ is \emph{$\Rule$-normal} if there is no term $\tmtwo$ such that $\tm \to_\Rule \tmtwo$; $\tm$ is \emph{$\Rule$-normalizable} if there is a $\Rule$-normal term $\tmtwo$ such that $\tm \to_\Rule^* \tmtwo$, and we then say that $\tmtwo$ is a \emph{$\Rule$-normal form of $\tm$}; 
    $\tm$ is \emph{strongly $\Rule$-normalizable} if it does not exist an infinite sequence of $\Rule$-reductions starting from $\tm$.
    Finally, $\to_\Rule$ is \emph{strongly normalizing} if every $\tmtwo \in \Lambda$ is strongly $\Rule$-normalizable.
    
    \item $\Rew{\Rule}$ is \emph{confluent} if $\MRevTo{\Rule} \cdot\! \Rew{\Rule}^* \ \subseteq \ \Rew{\Rule}^* \!\cdot \MRevTo{\Rule}$.
    From confluence it follows that: $\tm \simeq_{\Rule} \tmtwo$ iff $\tm \to_\Rule^* \tmthree \,\,{}_\Rule^*\!\!\!\leftarrow \tmtwo$ for some term $\tmthree$; and any $\Rule$-normalizable term has a \emph{unique} $\Rule$-normal form.

  \end{itemize}

\subsection{Omitted proofs and remarks of Section~\ref{sect:calculus}}

\setcounter{propositionAppendix}{\value{prop:syntactic-normal}}
\begin{propositionAppendix}[Syntactic characterization on $\shufm$-normal forms]
\label{propAppendix:syntactic-normal}
  Let $\tm$ be a term:
\NoteState{prop:syntactic-normal}
  \begin{itemize}
    \item $\tm$ is $\shufm$-normal iff $\tm \in \wnfSet$;
    \item $\tm$ is $\shufm$-normal and is neither a value nor a $\beta$-redex iff $\tm \in \anfSet$.
  \end{itemize}
\end{propositionAppendix}

\begin{proof}\hfill
  \begin{itemize}
    \item[$\Rightarrow$:] We prove the left-to-right direction of both statements simultaneously by induction on $\tm \in \Lambda$.

    If $\tm$ is a value then $\tm \in \wnfSet$ by definition.

    Otherwise $\tm = \tmtwo\tmthree$ for some terms $\tmtwo, \tmthree$. 
    By simple inspection of the rules of $\Rew{\shufm}$, one can deduce that $\tmtwo$ and $\tmthree$ $\shufm$-normal, $\tmtwo$ is not 
    a $\beta$-redex (otherwise $\tm$ would be a $\sigma_1$-redex) and if $\tmtwo$ is of the shape $\la{\var}{\tmtwo'}$ then $\tmtwo'$ is $\shufm$-normal; furthermore $\tm$ is neither a $\beta_v$- nor a $\sigma_3$-redex, hence there are only three possibilities:
    \begin{enumerate}
      \item $\tmtwo$ is not a value: by induction hypothesis $\tmtwo \in \anfSet$ and $\tmthree \in \wnfSet$, therefore $\tm \in 
      \anfSet$.
      \item $\tmtwo$ is not an abstraction and $\tmthree$ is not 
      a $\beta$-redex: either $\tmtwo$ is a variable or $\tmtwo \in \anfSet$ by induction hypothesis (since $\tmtwo$ is neither a value nor a $\beta$-redex). 
      If $\tmthree$ is not a value then $\tmthree \in \anfSet$ by induction hypothesis, 
      so $\tm \in \anfSet$ because $\tm$ is either of the form $\var\anf$ either of the form $\anf'\anf$ (with $\anfSet \subseteq \wnfSet$). 
      Otherwise $\tmthree$ is a value, 
	  thus $\tm \in \anfSet$ since $\tm$ is either of the form $\var\val$ either of the form $\anf\val$ (with $\valSet \subseteq \wnfSet$).
      \item $\tmthree$ is neither a value nor 
      a $\beta$-redex: by induction hypothesis $\tmthree \in \anfSet$; 
      if $\tmtwo$ is a variable then $\tm \in 
      \anfSet$ because $\tm$ is of the form $\var\anf$; 
      if $\tmtwo$ is an abstraction then $\tmtwo = \la{\var}{\tmtwo'}$ where $\tmtwo'$ is $\shufm$-normal, so $\tmtwo' \in \wnfSet$ by induction hypothesis and thus $\tm \in 
      \wnfSet$ since $\tm$ is of the form $(\la{\var}{\wnf})\anf$; 
      finally, if $\tmtwo$ is not a value then $\tmtwo \in \anfSet$ by induction hypothesis, hence $\tm \in 
      \anfSet$ because $\tm$ is of the form $\anf'\anf$ (with $\anfSet \subseteq \wnfSet$).
    \end{enumerate}

    \item[$\Leftarrow$:] 
    The second statement follows 
    from the first one, since $\anfSet \subseteq \wnfSet$ and if $\tm \in \anfSet$ then $\tm$ is neither a value nor a $\beta$-redex.
    We prove the first statement by induction on $\tm \in \wnfSet$.

    If $\tm$ is a value then $\tm$ is $\shufm$-normal (no rule of $\Rew{\shufm}$ can be applied to $\tm$).
  
    If $\tm = \var\val$ for some variable $\var$ and value $\val$ then $x$ and $\val$ are $\shufm$-normal and 
    $\var\val$ is not a $\shuf$-redex; therefore $\tm$ is $\shufm$-normal.

    If $\tm = \var\anf$ for some variable $\var$ and term $\anf \in \anfSet \subseteq \wnfSet$, then $\var$ and (by induction hypothesis) $\anf$ are $\shufm$-normal, moreover $\anf$ is not 
    a $\beta$-redex (so $\tm$ is not a $\sigma_3$-redex)
    , thus $\tm$ is $\shufm$-normal.

    If $\tm = \anf\wnf$ or $\tm = (\la{\var}{\wnf})\anf$ for some $\anf \in \anfSet \subseteq \wnfSet$ and $\wnf \in \wnfSet$, then $\anf$ and $\wnf$ are $\shufm$-normal by induction hypothesis; and $\anf$ is neither a value nor a 
    $\beta$-redex, thus 
    $\tm$ is not a $\shuf$-redex; hence, $\tm$ is $\shufm$-normal.
    \qedhere
  \end{itemize}
\end{proof}

\setcounter{corollaryAppendix}{\value{coro:syntactic-normal-closed}}
\begin{corollaryAppendix}[Syntactic characterization of closed $\shufm$- and $\betavm$-normal forms]
	\label{coroAppendix:syntactic-normal-closed}
	\NoteState{coro:syntactic-normal-closed}
	Let $\tm$ be a closed term:
	 $\tm$ is $\shufm$-normal iff $\tm$ is $\betavm$-normal iff $\tm$ is a value iff $\tm = \la{\var}{\tmtwo}$ for some term $\tmtwo$ with $\Fv{\tmtwo} \subseteq\{\var\}$.
\end{corollaryAppendix}

\begin{proof}
	By \refprop{syntactic-normal} and since $\tm$ is closed, $\tm$ is $\shufm$-normal iff $\tm$ is a value (as all terms in $\anfSet$ are open).
	Since $\tm$ is closed and variables are open, $\tm$ is a value iff $\tm = \la{\var}{\tmtwo}$ for some $\tmtwo$ with $\Fv{\tmtwo} \subseteq\{\var\}$.
	If $\tm$ is $\shufm$-normal then it is $\betavm$-normal because $\tobvm \, \subseteq \, \toshufm$;
	conversely, if $\tm$ is $\shufm$-normal then we have just proven that $\tm$ is an abstraction, which is $\betavm$-normal since $\tobvm$ does not reduce under $\lambda$'s. 
\end{proof}

\subsection{Omitted proofs and remarks of Section~\ref{sect:type}}

\begin{lemma}[Free variables in environment]
\label{l:free}
  If the judgment $\Gamma \vdash \tm \colon\! P$ is derivable then $\Dom{\Gamma} \subseteq \Fv{\tm}$.
\end{lemma}

\begin{proof}
  By straightforward induction 
  on $\tm \in \Lambda$.
\end{proof}

\begin{remark}
\label{rmk:positive-size}
  If $\tm$ is an application and $\Type{\pi}{\tm}$
  , then $\size{\pi} > 0$.
\end{remark}

\setcounter{lemmaAppendix}{\value{l:value}}
\begin{lemmaAppendix}[Judgment decomposition for values]
\label{lappendix:value}
  Let 
\NoteState{l:value}
  $\val \in \Lambda_v$, $\Delta$ be an environment, and $P_1, \dots, P_p$ be positive types (for some $p \in \nat$).
  There is a derivation $\concl{\pi}{\Delta}{\val}{P_1 \uplus \dots \uplus P_p}$ iff for all $1 \leq i \leq p$ there are an environment $\Delta_i$ and a derivation $\concl{\pi_i}{\Delta_i}{\val}{P_i}$ such that $\Delta = \biguplus_{i=1}^p \Delta_i$.
  Moreover, $\size{\pi} = \sum_{i=1}^p\size{\pi_i}$.
\end{lemmaAppendix}

\begin{proof}
  Both directions are proved by cases, depending on whether $\val$ is  a variable or an abstraction.
  \begin{description}
    \item [$\Rightarrow$:]
    If $\val = \vartwo$, then the last rule of $\pi$ is $\mathsf{ax}$ and thus $\Delta = y \colon\! P_1 \uplus \dots \uplus P_p$.
    So, for all $1 \leq i \leq p$, there are an environment $\Delta_i = y \colon\! P_i$ and a derivation 
    \begin{equation*}
      \pi_i = \AxiomC{}
      \RightLabel{\footnotesize$\mathsf{ax}$}
      \UnaryInfC{$\Delta_i \vdash \val \colon\! P_i$}
      \DisplayProof
    \end{equation*}
    with $\biguplus_{i=1}^p \Delta_i = \Delta$ and $\size{\pi} = 0 = \sum_{i=1}^p \size{\pi_i}$.
    
    If $\val = \la\var\tm$ then the last rule of $\pi$ is $\lambda$, so there are $n \in \nat$, positive types $Q_1, Q_1', \dots, Q_n, Q_n'$, environments $\Gamma_1, \dots, \Gamma_n$ such that $\Delta = \biguplus_{j=1}^n \!\Gamma_{\!j}$, $\biguplus_{i=1}^p \!P_i = [\Pair{Q_1\!}{Q_1'},\dots,\Pair{Q_n\!}{Q_n'}]$ and
    \begin{equation*}
      \pi = \AxiomC{$\ \vdots\,\pi_1'$}
      \noLine
      \UnaryInfC{$\Gamma_1, \var \colon\! Q_1 \vdash \tm \colon\! Q_1'$}
      \AxiomC{$\overset{n}{\dots}$}
      \AxiomC{$\ \vdots\,\pi_n'$}
      \noLine
      \UnaryInfC{$\Gamma_n, \var \colon\! Q_n \vdash \tm \colon\! Q_n'$}
      \RightLabel{\footnotesize$\lambda$}
      \TrinaryInfC{$\Delta \vdash \val \colon\! [\Pair{Q_1}{Q_1'}, \dots, \Pair{Q_n}{Q_n'}]$}
      \DisplayProof
    \end{equation*}
    So, up to renumbering the $Q_i$'s and $Q_i'$'s, there are environments $\Delta_1, \dots, \Delta_p$, derivations $\pi_1, \dots, \pi_p$ and integers $k_1 = 1 \leq k_2 \leq \dots \leq k_{p} \leq k_{p+1} = p$ such that, for all $1 \leq i \leq p$, $P_i = \biguplus_{j=k_i}^{k_{i+1}} [\Pair{Q_j}{Q_j'}]$ 
    \begin{align*}
      \text{and \ }
      \pi_i = \AxiomC{$\ \vdots\,\pi_{k_i}'$}
      \noLine
      \UnaryInfC{$\Gamma_{k_i}, \var \colon\! Q_{k_i} \vdash \tm \colon\! Q_{k_i}'$}
      \AxiomC{$\overset{k_{i+1}-k_i}{\ldots}$}
      \AxiomC{$\ \vdots\,\pi_{k_{i+1}}'$}
      \noLine
      \UnaryInfC{$\Gamma_{k_{i+1}}, \var \colon\! Q_{k_{i+1}} \vdash \tm \colon\! Q_{k_{i+1}}'$}
      \RightLabel{\footnotesize$\lambda$}
      \TrinaryInfC{$\Delta_i \vdash \val \colon\! P_i$}
      \DisplayProof
      &&\textup{with } \Delta_i = \biguplus_{j=k_i}^{k_{i+1}} \Gamma_j 
    \end{align*}
    where $\size{\pi_i} = \sum_{j=k_i}^{k_{i+1}} \size{\pi_j'}$, hence $\size{\pi} = \sum_{j=1}^n \size{\pi_j'} = \sum_{i=1}^p \sum_{j=k_i}^{k_{i+1}} \size{\pi_j'} = \sum_{i=1}^n \size{\pi_i}$.

    \item [$\Leftarrow$:]
    If $\val = \vartwo$ then, for all $1 \leq i \leq p$, the last rule of $\pi_i$ is $\mathsf{ax}$, so $\Delta_i = \vartwo \colon\! P_i$ and $\size{\pi_i} = 0$.
    Since $\Delta = \biguplus_{i=1}^p \Delta_i = \vartwo \colon\! \biguplus_{i=1}^p P_i$, there is a derivation 
    \begin{equation*}
      \pi = \AxiomC{}
      \RightLabel{\footnotesize$\mathsf{ax}$}
      \UnaryInfC{$\Delta \vdash \val \colon\! \biguplus_{i=1}^p P_i$}
      \DisplayProof
    \end{equation*}
     where $\size{\pi} = 0 = \sum_{i=1}^p \size{\pi_i}$.
     
    If $\val = \la\var\tm$ then, for all $1 \leq i \leq p$, the last rule of $\pi_i$ is $\lambda$, so there are $k_i \in \nat$, positive types $Q_{i1},Q_{i1}', \dots, Q_{ik_i},Q_{ik_i}'$, environments $\Delta_{i1}, \dots, \Delta_{ik_i}$ and derivations $\pi_{i1}, \dots, \pi_{ik_i}$ with $P_i = \biguplus_{j=1}^{k_i}[\Pair{Q_{ij}}{Q_{ij}'}]$, and $\Delta_i = \biguplus_{j=1}^{k_i} \Delta_{ij}$ and
    \begin{align*}
      \pi_i = \AxiomC{$\ \vdots\,\pi_{i1}$}
      \noLine
      \UnaryInfC{$\Delta_{i1}, \var \colon\! Q_{i1} \vdash \tm \colon\! Q_{i1}'$}
      \AxiomC{$\overset{k_i}{\dots}$}
      \AxiomC{$\ \vdots\,\pi_{ik_i}$}
      \noLine
      \UnaryInfC{$\Delta_{ik_i}, \var \colon\! Q_{ik_i} \vdash \tm \colon\! Q_{ik_i}'$}
      \RightLabel{\footnotesize$\lambda$}
      \TrinaryInfC{$\Delta_i \vdash \val \colon\! P_i$}
      \DisplayProof
      &&\textup{where } \size{\pi_i} = \sum_{j=1}^{k_i} \size{\pi_{ij}}.
    \end{align*}
    Since $\Delta = \biguplus_{i=1}^p \Delta_i$ and $\biguplus_{i=1}^p P_i = \biguplus_{i=1}^p\biguplus_{j=1}^{k_i}[\Pair{Q_{ij}}{Q_{ij}'}]$, there is a derivation $\pi =$
    {\footnotesize
    \begin{equation*}
      \AxiomC{$\ \vdots\,\pi_{11}$}
      \noLine
      \UnaryInfC{$\Delta_{11}, \var \colon\! Q_{11} \!\vdash\! \tm \colon\! Q_{11}' \ \, \overset{k_1}{\dots}$}
      \AxiomC{$\ \vdots\,\pi_{1k_1}$}
      \noLine
      \UnaryInfC{\!\!\!\!\!$\Delta_{1k_1}, \var \colon\! Q_{1k_1} \!\vdash\! \tm \colon\! Q_{1k_1}'$}
      \AxiomC{$\overset{p}{\dots}$}
      \AxiomC{$\ \vdots\,\pi_{p1}$}
      \noLine
      \UnaryInfC{$\Delta_{p1}, \var \colon\! Q_{p1} \!\vdash\! \tm \colon\! Q_{p1}' \ \, \overset{k_p}{\dots}$}
      \AxiomC{$\ \vdots\,\pi_{pk_p}$}
      \noLine
      \UnaryInfC{\!\!\!\!\!$\Delta_{pk_p}, \var \colon\! Q_{pk_p} \!\vdash\! \tm \colon\! Q_{pk_p}'$}
      \RightLabel{\tiny$\lambda$}
      \insertBetweenHyps{\hskip 3pt}
      \QuinaryInfC{\footnotesize$\Delta \vdash \val \colon\! \biguplus_{i=1}^p P_i$}
      \DisplayProof
    \end{equation*}
    }
    where $\size{\pi} = \sum_{i=1}^p \sum_{j=1}^{k_i} \size{\pi_{ij}} = \sum_{i=1}^p \size{\pi_i}$ because $\size{\pi_i} = \sum_{j=1}^{k_i} \size{\pi_{ij}}$.
    \qedhere
  \end{description}
\end{proof}

\begin{corollary}[Minimal derivation for values]
  \label{coro:minimal}
  For every $\val \in \valSet$, there exists 
  $\concl{\pi}{\,}{\val}{\emptymset}$ 
  with $\size{\pi} = 0$.
\end{corollary}

\begin{proof}
  Apply the right-to-left direction of \reflemma{value} taking $p = 0$.
\end{proof}

\setcounter{lemmaAppendix}{\value{l:substitution}}
\begin{lemmaAppendix}[Substitution]
\label{lappendix:substitution}
  Let 
\NoteState{l:substitution}
  $\tm \in \Lambda$ and $\val \in \Lambda_\val$.
  If $\concl{\pi}{\Gamma, \var \colon\! P}{\tm}{Q}$ and $\concl{\pi'}{\Delta}{\val}{P}$, then there exists $\concl{\pi''}{\Gamma \uplus \Delta}{\tm\isub\var\val}{Q}$ such that $\size{\pi''} = \size{\pi} + \size{\pi'}$.
\end{lemmaAppendix}

\begin{proof}
  By induction on $\tm \in \Lambda$.
  
  If $\tm = \var$, then $\tm\isub\var\val = \val$ and the last rule of $\pi$ is $\mathsf{ax}$ with $P = Q$ and $\Gamma = \vartwo_1 \colon\! \emptymset, \dots, \vartwo_n \colon\! \emptymset$ ($y_i \neq x$ for all $1 \leq i \leq n$), whence $\Gamma \uplus \Delta = \Delta$ and $\size{\pi} = 0$.
  We conclude by setting $\pi'' = \pi'$.
  
  If $\tm = \vartwo \neq \var$, then $\tm\isub\var\val = \vartwo$ and the last rule of $\pi$ is $\mathsf{ax}$ with $P = \emptymset$ (since $\var \neq \vartwo$), whence $\size{\pi} = 0$ and 
  $\Gamma = \vartwo \colon\! Q$.
  By \reflemma{value}, from $\concl{\pi'}{\Delta}{\val}{\emptymset}$ it follows that $\size{\pi'} = 0$ and $\Delta = \vartwo \colon\! \emptymset$, therefore $\Gamma \uplus \Delta = \Gamma$.
  So, the derivation 
  \begin{equation*}
    \pi'' = \AxiomC{}
    \RightLabel{\footnotesize$\mathsf{ax}$}
    \UnaryInfC{$\Gamma \vdash \vartwo \colon\! Q$}
    \DisplayProof
  \end{equation*}
  has conclusion $\Gamma \uplus \Delta \vdash \tm\isub\var\val \colon\! Q$ and is such that $\size{\pi''} = 0 = \size{\pi} + \size{\pi'}$.
  
  If $\tm = \tmtwo\tmthree$, then $\tm\isub\var\val = \tmtwo\isub\var\val\tmthree\isub\var\val$ and
  \begin{equation*}
      \pi = 
      \AxiomC{$\ \vdots\,\pi_1$}
      \noLine
      \UnaryInfC{$\Gamma_1, \var \colon\! P_1 \vdash \tmtwo \colon\! [\Pair{Q_2}{Q}]$}
      \AxiomC{$\ \vdots\,\pi_2$}
      \noLine
      \UnaryInfC{$\Gamma_2, \var \colon\! P_2 \vdash \tmthree \colon\! Q_2$}
      \RightLabel{\footnotesize$@$}
      \BinaryInfC{$\Gamma, \var \colon\! P \vdash \tm \colon\! Q$}
      \DisplayProof
  \end{equation*}
  with $\size{\pi} = \size{\pi_1} + \size{\pi_2} + 1$, $\Gamma = \Gamma_1 \uplus \Gamma_2$ and $P = P_1 \uplus P_2$.
  According to \reflemma{value}, there are environments $\Delta_1, \Delta_2$ and derivations $\concl{\pi_1'}{\Delta_1}{\val}{P_1}$ and $\concl{\pi_2'}{\Delta_2}{\val}{P_2}$ such that $\Delta = \Delta_1 \uplus \Delta_2$ and $\size{\pi'} = \size{\pi_1'} + \size{\pi_2'}$.
  By induction hypothesis, there are derivations $\pi_1''$ and $\pi_2''$ with conclusion $\Gamma_1 \uplus \Delta_1 \vdash \tmtwo\isub\var\val \colon\! [\Pair{Q_2}{Q}]$ and $\Gamma_2 \uplus \Delta_2 \vdash \tmthree\isub\var\val \colon\! Q_2$, respectively, such that $\size{\pi_1''} = \size{\pi_1} + \size{\pi_1'}$ and $\size{\pi_2''} = \size{\pi_2} + \size{\pi_2'}$.
  As $\Gamma \uplus \Delta = \Gamma_1 \uplus \Delta_1 \uplus \Gamma_2 \uplus \Delta_2$, there is a derivation 
  \begin{equation*}
      \pi'' = 
      \AxiomC{$\ \vdots\,\pi_1''$}
      \noLine
      \UnaryInfC{$\Gamma_1 \uplus \Delta_1 \vdash \tmtwo\isub\var\val \colon\! [\Pair{Q_2}{Q}]$}
      \AxiomC{$\ \vdots\,\pi_2''$}
      \noLine
      \UnaryInfC{$\Gamma_2 \uplus \Delta_2 \vdash \tmthree\isub\var\val \colon\! Q_2$}
      \RightLabel{\footnotesize$@$}
      \BinaryInfC{$\Gamma \uplus \Delta \vdash \tm\isub\var\val \colon\! Q$}
      \DisplayProof
  \end{equation*}
  where $\size{\pi''} = \size{\pi_1''} + \size{\pi_2''} + 1 =  \size{\pi_1} + \size{\pi_1'} + \size{\pi_2} + \size{\pi_2'} + 1 = \size{\pi} + \size{\pi'} + 1$.
  
  If $\tm = \la\vartwo\tmtwo$, then we can suppose without loss of generality that $\vartwo \notin \Fv{\val} \cup \{\var\}$, therefore $\tm\isub\var\val = \la\vartwo\tmtwo\isub\var\val$ and there are $n \in \nat$, environments $\Gamma_1, \dots, \Gamma_n$, positive types $P_1, Q_1, Q_1', \dots, P_n, Q_n, Q_n'$ such that $\Gamma = \biguplus_{i=1}^n \Gamma_i$ and $P = \biguplus_{i = 1}^n P_i$ and $Q = \biguplus_{i=1}^n [\Pair{Q_i}{Q_i'}]$ and 
  \begin{equation*}
    \pi = 
    \AxiomC{$\ \vdots\,\pi_1$}
    \noLine
    \UnaryInfC{$\Gamma_1, \vartwo \colon\! Q_1, \var \colon\! P_1 \vdash \tmtwo \colon\! Q_1'$}
    \AxiomC{$\overset{n}{\dots}$}
    \AxiomC{$\ \vdots\,\pi_n$}
    \noLine
    \UnaryInfC{$\Gamma_n, \vartwo \colon\! Q_n, \var \colon\! P_n \vdash \tmtwo \colon\! Q_n'$}
    \RightLabel{\footnotesize$\lambda$}
    \TrinaryInfC{$\Gamma, \var \colon\! P \vdash \tm \colon\! Q$}
    \DisplayProof
  \end{equation*}
  with $\size{\pi} = \sum_{i=1}^n \size{\pi_i}$.
  By applying \reflemma{value} to $\pi'$ (as $P = \biguplus_{i=1}^n P_i$), for all $1 \leq i \leq n$ there are an environment $\Delta_i$ and a derivation $\concl{\pi_i'}{\Delta_i}{\val}{P_i}$ such that $\Delta = \biguplus_{i=1}^n \Delta_i$ and $\size{\pi'} = \sum_{i=1}^n \size{\pi_i'}$.
  By induction hypothesis, for all $1 \leq i \leq n$, there is a derivation $\concl{\pi_i''}{\Gamma_i \uplus \Delta_i, \vartwo \colon\! Q_i}{\tmtwo\isub\var\val}{Q_i'}$ such that $\size{\pi_i''} = \size{\pi_i} + \size{\pi_i'}$.
  Since $\Gamma \uplus \Delta = \biguplus_{i=1}^n \Gamma_i \uplus \Delta_i$, there is a derivation
  \begin{equation*}
    \pi'' = 
    \AxiomC{$\ \vdots\,\pi_1''$}
    \noLine
    \UnaryInfC{$\Gamma_1 \uplus \Delta_1, \vartwo \colon\! Q_1 \vdash \tmtwo\isub\var\val \colon\! Q_1'$}
    \AxiomC{$\overset{n}{\dots}$}
    \AxiomC{$\ \vdots\,\pi_n''$}
    \noLine
    \UnaryInfC{$\Gamma_n \uplus \Delta_n, \vartwo \colon\! Q_n \vdash \tmtwo\isub\var\val \colon\! Q_n'$}
    \RightLabel{\footnotesize$\lambda$}
    \TrinaryInfC{$\Gamma \uplus \Delta \vdash \tm\isub\var\val \colon\! Q$}
    \DisplayProof
  \end{equation*}
  where $\size{\pi''} = \sum_{i=1}^n \size{\pi_i''} = \sum_{i=1}^n \size{\pi_i} + \sum_{i=1}^n \size{\pi_i'} = \size{\pi} + \size{\pi'}$.
  \qedhere
\end{proof}

\setcounter{propositionAppendix}{\value{prop:quant-subject-reduction}}
\begin{propositionAppendix}[Quantitative balanced subject reduction]
\label{propappendix:quant-subject-reduction}
  Let 
\NoteState{prop:quant-subject-reduction}
  $\tm, \tmp \in \Lambda$ and $\concl{\pi}{\Gamma}{\tm}{Q}$.
  \begin{enumerate}
    \item\label{pappendix:quant-subject-reduction-betav} \emph{Shrinkage under $\betavm$-step:} If $\tm \tobvm \tmp$ then $\size{\pi} > 0$ and there exists a derivation $\pi'$ with conclusion $\Gamma \vdash \tmp \colon\! Q$ such that $\size{\pi'} = \size{\pi} - 1$.
    \item\label{pappendix:quant-subject-reduction-sigma} \emph{Size invariance under $\sigm$-step:} If $\tm \tosigm \tmp$ then $\size{\pi} > 0$ and there exists a derivation $\pi'$ with conclusion $\Gamma \vdash \tmp \colon\! Q$ such that $\size{\pi'} = \size{\pi}$.
  \end{enumerate}
\end{propositionAppendix}

\begin{proof}
  \begin{enumerate}
    \item Since $\tm$ is not $\betav$-normal, $\tm$ is not a value and thus $\size{\pi} > 0$ according to \refrmk{positive-size}.
    The proof that there exists a derivation $\concl{\pi'}{\Gamma}{\tmp}{Q}$ such that $\size{\pi'} = \size{\pi} - 1$ is by induction on $\tm \in \Lambda$.
    Cases:
    \begin{itemize}
      \item \emph{Step at the root}, \ie~$\tm = (\la\var\tmtwo)\val \rtobv \tmtwo\isub\var\val = \tmp$: then, 
      \begin{equation*}
        \pi =
        \AxiomC{$\ \vdots\, \pi_1$}
        \noLine
        \UnaryInfC{$\Gamma_1, \var \colon\! P \vdash \tmtwo \colon\! Q$}
        \RightLabel{\footnotesize$\lambda$}
        \UnaryInfC{$\Gamma_1 \vdash \la\var\tmtwo \colon\! [\Pair{P}{Q}]$}
        \AxiomC{$\ \vdots\, \pi_2$}
        \noLine
        \UnaryInfC{$\Gamma_2 \vdash \val \colon\! P$}
        \RightLabel{\footnotesize$@$}
        \BinaryInfC{$\Gamma \vdash \tm \colon\! Q$}
        \DisplayProof
      \end{equation*}
      where $\Gamma = \Gamma_1 \uplus \Gamma_2$ and $\size{\pi} = \size{\pi_1} + \size{\pi_2} + 1$. 
      By the substitution lemma (\reflemma{substitution}), there exists a derivation $\concl{\pi'}{\Gamma}{\tmp}{Q}$ such that $\size{\pi'} = \size{\pi_1} + \size{\pi_2} = \size{\pi} - 1$.
      
      \item \emph{Application Left}, \ie~$\tm = \tmtwo\tmthree \tobvm \tmtwop\tmthree = \tmp$ with $\tmtwo \tobvm \tmtwop$: then,
      \begin{equation*}
        \pi =
        \AxiomC{$\ \vdots\, \pi_1$}
        \noLine
	\UnaryInfC{$\Gamma_1 \vdash \tmtwo \colon\! [\Pair{P}{Q}]$}
        \AxiomC{$\ \vdots\, \pi_2$}
        \noLine
        \UnaryInfC{$\Gamma_2 \vdash \tmthree \colon\! P$}
        \RightLabel{\footnotesize$@$}
        \BinaryInfC{$\Gamma \vdash \tm \colon\! Q$}
        \DisplayProof
      \end{equation*}
      where $\Gamma = \Gamma_1 \uplus \Gamma_2$ and $\size{\pi} = \size{\pi_1} + \size{\pi_2} + 1$. 
      By induction hypothesis, there exists a derivation $\concl{\pi_1'}{\Gamma_1}{\tmtwop}{[\Pair{P}{Q}]}$ such that $\size{\pi_1'} = \size{\pi_1} - 1$.
      Therefore, there exists a derivation 
      \begin{equation*}
        \pi' =
        \AxiomC{$\ \vdots\, \pi_1'$}
        \noLine
	\UnaryInfC{$\Gamma_1 \vdash \tmtwop \colon\! [\Pair{P}{Q}]$}
        \AxiomC{$\ \vdots\, \pi_2$}
        \noLine
        \UnaryInfC{$\Gamma_2 \vdash \tmthree \colon\! P$}
        \RightLabel{\footnotesize$@$}
        \BinaryInfC{$\Gamma \vdash \tmp \colon\! Q$}
        \DisplayProof
      \end{equation*}
      where $\size{\pi'} = \size{\pi_1'} + \size{\pi_2} + 1 = \size{\pi_1} + \size{\pi_2} + 1 -1 = \size{\pi} - 1$.
      
      \item \emph{Application Right}, \ie~$\tm = \tmtwo\tmthree \tobvm \tmtwo\tmthreep = \tmp$ with $\tmthree \tobvm \tmthreep$: then,
      \begin{equation*}
        \pi =
        \AxiomC{$\ \vdots\, \pi_1$}
        \noLine
	\UnaryInfC{$\Gamma_1 \vdash \tmtwo \colon\! [\Pair{P}{Q}]$}
        \AxiomC{$\ \vdots\, \pi_2$}
        \noLine
        \UnaryInfC{$\Gamma_2 \vdash \tmthree \colon\! P$}
        \RightLabel{\footnotesize$@$}
        \BinaryInfC{$\Gamma \vdash \tm \colon\! Q$}
        \DisplayProof
      \end{equation*}
      where $\Gamma = \Gamma_1 \uplus \Gamma_2$ and $\size{\pi} = \size{\pi_1} + \size{\pi_2} + 1$. 
      By induction hypothesis, there is a derivation $\concl{\pi_2'}{\Gamma_2}{\tmtwop}{Q}$ such that $\size{\pi_2'} = \size{\pi_2} - 1$.
      Therefore, there exists a derivation 
      \begin{equation*}
        \pi' =
        \AxiomC{$\ \vdots\, \pi_1$}
        \noLine
	\UnaryInfC{$\Gamma_1 \vdash \tmtwo \colon\! [\Pair{P}{Q}]$}
        \AxiomC{$\ \vdots\, \pi_2'$}
        \noLine
        \UnaryInfC{$\Gamma_2 \vdash \tmthreep \colon\! P$}
        \RightLabel{\footnotesize$@$}
        \BinaryInfC{$\Gamma \vdash \tmp \colon\! Q$}
        \DisplayProof
      \end{equation*}
      where $\size{\pi'} = \size{\pi_1} + \size{\pi_2'} + 1 = \size{\pi_1} + \size{\pi_2} + 1 -1 = \size{\pi} - 1$.

      \item \emph{Step inside a $\beta$-redex}, \ie~$\tm = (\la\var\tmtwo)\tmthree \tobvm (\la\var\tmtwop)\tmthree = \tmp$ with $\tmtwo \tobvm \tmtwop$: then, 
      \begin{equation*}
        \pi =
        \AxiomC{$\ \vdots\, \pi_1$}
        \noLine
        \UnaryInfC{$\Gamma_1, \var \colon\! P \vdash \tmtwo \colon\! Q$}
        \RightLabel{\footnotesize$\lambda$}
        \UnaryInfC{$\Gamma_1 \vdash \la\var\tmtwo \colon\! [\Pair{P}{Q}]$}
        \AxiomC{$\ \vdots\, \pi_2$}
        \noLine
        \UnaryInfC{$\Gamma_2 \vdash \tmthree \colon\! P$}
        \RightLabel{\footnotesize$@$}
        \BinaryInfC{$\Gamma \vdash \tm \colon\! Q$}
        \DisplayProof
      \end{equation*}
      where $\Gamma = \Gamma_1 \uplus \Gamma_2$ and $\size{\pi} = \size{\pi_1} + \size{\pi_2} + 1 > 0$. 
      By induction hypothesis, there exists a derivation $\concl{\pi_1'}{\Gamma_1}{\tmtwop}{[\Pair{P}{Q}]}$ such that $\size{\pi_1'} = \size{\pi_1} - 1$.
      Therefore, there is a derivation 
      \begin{equation*}
        \pi' =
        \AxiomC{$\ \vdots\, \pi_1'$}
        \noLine
        \UnaryInfC{$\Gamma_1, \var \colon\! P \vdash \tmtwop \colon\! Q$}
        \RightLabel{\footnotesize$\lambda$}
        \UnaryInfC{$\Gamma_1 \vdash \la\var\tmtwop \colon\! [\Pair{P}{Q}]$}
        \AxiomC{$\ \vdots\, \pi_2$}
        \noLine
        \UnaryInfC{$\Gamma_2 \vdash \tmthree \colon\! P$}
        \RightLabel{\footnotesize$@$}
        \BinaryInfC{$\Gamma \vdash \tmp \colon\! Q$}
        \DisplayProof
      \end{equation*}
      where $\size{\pi'} = \size{\pi_1'} + \size{\pi_2} + 1 = \size{\pi_1} + \size{\pi_2} + 1 - 1 = \size{\pi} - 1$. 
    \end{itemize}

    \item Since $\tm$ is not $\sigma$-normal, $\tm$ is not a value and thus $\size{\pi} > 0$ according to \refrmk{positive-size}.
    The proof that there exists a derivation $\pi'$ with conclusion $\Gamma \vdash \tmp \colon\! Q$ such that $\size{\pi'} = \size{\pi}$ is by induction con $\tm \in \Lambda$.
    Cases:
    \begin{itemize}
      \item \emph{Step at the root}: there are two sub-cases:
      \begin{itemize}
        \item $\tm = (\la\var\tmtwo)\tmthree\tmfour \rtosl (\la\var\tmtwo\tmfour)\tmthree = \tmp$ with $\var \notin \Fv{\tmfour}$: then, 
	\begin{equation*}
	  \pi =
	  \AxiomC{$\ \vdots\, \pi_1$}
	  \noLine
	  \UnaryInfC{$\Gamma_1, \var \colon\! P \vdash \tmtwo \colon\! [\Pair{Q'}{Q}]$}
	  \RightLabel{\footnotesize$\lambda$}
	  \UnaryInfC{$\Gamma_1 \vdash \la\var\tmtwo \colon\! [\Pair{P}{[\Pair{Q'}{Q}]}]$}
	  \AxiomC{$\ \vdots\, \pi_2$}
	  \noLine
	  \UnaryInfC{$\Gamma_2 \vdash \tmthree \colon\! P$}
	  \RightLabel{\footnotesize$@$}
	  \BinaryInfC{$\Gamma_1 \uplus \Gamma_2 \vdash \tm \colon\! [\Pair{Q'}{Q}]$}
	  \AxiomC{$\ \vdots\, \pi_3$}
	  \noLine
	  \UnaryInfC{$\Gamma_3 \vdash \tmfour \colon\! Q'$}
	  \RightLabel{\footnotesize$@$}
	  \BinaryInfC{$\Gamma \vdash \tm \colon\! Q$}
	  \DisplayProof
	\end{equation*}
	with  $\Gamma = \Gamma_1 \uplus \Gamma_2 \uplus \Gamma_3$ and $\size{\pi} = \size{\pi_1} + \size{\pi_2} + \size{\pi_3} + 2$. 
	By \reflemma{free}, $\var \notin \Dom{\Gamma_3}$.
	Therefore, there is a derivation 
	\begin{equation*}
	  \pi' =
	  \AxiomC{$\ \vdots\, \pi_1$}
	  \noLine
	  \UnaryInfC{$\Gamma_1, \var \colon\! P \vdash \tmtwo \colon\! [\Pair{Q'}{Q}]$}
	  \AxiomC{$\ \vdots\, \pi_3$}
	  \noLine
	  \UnaryInfC{$\Gamma_3 \vdash \tmfour \colon\! Q'$}
	  \RightLabel{\footnotesize$@$}
	  \BinaryInfC{$\Gamma_1 \uplus \Gamma_3, \var \colon\! P \vdash \tmtwo\tmfour \colon\! Q$}
	  \RightLabel{\footnotesize$\lambda$}
	  \UnaryInfC{$\Gamma_1 \uplus \Gamma_3 \vdash \la\var{\tmtwo\tmfour} \colon\! [\Pair{P}{Q}]$}
	  \AxiomC{$\ \vdots\, \pi_2$}
	  \noLine
	  \UnaryInfC{$\Gamma_2 \vdash \tmthree \colon\! P$}
	  \RightLabel{\footnotesize$@$}
	  \BinaryInfC{$\Gamma \vdash \tmp \colon\! {Q}$}
	  \DisplayProof
	\end{equation*}
	where $\size{\pi'} = \size{\pi_1} + \size{\pi_3} + 1 + \size{\pi_2} +1 = \size{\pi}$.

	\item $\tm = \val((\la\var\tmtwo)\tmthree) \rtosr (\la\var\val\tmtwo)\tmthree = \tmp$ with $\var \notin \Fv{\val}$: then,
	\begin{equation*}
	  \pi =
	  \AxiomC{$\ \vdots\, \pi_1$}
	  \noLine
	  \UnaryInfC{$\Gamma_1 \vdash \val \colon\! [\Pair{Q'}{Q}]$}
	  \AxiomC{$\ \vdots\, \pi_2$}
	  \noLine
	  \UnaryInfC{$\Gamma_2, \var \colon\! P \vdash \tmtwo \colon\! Q'$}
	  \RightLabel{\footnotesize$\lambda$}
	  \UnaryInfC{$\Gamma_2 \vdash \la\var\tmtwo \colon\! [\Pair{P}{Q'}]$}
	  \AxiomC{$\ \vdots\, \pi_3$}
	  \noLine
	  \UnaryInfC{$\Gamma_3 \vdash \tmthree \colon\! P$}
	  \RightLabel{\footnotesize$@$}
	  \BinaryInfC{$\Gamma_2 \uplus \Gamma_3 \vdash (\la\var\tmtwo)\tmthree \colon\! Q'$}
	  \RightLabel{\footnotesize$@$}
	  \BinaryInfC{$\Gamma \vdash \tm \colon\! Q$}
	  \DisplayProof
	\end{equation*}
	with $\Gamma = \Gamma_1 \uplus \Gamma_2 \uplus \Gamma_3$ and $\size{\pi} = \size{\pi_1} + \size{\pi_2} + \size{\pi_3} + 2$. 
		By \reflemma{free}, $\var \notin \Dom{\Gamma_1}$.
	So, there is a derivation 
	\begin{equation*}
	  \pi' =
	  \AxiomC{$\ \vdots\, \pi_1$}
	  \noLine
	  \UnaryInfC{$\Gamma_1 \vdash \val \colon\! [\Pair{Q'}{Q}]$}
	  \AxiomC{$\ \vdots\, \pi_3$}
	  \noLine
	  \UnaryInfC{$\Gamma_3, \var \colon\! P \vdash \tmtwo \colon\! Q'$}
	  \RightLabel{\footnotesize$@$}
	  \BinaryInfC{$\Gamma_1 \uplus \Gamma_3, \var \colon\! P \vdash \val\tmtwo \colon\! Q$}
	  \RightLabel{\footnotesize$\lambda$}
	  \UnaryInfC{$\Gamma_1 \uplus \Gamma_3 \vdash \la\var{\tmtwo\tmfour} \colon\! [\Pair{P}{Q}]$}
	  \AxiomC{$\ \vdots\, \pi_2$}
	  \noLine
	  \UnaryInfC{$\Gamma_2 \vdash \tmthree \colon\! P$}
	  \RightLabel{\footnotesize$@$}
	  \BinaryInfC{$\Gamma \vdash \tmp \colon\! {Q}$}
	  \DisplayProof
	\end{equation*}
	where $\size{\pi'} =  \size{\pi_1} + \size{\pi_3} + 1 + \size{\pi_2} +1 = \size{\pi}$.
      \end{itemize}

      \item \emph{Application Left}, \ie~$\tm = \tmtwo\tmthree \tosigm \tmtwop\tmthree = \tmp$ with $\tmtwo \tosigm \tmtwop$: then,
      \begin{equation*}
        \pi =
        \AxiomC{$\ \vdots\, \pi_1$}
        \noLine
	\UnaryInfC{$\Gamma_1 \vdash \tmtwo \colon\! [\Pair{P}{Q}]$}
        \AxiomC{$\ \vdots\, \pi_2$}
        \noLine
        \UnaryInfC{$\Gamma_2 \vdash \tmthree \colon\! P$}
        \RightLabel{\footnotesize$@$}
        \BinaryInfC{$\Gamma \vdash \tm \colon\! Q$}
        \DisplayProof
      \end{equation*}
      where $\Gamma = \Gamma_1 \uplus \Gamma_2$ and $\size{\pi} = \size{\pi_1} + \size{\pi_2} + 1$. 
      By induction hypothesis, there exists a derivation $\concl{\pi_1'}{\Gamma_1}{\tmtwop}{[\Pair{P}{Q}]}$ such that $\size{\pi_1'} = \size{\pi_1}$.
      Therefore, there exists a derivation 
      \begin{equation*}
        \pi' =
        \AxiomC{$\ \vdots\, \pi_1'$}
        \noLine
	\UnaryInfC{$\Gamma_1 \vdash \tmtwop \colon\! [\Pair{P}{Q}]$}
        \AxiomC{$\ \vdots\, \pi_2$}
        \noLine
        \UnaryInfC{$\Gamma_2 \vdash \tmthree \colon\! P$}
        \RightLabel{\footnotesize$@$}
        \BinaryInfC{$\Gamma \vdash \tmp \colon\! Q$}
        \DisplayProof
      \end{equation*}
      where $\size{\pi'} = \size{\pi_1'} + \size{\pi_2} + 1 = \size{\pi_1} + \size{\pi_2} + 1 = \size{\pi}$.
      
      \item \emph{Application Right}, \ie~$\tm = \tmtwo\tmthree \tosigm \tmtwo\tmthreep = \tmp$ with $\tmthree \tosigm \tmthreep$: then,
      \begin{equation*}
        \pi =
        \AxiomC{$\ \vdots\, \pi_1$}
        \noLine
	\UnaryInfC{$\Gamma_1 \vdash \tmtwo \colon\! [\Pair{P}{Q}]$}
        \AxiomC{$\ \vdots\, \pi_2$}
        \noLine
        \UnaryInfC{$\Gamma_2 \vdash \tmthree \colon\! P$}
        \RightLabel{\footnotesize$@$}
        \BinaryInfC{$\Gamma \vdash \tm \colon\! Q$}
        \DisplayProof
      \end{equation*}
      where $\Gamma = \Gamma_1 \uplus \Gamma_2$ and $\size{\pi} = \size{\pi_1} + \size{\pi_2} + 1$. 
      By induction hypothesis, there exists a derivation $\concl{\pi_2'}{\Gamma_2}{\tmtwop}{Q}$ such that $\size{\pi_2'} = \size{\pi_2}$.
      Therefore, there exists a derivation 
      \begin{equation*}
        \pi' =
        \AxiomC{$\ \vdots\, \pi_1$}
        \noLine
	\UnaryInfC{$\Gamma_1 \vdash \tmtwo \colon\! [\Pair{P}{Q}]$}
        \AxiomC{$\ \vdots\, \pi_2'$}
        \noLine
        \UnaryInfC{$\Gamma_2 \vdash \tmthreep \colon\! P$}
        \RightLabel{\footnotesize$@$}
        \BinaryInfC{$\Gamma \vdash \tmp \colon\! Q$}
        \DisplayProof
      \end{equation*}
      where $\size{\pi'} = \size{\pi_1} + \size{\pi_2'} + 1 = \size{\pi_1} + \size{\pi_2} + 1 = \size{\pi}$.

      \item \emph{Step inside a $\beta$-redex}, \ie~$\tm = (\la\var\tmtwo)\tmthree \tosigm (\la\var\tmtwop)\tmthree = \tmp$ with $\tmtwo \tosigm \tmtwop$: then, 
      \begin{equation*}
        \pi =
        \AxiomC{$\ \vdots\, \pi_1$}
        \noLine
        \UnaryInfC{$\Gamma_1, \var \colon\! P \vdash \tmtwo \colon\! Q$}
        \RightLabel{\footnotesize$\lambda$}
        \UnaryInfC{$\Gamma_1 \vdash \la\var\tmtwo \colon\! [\Pair{P}{Q}]$}
        \AxiomC{$\ \vdots\, \pi_2$}
        \noLine
        \UnaryInfC{$\Gamma_2 \vdash \tmthree \colon\! P$}
        \RightLabel{\footnotesize$@$}
        \BinaryInfC{$\Gamma \vdash \tm \colon\! Q$}
        \DisplayProof
      \end{equation*}
      where $\Gamma = \Gamma_1 \uplus \Gamma_2$ and $\size{\pi} = \size{\pi_1} + \size{\pi_2} + 1$. 
      By induction hypothesis, there exists a derivation $\concl{\pi_1'}{\Gamma_1}{\tmtwop}{[\Pair{P}{Q}]}$ such that $\size{\pi_1'} = \size{\pi_1}$.
      Therefore, there is a derivation 
      \begin{equation*}
        \pi' =
        \AxiomC{$\ \vdots\, \pi_1'$}
        \noLine
        \UnaryInfC{$\Gamma_1, \var \colon\! P \vdash \tmtwop \colon\! Q$}
        \RightLabel{\footnotesize$\lambda$}
        \UnaryInfC{$\Gamma_1 \vdash \la\var\tmtwop \colon\! [\Pair{P}{Q}]$}
        \AxiomC{$\ \vdots\, \pi_2$}
        \noLine
        \UnaryInfC{$\Gamma_2 \vdash \tmthree \colon\! P$}
        \RightLabel{\footnotesize$@$}
        \BinaryInfC{$\Gamma \vdash \tmp \colon\! Q$}
        \DisplayProof
      \end{equation*}
      where $\size{\pi'} = \size{\pi_1'} + \size{\pi_2} + 1 = \size{\pi_1} + \size{\pi_2} + 1 = \size{\pi}$. 
      \qedhere
    \end{itemize}
  \end{enumerate}

\end{proof}

\setcounter{lemmaAppendix}{\value{l:commutation}}
\begin{lemmaAppendix}[Abstraction commutation]\hfill
\label{lappendix:commutation}
\NoteState{l:commutation}
  \begin{enumerate}
    \item\label{pappendix:commutation-abstraction} \emph{Abstraction vs.~abstraction:}
    Let $k \in \nat$.
    If $\concl{\pi}{\Delta}{\la{\vartwo}{(\la{\var}{\tm})\val}}{\biguplus_{i=1}^k [\Pair{P_i'}{P_i}]}$ and $\vartwo \notin \Fv{\val}$, then there is $\concl{\pi'}{\Delta}{(\la{\var}{\la{\vartwo}{\tm}})\val}{\biguplus_{i=1}^k [\Pair{P_i'}{P_i}]}$ such that $\size{\pi'} = \size{\pi} + 1 - k$.

    \item \label{pappendix:commutation-application}\emph{Application vs.~abstraction:}
    If $\concl{\pi}{\Delta}{((\la{\var}{\tm})\val)((\la{\var}{\tmtwo})\val)}{P}$ then there exists a derivation $\concl{\pi'}{\Delta}{(\la{\var}{\tm\tmtwo})\val}{P}$ such that $\size{\pi'} = \size{\pi} - 1$.
  \end{enumerate}
\end{lemmaAppendix}

\begin{proof}
  \begin{enumerate}
    \item  There are derivations $\pi_1^1,\pi_1^2, \dots, \pi_k^1,\pi_k^2$, environments $\Gamma_1, \Delta_1, \dots, \Gamma_k, \Delta_k$ and positive types $Q_i, \dots, Q_k$ such that $\vartwo \notin \Dom{\Delta_i}$ for all $1 \leq i \leq k$ (by \reflemma{free}, since $\vartwo \notin \Fv{\val}$) and
    \begin{equation*}
      \pi =
      \AxiomC{$\ \vdots\, \pi_i^1$}
      \noLine
      \UnaryInfC{$\Gamma_i, \vartwo \colon P_i', \var \colon\! Q_i \vdash \tm \colon\! P_i$}
      \RightLabel{\footnotesize$\lambda$}
      \UnaryInfC{$\Gamma_i, \vartwo \colon P_i' \vdash \la{\var}{\tm} \colon\! [\Pair{Q_i}{P_i}]$}
      \AxiomC{$\ \vdots\, \pi_i^2$}
      \noLine
      \UnaryInfC{$\Delta_i \vdash \val \colon\! Q_i$}
      \RightLabel{\footnotesize$@$}
      \BinaryInfC{$\Gamma_i \uplus \Delta_i, \vartwo \colon P_i' \vdash (\la{\var}{\tm})\val \colon\! P_i$}
      \AxiomC{\textup{(for all $1 \leq i \leq n$)}}
      \RightLabel{\footnotesize$\lambda$}
      \BinaryInfC{$\Delta \vdash \la{\vartwo}{(\la{\var}{\tm})\val} \colon\! \biguplus_{i=1}^k [\Pair{P_i'}{P_i$}]}
      \DisplayProof
    \end{equation*}
    where $\Delta = \biguplus_{i=1}^k \Gamma_i \uplus \Delta_i$ and $\size{\pi} = \sum_{i=1}^k (\size{\pi_i^1} + \size{\pi_i^2} +1) = k + \sum_{i=1}^k (\size{\pi_i^1} + \size{\pi_i^2})$.
    According to \reflemma{value}, there exists $\concl{\pi_2}{\Delta'}{\val}{\biguplus_{i=1}^k Q_i}$ with $\Delta' = \biguplus_{i=1}^k \Delta_i$ (whence $\Delta = \Delta' \uplus \biguplus_{i=1}^k \Gamma_i$) and $\size{\pi_2} = \sum_{i=1}^k \size{\pi_i^2}$, thus one has
    \begin{equation*}
      \pi' =
      \AxiomC{$\ \vdots\, \pi_i^1$}
      \noLine
      \UnaryInfC{$\Gamma_i, \vartwo \colon\! P_i', \var \colon\! Q_i \vdash \tm \colon\! P_i$}
      \AxiomC{\textup{(for all $1 \leq i \leq n$)}}
      \RightLabel{\footnotesize$\lambda$}
      \BinaryInfC{$\biguplus_{i=1}^n \Gamma_i, \var \colon\! \biguplus_{i=1}^n Q_i \vdash \la{\vartwo}{\tm} \colon\! \biguplus_{i=1}^n [\Pair{P_i'}{P_i}]$}
      \RightLabel{\footnotesize$\lambda$}
      \UnaryInfC{$\biguplus_{i=1}^n \Gamma_i \vdash \la{\var}{\la\vartwo\tm} \colon\! [\Pair{\biguplus_{i=1}^n Q_i}{\biguplus_{i=n}^n[\Pair{P_i'}{P_i}]}]$}
      \AxiomC{$\ \vdots\, \pi_2$}
      \noLine
      \UnaryInfC{$\Delta' \vdash \val \colon\! \biguplus_{i=1}^n Q_i$}
      \RightLabel{\footnotesize$@$}
      \BinaryInfC{$\Delta \vdash (\la{\var}{\la\vartwo\tm})\val \colon\! \biguplus_{i=n}^n [\Pair{P_i'}{P_i}]$}
      \DisplayProof
    \end{equation*}
    where $\size{\pi'} = \size{\pi_2} + 1 + \sum_{i=1}^n \size{\pi_i^1} = 1+ \sum_{i=1}^n (\size{\pi_i^1} + \size{\pi_i^2}) = \size{\pi} +1 - n$.

    \item There are environments $\Delta_1, \Delta_2, \Delta_3, \Delta_4$, positive types $Q, P_1, P_2$ and derivations $\pi_1, \pi_2, \pi_3, \pi_4$ such that
    
    {\footnotesize
    \begin{equation*}
      \pi =
      \AxiomC{$\ \vdots\, \pi_3$}
      \noLine
      \UnaryInfC{$\Delta_3, \var \colon\! P_1 \vdash \tm \colon\! [\Pair{Q}{P}]$}
      \RightLabel{\footnotesize$\lambda$}
      \UnaryInfC{$\Delta_3 \vdash \la\var\tmtwo \colon\! [\Pair{P_1}{[\Pair{Q}{P}]}]$}
      \AxiomC{$\ \vdots\, \pi_1$}
      \noLine
      \UnaryInfC{$\Delta_1 \vdash \val \colon\! P_1$}
      \RightLabel{\footnotesize$@$}
      \BinaryInfC{$\Delta_3 \uplus \Delta_1 \vdash (\la{\var}{\tm})\val \colon\! [\Pair{Q}{P}]$}
      \AxiomC{$\ \vdots\, \pi_4$}
      \noLine
      \UnaryInfC{$\Delta_4, \var \colon\! P_2 \vdash \tmtwo \colon\! Q$}
      \RightLabel{\footnotesize$\lambda$}
      \UnaryInfC{$\Delta_4 \vdash \la\var\tmtwo \colon\! [\Pair{P_2}{Q}]$}
      \AxiomC{$\ \vdots\, \pi_2$}
      \noLine
      \UnaryInfC{$\Delta_2 \vdash \val \colon\! P_2$}
      \RightLabel{\footnotesize$@$}
      \BinaryInfC{$\Delta_4 \uplus \Delta_2 \vdash (\la{\var}{\tmtwo})\val \colon\! Q$}
      \RightLabel{\footnotesize$@$}
      \BinaryInfC{$\Delta \vdash ((\la{\var}{\tm})\val)((\la{\var}{\tmtwo})\val) \colon\! P$}
      \DisplayProof
    \end{equation*}
    }
    
    \noindent where $\Delta =  \Delta_1 \uplus \Delta_2 \uplus \Delta_3 \uplus \Delta_4$ and $\size{\pi} = \size{\pi_1} + \size{\pi_2} + \size{\pi_3} + \size{\pi_4} + 3$.
    According to \reflemma{value}, there is $\concl{\pi_0}{\Delta_1 \uplus \Delta_2}{\val}{P_1 \uplus P_2}$ such that $\size{\pi_0} = \size{\pi_1} + \size{\pi_2}$, thus there exists
    \begin{equation*}
      \pi' = 
      \AxiomC{$\ \vdots\, \pi_3$}
      \noLine
      \UnaryInfC{$\Delta_3, \var \colon\! P_1 \vdash \tm \colon\! [\Pair{Q}{P}]$}
      \AxiomC{$\ \vdots\, \pi_4$}
      \noLine
      \UnaryInfC{$\Delta_4, \var \colon\! P_2 \vdash \tmtwo \colon\! Q$}
      \RightLabel{\footnotesize$@$}
      \BinaryInfC{$\Delta_3 \uplus \Delta_4, \var \colon\! P_1 \uplus P_2 \vdash \tm\tmtwo \colon\! P$}
      \RightLabel{\footnotesize$\lambda$}
      \UnaryInfC{$\Delta_3 \uplus \Delta_4 \vdash \la\var{\tm\tmtwo} \colon\! [\Pair{P_1 \uplus P_2}{P}]$}
      \AxiomC{$\ \vdots\, \pi_0$}
      \noLine
      \UnaryInfC{$\Delta_1 \uplus \Delta_2 \vdash \val \colon\! P_1 \uplus P_2$}
      \RightLabel{\footnotesize$@$}
      \BinaryInfC{$\Delta \vdash (\la{\var}{\tm\tmtwo})\val \colon\! P$}
      \DisplayProof
    \end{equation*}
    where $\size{\pi'} = \size{\pi_0} + \size{\pi_3} + \size{\pi_4} + 2 =  \size{\pi_1} + \size{\pi_2} + \size{\pi_3} + \size{\pi_4} + 2 = \size{\pi} - 1$.
    \qedhere
  \end{enumerate}
\end{proof}

\setcounter{propositionAppendix}{\value{prop:quant-subject-expansion}}
\begin{propositionAppendix}[Quantitative balanced subject expansion]
\label{propappendix:quant-subject-expansion}
\NoteState{prop:quant-subject-expansion}
  Let $\tm, \tmp \in \Lambda$ and $\concl{\pi'}{\Gamma}{\tmp}{Q}$.
  \begin{enumerate}
    \item\label{pappendix:quant-subject-expansion-betav}\emph{Enlargement under anti-$\betavm$-step:} If $\tm \tobvm \tmp$ then there is $\concl{\pi}{\Gamma}{\tm}{Q}$ with $\size{\pi} = \size{\pi'} + 1$.
    \item\label{pappendix:quant-subject-expansion-sigma}\emph{Size invariance under anti-$\sigm$-step:} If $\tm \tosigm \tmp$ then $\size{\pi'} > 0$ and there is $\concl{\pi}{\Gamma}{\tm}{Q}$ with $\size{\pi} = \size{\pi'}$.
  \end{enumerate}
\end{propositionAppendix}

\begin{proof}
  \begin{enumerate}
    \item 
    By induction on $\tm \in \Lambda$.
    Cases:
    \begin{itemize}
      \item \emph{Step at the root}, \ie~$\tm = (\la\var\tmtwo)\val \rtobv \tmtwo\isub\var\val = \tmp$. 
      We proceed by induction on $\tmtwo \in \Lambda$.
      \begin{itemize}
        \item If $\tmtwo = \var$, then $\tmp = \val$ and $\concl{\pi'}{\Gamma}{\val}{Q}$, while $\tm = (\la{\var}{\var})\val$.
        We have the derivation
      \begin{equation*}
        \pi =
        \AxiomC{}
        \RightLabel{\footnotesize$\Ax$}
        \UnaryInfC{$\var \colon\! Q \vdash \var \colon\! Q$}
        \RightLabel{\footnotesize$\lambda$}
        \UnaryInfC{$\vdash \la{\var}{\var} \colon\! [\Pair{Q}{Q}]$}
        \AxiomC{$\ \vdots\, \pi'$}
        \noLine
        \UnaryInfC{$\Gamma \vdash \val \colon\! Q$}
        \RightLabel{\footnotesize$@$}
        \BinaryInfC{$\Gamma \vdash \tm \colon\! Q$}
        \DisplayProof
      \end{equation*}
      with $\size{\pi} = \size{\pi'} + 1$.
        
        \item If $\tmtwo = \vartwo \neq \var$ (we can suppose without loss of generality that $\var \notin \Fv{\val}$), then $\tmp = \vartwo$ and $\pi' = 
        \AxiomC{}
        \RightLabel{\footnotesize$\Ax$}
        \UnaryInfC{$\vartwo \colon\! Q \vdash \vartwo \colon\! Q$}\DisplayProof$
        with $\Gamma = \vartwo \colon\! Q$, while $\tm = (\la{\var}{\vartwo})\val$.
        Notice that $\size{\pi'} = 0$.
        We have:
	\begin{equation*}
	  \pi =
	  \AxiomC{}
	  \RightLabel{\footnotesize$\Ax$}
	  \UnaryInfC{$ \var \colon\! \emptymset, \vartwo \colon\! Q \vdash \vartwo \colon\! Q$}
	  \RightLabel{\footnotesize$\lambda$}
	  \UnaryInfC{$\vartwo \colon\! Q \vdash \la{\var}{\vartwo} \colon\! [\Pair{\emptymset}{Q}]$}
	  \AxiomC{}
	  \RightLabel{\footnotesize$\lambda$}
	  \UnaryInfC{$ \vdash \val \colon\! \emptymset$}
	  \RightLabel{\footnotesize$@$}
	  \BinaryInfC{$\Gamma \vdash \tm \colon\! Q$}
	  \DisplayProof
	\end{equation*}
	(notice that the rule $\lambda$ in $\pi$ has 0 premises) with $\size{\pi} = 1 = \size{\pi'} + 1$.
	
	\item If $\tmtwo = \la{\vartwo}{\tmthree}$ (we can suppose without loss of generality that $\vartwo \notin \Fv{\val} \cup \{\var\}$), then $\tmp = \la{\vartwo}{\tmthree \isub{\var}{\val}}$ and $\tm = (\la{\var}{\la{\vartwo}{\tmthree}})\val$.
	As $\concl{\pi'}{\Gamma}{\tmp}{Q}$,
	there are $n \in \nat$, positive types $P_1, Q_1, \dots, P_n, Q_n$, environments $\Gamma_1, \dots, \Gamma_n$ and derivations $\pi_1', \dots, \pi_n'$ such that $Q = [\Pair{P_1}{Q_1}, \dots, \Pair{P_n}{Q_n}]$ and
	\begin{equation*}
	  \pi' = 
	  \AxiomC{$\ \vdots\, \pi_1'$}
	  \noLine
	  \UnaryInfC{$\Gamma_1, \vartwo \colon\! P_1 \vdash \tmthree\isub{\var}{\val} \colon\! Q_1$}
	  \AxiomC{$\overset{n \in \nat}{\ldots}$}
	  \AxiomC{$\ \vdots\, \pi_n'$}
	  \noLine
	  \UnaryInfC{$\Gamma_n, \vartwo \colon\! P_n \vdash \tmthree\isub{\var}{\val} \colon\! Q_n$}
	  \RightLabel{\footnotesize$\lambda$}
	  \TrinaryInfC{$\Gamma \vdash \tmp \colon\! Q$}
	  \DisplayProof
	\end{equation*}
	where $\Gamma = \biguplus_{i=1}^n \Gamma_i$ and $\size{\pi'} = \sum_{i=1}^n \size{\pi_i'}$.
	Let $1 \leq i \leq n$: since $(\la{\var}{\tmthree})\val \rtobv \tmthree\isub{\var}{\val}$, then by \ih there is $\concl{\pi_i}{\Gamma_i, \vartwo \colon\! P_i}{(\la\var\tmthree)\val}{Q_i}$ with $\size{\pi_i} = \size{\pi_i'} + 1$.
	So, we set
	\begin{equation*}
	  \pi'' = 
	  \AxiomC{$\ \vdots\, \pi_1$}
	  \noLine
	  \UnaryInfC{$\Gamma_1, \vartwo \colon\! P_1 \vdash (\isub{\var}{\tmthree}){\val} \colon\! Q_1$}
	  \AxiomC{$\overset{n \in \nat}{\ldots}$}
	  \AxiomC{$\ \vdots\, \pi_n$}
	  \noLine
	  \UnaryInfC{$\Gamma_n, \vartwo \colon\! P_n \vdash (\la{\var}{\tmthree}){\val} \colon\! Q_n$}
	  \RightLabel{\footnotesize$\lambda$}
	  \TrinaryInfC{$\Gamma \vdash \la{\vartwo}{(\la{\var}{\tmthree})\val} \colon\! Q$}
	  \DisplayProof
	\end{equation*}
	where $\size{\pi''} = \sum_{i=1}^n \size{\pi_i} = \sum_{i=1}^n \size{\pi_i'} + n = \size{\pi'} + n$.
	According to \reflemmap{commutation}{abstraction}, there is a derivation $\concl{\pi}{\Gamma}{\tm}{Q}$ where $\size{\pi} = \size{\pi''} - n + 1 = \size{\pi'} + 1$.
	
	\item Finally, if $\tmtwo = \tmthree\tmfour$, then $\tmp = \tmthree\isub{\var}{\val} \tmfour\isub{\var}{\val}$ and $\tm = (\la{\var}{\tmthree\tmfour})\val$.
	Since $\concl{\pi'}{\Gamma}{\tmp}{Q}$, there are derivations $\pi_1'$ and $\pi_2'$, a positive type $P$, environments $\Gamma_1$ and $\Gamma_2$ (where $\Gamma = \Gamma_1 \uplus \Gamma_2$) such that
	\begin{equation*}
	  \pi' = 
	  \AxiomC{$\ \vdots\, \pi_1'$}
	  \noLine
	  \UnaryInfC{$\Gamma_1 \vdash \tmthree\isub{\var}{\val} \colon\! [\Pair{P}{Q}]$}
	  \AxiomC{$\ \vdots\, \pi_2'$}
	  \noLine
	  \UnaryInfC{$\Gamma_2 \vdash \tmfour\isub{\var}{\val} \colon\! P$}
	  \RightLabel{\footnotesize$@$}
	  \BinaryInfC{$\Gamma \vdash \tmp \colon\! Q$}
	  \DisplayProof
	\end{equation*}
	where $\size{\pi'} = \size{\pi_1'} + \size{\pi_2'} + 1$.
	Since $(\la{\var}{\tmthree})\val \rtobv \tmthree\isub{\var}{\val}$ and $(\la{\var}{\tmfour})\val \rtobv \tmfour\isub{\var}{\val}$, then by \ih there are $\concl{\pi_1}{\Gamma_1}{(\la{\var}{\tmthree})\val}{[\Pair{P}{Q}]}$ and $\concl{\pi_2}{\Gamma_2}{(\la{\var}{\tmfour})\val}{P}$ with $\size{\pi_1} = \size{\pi_1'} + 1$ and $\size{\pi_2} = \size{\pi_2'} + 1$.
	So, we set
	\begin{equation*}
	  \pi'' = 
	  \AxiomC{$\ \vdots\, \pi_1$}
	  \noLine
	  \UnaryInfC{$\Gamma_1 \vdash (\la{\var}{\tmthree})\val \colon\! [\Pair{P}{Q}]$}
	  \AxiomC{$\ \vdots\, \pi_2$}
	  \noLine
	  \UnaryInfC{$\Gamma_2 \vdash (\la{\var}{\tmfour})\val \colon\! P$}
	  \RightLabel{\footnotesize$@$}
	  \BinaryInfC{$\Gamma \vdash ((\la{\var}{\tmthree})\val)((\la{\var}{\tmfour})\val) \colon\! Q$}
	  \DisplayProof
	\end{equation*}
	where $\size{\pi''} = \size{\pi_1} + \size{\pi_2} + 1 = \size{\pi_1'} + \size{\pi_2'} + 3 = \size{\pi'} + 2$.
	According to \reflemmap{commutation}{application}, there is a derivation $\concl{\pi}{\Gamma}{\tm}{Q}$ with $\size{\pi} = \size{\pi''} - 1 = \size{\pi'} + 1$.
      \end{itemize}
      
      \item \emph{Application Left}, \ie~$\tm = \tmtwo\tmthree \tobvm \tmtwop\tmthree = \tmp$ with $\tmtwo \tobvm \tmtwop$: then,
      \begin{equation*}
        \pi' =
        \AxiomC{$\ \vdots\, \pi_1'$}
        \noLine
	\UnaryInfC{$\Gamma_1 \vdash \tmtwop \colon\! [\Pair{P}{Q}]$}
        \AxiomC{$\ \vdots\, \pi_2$}
        \noLine
        \UnaryInfC{$\Gamma_2 \vdash \tmthree \colon\! P$}
        \RightLabel{\footnotesize$@$}
        \BinaryInfC{$\Gamma \vdash \tmp \colon\! Q$}
        \DisplayProof
      \end{equation*}
      where $\Gamma = \Gamma_1 \uplus \Gamma_2$ and $\size{\pi'} = \size{\pi_1'} + \size{\pi_2} + 1$. 
      By induction hypothesis, there is $\concl{\pi_1}{\Gamma}{\tmtwo}{[\Pair{P}{Q}]}$ with $\size{\pi_1} = \size{\pi_1'} + 1$.
      So, there is
      \begin{equation*}
        \pi =
        \AxiomC{$\ \vdots\, \pi_1$}
        \noLine
	\UnaryInfC{$\Gamma_1 \vdash \tmtwo \colon\! [\Pair{P}{Q}]$}
        \AxiomC{$\ \vdots\, \pi_2$}
        \noLine
        \UnaryInfC{$\Gamma_2 \vdash \tmthree \colon\! P$}
        \RightLabel{\footnotesize$@$}
        \BinaryInfC{$\Gamma \vdash \tm \colon\! Q$}
        \DisplayProof
      \end{equation*}
      where $\size{\pi} = \size{\pi_1} + \size{\pi_2} + 1 = \size{\pi_1'} + 1 + \size{\pi_2} + 1 = \size{\pi'} + 1$.
      
      \item \emph{Application Right}, \ie~$\tm = \tmtwo\tmthree \tobvm \tmtwo\tmthreep = \tmp$ with $\tmthree \tobvm \tmthreep$: then,
      \begin{equation*}
        \pi' =
        \AxiomC{$\ \vdots\, \pi_1$}
        \noLine
	\UnaryInfC{$\Gamma_1 \vdash \tmtwo \colon\! [\Pair{P}{Q}]$}
        \AxiomC{$\ \vdots\, \pi_2'$}
        \noLine
        \UnaryInfC{$\Gamma_2 \vdash \tmthreep \colon\! P$}
        \RightLabel{\footnotesize$@$}
        \BinaryInfC{$\Gamma \vdash \tmp \colon\! Q$}
        \DisplayProof
      \end{equation*}
      where $\Gamma = \Gamma_1 \uplus \Gamma_2$ and $\size{\pi'} = \size{\pi_1} + \size{\pi_2'} + 1$. 
      By induction hypothesis, there exists $\concl{\pi_2}{\Gamma_2}{\tmthree}{P}$ with $\size{\pi_2} = \size{\pi_2'} + 1$.
      So, there is
      \begin{equation*}
        \pi =
        \AxiomC{$\ \vdots\, \pi_1$}
        \noLine
	\UnaryInfC{$\Gamma_1 \vdash \tmtwo \colon\! [\Pair{P}{Q}]$}
        \AxiomC{$\ \vdots\, \pi_2$}
        \noLine
        \UnaryInfC{$\Gamma_2 \vdash \tmthree \colon\! P$}
        \RightLabel{\footnotesize$@$}
        \BinaryInfC{$\Gamma \vdash \tm \colon\! Q$}
        \DisplayProof
      \end{equation*}
      where $\size{\pi} = \size{\pi_1} + \size{\pi_2} + 1 = \size{\pi_1} + \size{\pi_2'} +1 + 1 = \size{\pi'} + 1$.

      \item \emph{Step inside a $\beta$-redex}, \ie~$\tm = (\la\var\tmtwo)\tmthree \tobvm (\la\var\tmtwop)\tmthree = \tmp$ with $\tmtwo \tobvm \tmtwop$: then, 
      \begin{equation*}
        \pi' =
        \AxiomC{$\ \vdots\, \pi_1'$}
        \noLine
        \UnaryInfC{$\Gamma_1, \var \colon\! P \vdash \tmtwop \colon\! Q$}
        \RightLabel{\footnotesize$\lambda$}
        \UnaryInfC{$\Gamma_1 \vdash \la\var\tmtwop \colon\! [\Pair{P}{Q}]$}
        \AxiomC{$\ \vdots\, \pi_2$}
        \noLine
        \UnaryInfC{$\Gamma_2 \vdash \tmthree \colon\! P$}
        \RightLabel{\footnotesize$@$}
        \BinaryInfC{$\Gamma \vdash \tmp \colon\! Q$}
        \DisplayProof
      \end{equation*}
      where $\Gamma = \Gamma_1 \uplus \Gamma_2$ and $\size{\pi'} = \size{\pi_1'} + \size{\pi_2} + 1$. 
      By induction hypothesis, there exists $\concl{\pi_1}{\Gamma_1, \var \colon\! P}{\tmtwo}{Q}$ with $\size{\pi_1} = \size{\pi_1'} + 1$.
      So,~there~is 
      \begin{equation*}
        \pi =
        \AxiomC{$\ \vdots\, \pi_1$}
        \noLine
        \UnaryInfC{$\Gamma_1, \var \colon\! P \vdash \tmtwo \colon\! Q$}
        \RightLabel{\footnotesize$\lambda$}
        \UnaryInfC{$\Gamma_1 \vdash \la\var\tmtwo \colon\! [\Pair{P}{Q}]$}
        \AxiomC{$\ \vdots\, \pi_2$}
        \noLine
        \UnaryInfC{$\Delta \vdash \tmthree \colon\! P$}
        \RightLabel{\footnotesize$@$}
        \BinaryInfC{$\Gamma \vdash \tm \colon\! Q$}
        \DisplayProof
      \end{equation*}
      where $\size{\pi} = \size{\pi_1} + \size{\pi_2} + 1 = \size{\pi_1'} + 1 + \size{\pi_2} + 1 = \size{\pi'} + 1$. 
    \end{itemize}

    \item Since $\tosigm$ cannot reduce to a value,
    $\tm'$ is not a value and thus $\size{\pi'} > 0$ according to \refrmk{positive-size}.
    The proof that there exists $\concl{\pi}{\Gamma}{\tm}{Q}$ with $\size{\pi} = \size{\pi'}$ is by induction con $\tm \in \Lambda$.
    Cases:
    \begin{itemize}
      \item \emph{Step at the root}: there are two sub-cases:
      \begin{itemize}
        \item $\tm = (\la\var\tmtwo)\tmthree\tmfour \rtosl (\la\var\tmtwo\tmfour)\tmthree = \tmp$ with $\var \notin \Fv{\tmfour}$. 
        So, 
	\begin{equation*}
	  \pi' =
	  \AxiomC{$\ \vdots\, \pi_1$}
	  \noLine
	  \UnaryInfC{$\Gamma_1, \var \colon\! P \vdash \tmtwo \colon\! [\Pair{Q'}{Q}]$}
	  \AxiomC{$\ \vdots\, \pi_3$}
	  \noLine
	  \UnaryInfC{$\Gamma_3 \vdash \tmfour \colon\! Q'$}
	  \RightLabel{\footnotesize$@$}
	  \BinaryInfC{$\Gamma_1 \uplus \Gamma_3, \var \colon\! P \vdash \tmtwo\tmfour \colon\! Q$}
	  \RightLabel{\footnotesize$\lambda$}
	  \UnaryInfC{$\Gamma_1 \uplus \Gamma_3 \vdash \la\var{\tmtwo\tmfour} \colon\! [\Pair{P}{Q}]$}
	  \AxiomC{$\ \vdots\, \pi_2$}
	  \noLine
	  \UnaryInfC{$\Gamma_2 \vdash \tmthree \colon\! P$}
	  \RightLabel{\footnotesize$@$}
	  \BinaryInfC{$\Gamma \vdash \tmp \colon\! {Q}$}
	  \DisplayProof
	\end{equation*}
	with $\var \notin \Dom{\Gamma_3}$ (by \reflemma{free} since $\var \notin \Fv{\tmfour}$), $\Gamma = \Gamma_1 \uplus \Gamma_2 \uplus \Gamma_3$ and $\size{\pi'} = \size{\pi_1} + \size{\pi_2} + \size{\pi_3} + 2$. 
	Therefore, there is a derivation 
	\begin{equation*}
	  \pi =
	  \AxiomC{$\ \vdots\, \pi_1$}
	  \noLine
	  \UnaryInfC{$\Gamma_1, \var \colon\! P \vdash \tmtwo \colon\! [\Pair{Q'}{Q}]$}
	  \RightLabel{\footnotesize$\lambda$}
	  \UnaryInfC{$\Gamma_1 \vdash \la\var\tmtwo \colon\! [\Pair{P}{[\Pair{Q'}{Q}]}]$}
	  \AxiomC{$\ \vdots\, \pi_2$}
	  \noLine
	  \UnaryInfC{$\Gamma_2 \vdash \tmthree \colon\! P$}
	  \RightLabel{\footnotesize$@$}
	  \BinaryInfC{$\Gamma_1 \uplus \Gamma_2 \vdash \tm \colon\! [\Pair{Q'}{Q}]$}
	  \AxiomC{$\ \vdots\, \pi_3$}
	  \noLine
	  \UnaryInfC{$\Gamma_3 \vdash \tmfour \colon\! Q'$}
	  \RightLabel{\footnotesize$@$}
	  \BinaryInfC{$\Gamma \vdash \tm \colon\! Q$}
	  \DisplayProof
	\end{equation*}
	where $\size{\pi} = \size{\pi_1} + \size{\pi_3} + \size{\pi_2} + 2 = \size{\pi'}$.

	\item $\tm = \val((\la\var\tmtwo)\tmthree) \rtosr (\la\var\val\tmtwo)\tmthree = \tmp$ with $\var \notin \Fv{\val}$.
	Therefore,
	\begin{equation*}
	  \pi' =
	  \AxiomC{$\ \vdots\, \pi_1$}
	  \noLine
	  \UnaryInfC{$\Gamma_1 \vdash \val \colon\! [\Pair{Q'}{Q}]$}
	  \AxiomC{$\ \vdots\, \pi_3$}
	  \noLine
	  \UnaryInfC{$\Gamma_3, \var \colon\! P \vdash \tmtwo \colon\! Q'$}
	  \RightLabel{\footnotesize$@$}
	  \BinaryInfC{$\Gamma_1 \uplus \Gamma_3, \var \colon\! P \vdash \val\tmtwo \colon\! Q$}
	  \RightLabel{\footnotesize$\lambda$}
	  \UnaryInfC{$\Gamma_1 \uplus \Gamma_3 \vdash \la\var{\tmtwo\tmfour} \colon\! [\Pair{P}{Q}]$}
	  \AxiomC{$\ \vdots\, \pi_2$}
	  \noLine
	  \UnaryInfC{$\Gamma_2 \vdash \tmthree \colon\! P$}
	  \RightLabel{\footnotesize$@$}
	  \BinaryInfC{$\Gamma \vdash \tmp \colon\! {Q}$}
	  \DisplayProof
	\end{equation*}
	with $\var \notin \Dom{\Gamma_1}$ (by \reflemma{free} since $\var \notin \Fv{\val}$), $\Gamma = \Gamma_1 \uplus \Gamma_2 \uplus \Gamma_3$ and $\size{\pi'} = \size{\pi_1} + \size{\pi_2} + \size{\pi_3} + 2$. 
	Thus, there is a derivation 
	\begin{equation*}
	  \pi =
	  \AxiomC{$\ \vdots\, \pi_1$}
	  \noLine
	  \UnaryInfC{$\Gamma_1 \vdash \val \colon\! [\Pair{Q'}{Q}]$}
	  \AxiomC{$\ \vdots\, \pi_2$}
	  \noLine
	  \UnaryInfC{$\Gamma_2, \var \colon\! P \vdash \tmtwo \colon\! Q'$}
	  \RightLabel{\footnotesize$\lambda$}
	  \UnaryInfC{$\Gamma_2 \vdash \la\var\tmtwo \colon\! [\Pair{P}{Q'}]$}
	  \AxiomC{$\ \vdots\, \pi_3$}
	  \noLine
	  \UnaryInfC{$\Gamma_3 \vdash \tmthree \colon\! P$}
	  \RightLabel{\footnotesize$@$}
	  \BinaryInfC{$\Gamma_2 \uplus \Gamma_3 \vdash (\la\var\tmtwo)\tmthree \colon\! Q'$}
	  \RightLabel{\footnotesize$@$}
	  \BinaryInfC{$\Gamma \vdash \tm \colon\! Q$}
	  \DisplayProof
	\end{equation*}
	where $\size{\pi} = \size{\pi_1} + \size{\pi_3} + \size{\pi_2} + 2 = \size{\pi'}$.
      \end{itemize}

      \item \emph{Application Left}, \ie~$\tm = \tmtwo\tmthree \tosigm \tmtwop\tmthree = \tmp$ with $\tmtwo \tosigm \tmtwop$: then,
      \begin{equation*}
        \pi' =
        \AxiomC{$\ \vdots\, \pi_1'$}
        \noLine
	\UnaryInfC{$\Gamma_1 \vdash \tmtwop \colon\! [\Pair{P}{Q}]$}
        \AxiomC{$\ \vdots\, \pi_2$}
        \noLine
        \UnaryInfC{$\Gamma_2 \vdash \tmthree \colon\! P$}
        \RightLabel{\footnotesize$@$}
        \BinaryInfC{$\Gamma \vdash \tmp \colon\! Q$}
        \DisplayProof
      \end{equation*}
      where $\Gamma = \Gamma_1 \uplus \Gamma_2$ and $\size{\pi'} = \size{\pi_1'} + \size{\pi_2} + 1$. 
      By induction hypothesis, there exists $\concl{\pi_1}{\Gamma_1}{\tmtwop}{[\Pair{P}{Q}]}$ with $\size{\pi_1} = \size{\pi_1'}$.
      So, there is 
      \begin{equation*}
        \pi =
        \AxiomC{$\ \vdots\, \pi_1$}
        \noLine
	\UnaryInfC{$\Gamma_1 \vdash \tmtwo \colon\! [\Pair{P}{Q}]$}
        \AxiomC{$\ \vdots\, \pi_2$}
        \noLine
        \UnaryInfC{$\Gamma_2 \vdash \tmthree \colon\! P$}
        \RightLabel{\footnotesize$@$}
        \BinaryInfC{$\Gamma \vdash \tm \colon\! Q$}
        \DisplayProof
      \end{equation*}
      where $\size{\pi} = \size{\pi_1} + \size{\pi_2} + 1 = \size{\pi_1'} + \size{\pi_2} + 1 = \size{\pi'}$.
      
      \item \emph{Application Right}, \ie~$\tm = \tmtwo\tmthree \tosigm \tmtwo\tmthreep = \tmp$ with $\tmthree \tosigm \tmthreep$: then,
      \begin{equation*}
        \pi' =
        \AxiomC{$\ \vdots\, \pi_1$}
        \noLine
	\UnaryInfC{$\Gamma_1 \vdash \tmtwo \colon\! [\Pair{P}{Q}]$}
        \AxiomC{$\ \vdots\, \pi_2'$}
        \noLine
        \UnaryInfC{$\Gamma_2 \vdash \tmthreep \colon\! P$}
        \RightLabel{\footnotesize$@$}
        \BinaryInfC{$\Gamma \vdash \tmp \colon\! Q$}
        \DisplayProof
      \end{equation*}
      where $\Gamma = \Gamma_1 \uplus \Gamma_2$ and $\size{\pi'} = \size{\pi_1} + \size{\pi_2'} + 1$. 
      By induction hypothesis, there exists $\concl{\pi_2}{\Gamma_2}{\tmthree}{P}$ such that $\size{\pi_2} = \size{\pi_2'}$.
      So, there is 
      \begin{equation*}
        \pi =
        \AxiomC{$\ \vdots\, \pi_1$}
        \noLine
	\UnaryInfC{$\Gamma_1 \vdash \tmtwo \colon\! [\Pair{P}{Q}]$}
        \AxiomC{$\ \vdots\, \pi_2$}
        \noLine
        \UnaryInfC{$\Gamma_2 \vdash \tmthree \colon\! P$}
        \RightLabel{\footnotesize$@$}
        \BinaryInfC{$\Gamma \vdash \tm \colon\! Q$}
        \DisplayProof
      \end{equation*}
      where $\size{\pi} = \size{\pi_1} + \size{\pi_2} + 1 = \size{\pi_1} + \size{\pi_2'} + 1 = \size{\pi'}$.

      \item \emph{Step inside a $\beta$-redex}, \ie~$\tm = (\la\var\tmtwo)\tmthree \tosigm (\la\var\tmtwop)\tmthree = \tmp$ with $\tmtwo \tosigm \tmtwop$: then, 
      \begin{equation*}
        \pi' =
        \AxiomC{$\ \vdots\, \pi_1'$}
        \noLine
        \UnaryInfC{$\Gamma_1, \var \colon\! P \vdash \tmtwop \colon\! Q$}
        \RightLabel{\footnotesize$\lambda$}
        \UnaryInfC{$\Gamma_1 \vdash \la\var\tmtwop \colon\! [\Pair{P}{Q}]$}
        \AxiomC{$\ \vdots\, \pi_2$}
        \noLine
        \UnaryInfC{$\Gamma_2 \vdash \tmthree \colon\! P$}
        \RightLabel{\footnotesize$@$}
        \BinaryInfC{$\Gamma \vdash \tmp \colon\! Q$}
        \DisplayProof
      \end{equation*}
      where $\Gamma = \Gamma_1 \uplus \Gamma_2$ and $\size{\pi'} = \size{\pi_1} + \size{\pi_2} + 1$. 
      By induction hypothesis, there exists $\concl{\pi_1}{\Gamma_1, \var \colon\! P}{\tmtwo}{Q}$ with $\size{\pi_1} = \size{\pi_1'}$.
      So, there is
      \begin{equation*}
        \pi =
        \AxiomC{$\ \vdots\, \pi_1$}
        \noLine
        \UnaryInfC{$\Gamma_1, \var \colon\! P \vdash \tmtwo \colon\! Q$}
        \RightLabel{\footnotesize$\lambda$}
        \UnaryInfC{$\Gamma_1 \vdash \la\var\tmtwo \colon\! [\Pair{P}{Q}]$}
        \AxiomC{$\ \vdots\, \pi_2$}
        \noLine
        \UnaryInfC{$\Gamma_2 \vdash \tmthree \colon\! P$}
        \RightLabel{\footnotesize$@$}
        \BinaryInfC{$\Gamma \vdash \tm \colon\! Q$}
        \DisplayProof
      \end{equation*}
      where $\size{\pi} = \size{\pi_1} + \size{\pi_2} + 1 = \size{\pi_1'} + \size{\pi_2} + 1 = \size{\pi'}$. 
      \qedhere
    \end{itemize}
  \end{enumerate}
\end{proof}

\setcounter{lemmaAppendix}{\value{l:subject-reduction}}
\begin{lemmaAppendix}[Subject reduction]
\label{lappendix:subject-reduction}
\NoteState{l:subject-reduction}
  Let $\tm, \tmp \in \Lambda$ and $\concl{\pi}{\Gamma}{\tm}{Q}$.
  \begin{enumerate}
    \item\label{pappendix:subject-reduction-betav}\emph{Shrinkage under $\betav$-step:} If $\tm \tobv \tmp$ then there is $\concl{\pi'}{\Gamma}{\tmp}{Q}$ with $\size{\pi} \geq \size{\pi'}$.
    \item\label{pappendix:subject-reduction-sigma}\emph{Size invariance under $\sigma$-step:} If $\tm \tosig \tmp$ then there is $\concl{\pi'}{\Gamma}{\tmp}{Q}$ such that $\size{\pi} = \size{\pi'}$.
  \end{enumerate}
\end{lemmaAppendix}

\begin{proof}
  Analogous to the proofs of \refpropps{quant-subject-reduction}{betav}{sigma}, paying attention that now the induction hypothesis is weaker.
  The only novelty is the presence of the following case, since $\toshuf$ reduces under $\lambda$'s:
  \begin{itemize}
    \item \emph{Abstraction}, \ie $\tm = \la\var\tmtwo \to_{\Rule} \la\var\tmtwop = \tmp$ with $\tmtwo \to_{\Rule} \tmtwop$ and $\Rule \in \{\betav,\sigma\}$: then, 
      \begin{equation*}
        \pi =
        \AxiomC{$\ \vdots\, \pi_1$}
        \noLine
        \UnaryInfC{$\Gamma_1, \var \colon\! P_1 \vdash \tmtwo \colon\! Q_1$}
        \AxiomC{$\overset{n \in \nat}{\ldots}$}
        \AxiomC{$\ \vdots\, \pi_n$}
        \noLine
        \UnaryInfC{$\Gamma_n, \var \colon\! P_n \vdash \tmtwo \colon\! Q_n$}
        \RightLabel{\footnotesize$\lambda$}
        \TrinaryInfC{$\Gamma \vdash \tm \colon\! [\Pair{P_1}{Q_1}, \dots, \Pair{P_n}{Q_n}]$}
        \DisplayProof
      \end{equation*}
      where $\Gamma = \biguplus_{i=1}^n \Gamma_i$ and $\size{\pi} = \sum_{i=1}^n\size{\pi_i}$. 
      By induction hypothesis, for all $1 \leq i \leq n$ there is $\concl{\pi_i'}{\Gamma_i, \var \colon\! P_i}{\tmtwop}{Q_i}$ with $\size{\pi_i} \geq \size{\pi_i'}$ if $\Rule = \betav$, and $\size{\pi_i} = \size{\pi_i'}$ if $\Rule = \sigma$.
      So, there is
      \begin{equation*}
        \pi' =
        \AxiomC{$\ \vdots\, \pi_1'$}
        \noLine
        \UnaryInfC{$\Gamma_1, \var \colon\! P_1 \vdash \tmtwop \colon\! Q_1$}
        \AxiomC{${\ldots}$}
        \AxiomC{$\ \vdots\, \pi_n'$}
        \noLine
        \UnaryInfC{$\Gamma_n, \var \colon\! P_n \vdash \tmtwop \colon\! Q_n$}
        \RightLabel{\footnotesize$\lambda$}
        \TrinaryInfC{$\Gamma \vdash \tmp \colon\! [\Pair{P_1}{Q_1}, \dots, \Pair{P_n}{Q_n}]$}
        \DisplayProof
      \end{equation*}
      where $\size{\pi'} = \sum_{i=1}^n \size{\pi_i'}$. 
      Therefore, $\size{\pi} \geq \size{\pi'}$ if $\Rule = \betav$, and $\size{\pi} = \size{\pi'}$ if $\Rule = \sigma$.
      \qedhere
  \end{itemize}
\end{proof}

\setcounter{lemmaAppendix}{\value{l:subject-expansion}}
\begin{lemmaAppendix}[Subject expansion]
\label{lappendix:subject-expansion}
\NoteState{l:subject-expansion}
  Let $\tm, \tmp \in \Lambda$ and $\concl{\pi'}{\Gamma}{\tmp}{Q}$.
  \begin{enumerate}
    \item\label{pappendix:subject-expansion-betav}\emph{Enlargement under anti-$\betav$-step:} If $\tm \tobv \tmp$ then there is $\concl{\pi}{\Gamma}{\tm}{Q}$ with $\size{\pi} \geq \size{\pi'}$.
    \item\label{pappendix:subject-expansion-sigma}\emph{Size invariance under anti-$\sigma$-step:} If $\tm \tosig \tmp$ then there is $\concl{\pi}{\Gamma}{\tm}{Q}$ such that $\size{\pi} = \size{\pi'}$.
  \end{enumerate}
\end{lemmaAppendix}

\begin{proof}
  Analogous to the proofs of \refpropps{quant-subject-expansion}{betav}{sigma}, paying attention that now the induction hypothesis is weaker.
  The only novelty is the presence of the following case, since $\toshuf$ reduces under $\lambda$'s:
  \begin{itemize}
    \item \emph{Abstraction}, \ie $\tm = \la\var\tmtwo \to_{\Rule} \la\var\tmtwop = \tmp$ with $\tmtwo \to_{\Rule} \tmtwop$ and $\Rule \in \{\betav,\sigma\}$: then, 
      \begin{equation*}
        \pi' =
        \AxiomC{$\ \vdots\, \pi_1'$}
        \noLine
        \UnaryInfC{$\Gamma_1, \var \colon\! P_1 \vdash \tmtwop \colon\! Q_1$}
        \AxiomC{$\overset{n \in \nat}{\ldots}$}
        \AxiomC{$\ \vdots\, \pi_n'$}
        \noLine
        \UnaryInfC{$\Gamma_n, \var \colon\! P_n \vdash \tmtwop \colon\! Q_n$}
        \RightLabel{\footnotesize$\lambda$}
        \TrinaryInfC{$\Gamma \vdash \tmp \colon\! [\Pair{P_1}{Q_1}, \dots, \Pair{P_n}{Q_n}]$}
        \DisplayProof
      \end{equation*}
      where $\Gamma = \biguplus_{i=1}^n \Gamma_i$ and $\size{\pi'} = \sum_{i=1}^n\size{\pi_i'}$. 
      By induction hypothesis, for all $1 \leq i \leq n$ there is $\concl{\pi_i}{\Gamma_i, \var \colon\! P_i}{\tmtwo}{Q_i}$ with $\size{\pi_i} \geq \size{\pi_i'}$ if $\Rule = \betav$, and $\size{\pi_i} = \size{\pi_i'}$ if $\Rule = \sigma$.
      So, there is
      \begin{equation*}
        \pi =
        \AxiomC{$\ \vdots\, \pi_1$}
        \noLine
        \UnaryInfC{$\Gamma_1, \var \colon\! P_1 \vdash \tmtwo \colon\! Q_1$}
        \AxiomC{${\ldots}$}
        \AxiomC{$\ \vdots\, \pi_n$}
        \noLine
        \UnaryInfC{$\Gamma_n, \var \colon\! P_n \vdash \tmtwo \colon\! Q_n$}
        \RightLabel{\footnotesize$\lambda$}
        \TrinaryInfC{$\Gamma \vdash \tm \colon\! [\Pair{P_1}{Q_1}, \dots, \Pair{P_n}{Q_n}]$}
        \DisplayProof
      \end{equation*}
      where $\size{\pi} = \sum_{i=1}^n \size{\pi_i}$. 
      Therefore, $\size{\pi} \geq \size{\pi'}$ if $\Rule = \betav$, and $\size{\pi} = \size{\pi'}$ if $\Rule = \sigma$.
      \qedhere
  \end{itemize}
\end{proof}

\subsection{Omitted proofs and remarks of Section~\ref{sect:qualitative-semantics}}

\setcounter{theoremAppendix}{\value{thm:invariance}}
\begin{theoremAppendix}[Invariance under $\shuf$-equivalence]
\label{thmappendix:invariance}
\NoteState{thm:invariance}
  Let $\tm, \tmtwo \in \Lambda$, let $k \in \nat$ and let $\vec{\var} = (\var_1, \dots, \var_k)$ be a suitable list of variables for $\tm$ and $\tmtwo$.
  If $\tm \shufeq \tmtwo$ then $\sem{\tm}{\vec{\var}} = \sem{\tmtwo}{\vec{\var}}$.
\end{theoremAppendix}

\begin{proof}
  Since $\tm \shufeq \tmtwo$, there exist $q \in \nat$ and $\tm_0, \dots, \tm_q$ such that $\tm = \tm_0$, $\tmtwo = \tm_q$ and $\tm_i \toshuf \tm_{i+1}$ or $\tm_{i+1} \toshuf \tm_i$, for all $0 \leq i < q$.
  Using subject reduction (\reflemma{subject-reduction}) and subject expansion (\reflemma{subject-expansion}), it is immediate to prove by induction on $q \in \nat$ that $\sem{\tm}{\vec{\var}} = \sem{\tmtwo}{\vec{\var}}$.
\end{proof}

\setcounter{lemmaAppendix}{\value{l:semantics-normal}}
\begin{lemmaAppendix}[Semantics and typability of $\shufm$-normal forms]
\label{lappendix:semantics-normal}
\NoteState{l:semantics-normal}
  Let $\tm$ be a term, let $k \in \nat$ and let $\vec{\var} = (\var_1, \dots, \var_k)$ be a list of variables suitable for $\tm$.
  \begin{enumerate}
    \item\label{pappendix:semantics-normal-anf} If $\tm \in \anfSet$ then for every positive type $Q$ there exist positive types $P_1, \dots, P_k$ and a derivation $\concl{\pi}{\var_1 \colon\! P_1, \dots, \var_k \colon\!  P_k}{\tm}{Q}$.
    \item\label{pappendix:semantics-normal-wnf} If $\tm \in \wnfSet$ then there are positive types $Q, P_1, \dots, P_k$ and a derivation $\concl{\pi}{\var_1 \colon\! P_1, \dots, \var_k \colon\!  P_k\allowbreak}{\tm}{Q}$.
    \item\label{pappendix:semantics-normal-nonempty} If $\tm$ is $\shufm$-normal then $\sem{\tm}{\vec{\var}} \neq \emptyset$.
  \end{enumerate}
\end{lemmaAppendix}

\begin{proof}
  \refpoint{semantics-normal-nonempty} is an immediate consequence of \refpoint{semantics-normal-wnf} via the syntactic characterization of $\shufm$-normal forms (\refprop{syntactic-normal}).
  
  We prove simultaneously \refpoints{semantics-normal-anf}{semantics-normal-wnf} by mutual induction on $\tm \in \anfSet \cup \wnfSet$.
  
  Cases for $\tm \in \anfSet$:
  \begin{itemize}
    \item $\tm = \var\val$ for some variable $\var$ and value $\val$:
    since $\vec{\var}$ is suitable for $\tm$, $\var = \var_i$ for some $1 \leq i \leq k$.
    According to \refcor{minimal} there is a derivation $\concl{\pi'}{\,}{\val}{\emptymset}$; 
    thus, for any positive type $Q$ there exists the derivation
    \begin{equation*}
      \pi =
      \AxiomC{}
      \RightLabel{\footnotesize{$\mathsf{ax}$}}
      \UnaryInfC{$\var_i \colon\! [\Pair{\emptymset}{Q}] \vdash \var_i \colon\! [\Pair{\emptymset}{Q}]$}
      \AxiomC{$\ \vdots\,\pi'$}
      \noLine
      \UnaryInfC{$ \vdash \val \colon\! \emptymset$}
      \RightLabel{\footnotesize{$@$}}
      \BinaryInfC{$\ \var_i \colon\! [\Pair{\emptymset}{Q}] \vdash \tm \colon\! Q$}
      \DisplayProof
    \end{equation*}
    \item $\tm = \var\anf$ for some variable $\var$ and $\anf \in \anfSet$:
    since $\vec{\var}$ is suitable for $\tm$, one has $\var = \var_i$ for some $1 \leq i \leq k$.
    By \ih, there exists a derivation $\concl{\pi'}{\var_1 \colon\! P_1, \dots, \var_k \colon\! P_k}{\anf}{[\,]}$ for some positive types $P_1, \dots, P_k$.
    So, for any positive type $Q$ there exists a derivation
    \begin{equation*}
      \pi =
      \AxiomC{}
      \RightLabel{\footnotesize{$\mathsf{ax}$}}
      \UnaryInfC{$\var_i \colon\! [\Pair{\emptymset}{Q}] \vdash \var_i \colon\! [\Pair{\emptymset}{Q}]$}
      \AxiomC{$\ \vdots\,\pi'$}
      \noLine
      \UnaryInfC{$\var_1 \colon\! P_1, \dots, \var_k \colon\! P_k \vdash \anf \colon\! [\,]$}
      \RightLabel{\footnotesize{$@$}}
      \BinaryInfC{$\var_1 \colon\! P_1, \dots, \var_i \colon\! [\Pair{[\,]}{Q}] \uplus P_i, \dots, \var_k \colon\! P_k \vdash \tm \colon\! Q$}
      \DisplayProof
    \end{equation*}
    \item $\tm = \anf\wnf$ for some $\anf \in \anfSet$ and $\wnf \in \wnfSet$:
    by \ih applied to $\wnf$, there is a derivation $\concl{\pi''}{\var_1 \colon\! P_1, \dots, \var_k \colon\! P_k}{\wnf}{P}$ for some positive types $P, P_1, \dots, P_k$.
    Given a positive type $Q$, by \ih applied to $\anf$, there is a derivation $\concl{\pi'}{\var_1 \colon\! P_1', \dots, \var_k \colon\! P_k'}{\anf}{[\Pair{P}{Q}]}$ for some positive types $P_1', \dots, P_k'$.
    So, there is a derivation
    \begin{equation*}
      \pi =
      \AxiomC{$\ \vdots\,\pi'$}
      \noLine
      \UnaryInfC{$\var_1 \colon\! P_1', \dots, \var_k \colon\! P_k' \vdash \anf \colon\! [\Pair{P}{Q}]$}
      \AxiomC{$\ \vdots\,\pi''$}
      \noLine
      \UnaryInfC{$\var_1 \colon\! P_1, \dots, \var_k \colon\! P_k \vdash \wnf \colon\! P$}
      \RightLabel{\footnotesize{$@$}}
      \BinaryInfC{$\var_1 \colon\! P_1 \uplus P_1', \dots, \var_k \colon\! P_k \uplus P_k' \vdash \tm \colon\! Q$}
      \DisplayProof.
    \end{equation*}
  \end{itemize}
  
  This completes the case analysis for $\tm \in \anfSet$.
  Cases for $\tm \in \wnfSet$:
  \begin{itemize}
    \item $\tm \in \anfSet$: see above.
    \item $\tm$ is a value: the statement follows from \refcor{minimal}, taking $P_1 = \ldots = P_k = Q = \emptymset$.
    \item $\tm = (\la{\var}{\wnf})\anf$ for some $\wnf \in \wnfSet$ and $\anf \in \anfSet$:
    by \ih~applied to $\wnf$, there is $\concl{\pi'}{\var_1 \colon\! P_1', \dots, \var_k \colon\! P_k', \var \colon\! P}{\wnf}{Q}$ for some positive types $P_1', \dots, P_k', P, Q$.
    By \ih~applied to $\anf$, there exists a derivation $\concl{\pi''}{\var_1 \colon\! P_1, \dots, \var_k \colon\! P_k}{\anf}{P}$ for some positive types $P_1, \dots, P_k$.
    Therefore, there is a derivation
    \begin{equation*}
      \pi =
      \AxiomC{$\ \vdots\,\pi'$}
      \noLine
      \UnaryInfC{$\var_1 \colon\! P_1', \dots, \var_k \colon\! P_k', \var \colon\! P \vdash \wnf \colon\! Q$}
      \RightLabel{\footnotesize{$\lambda$}}
      \UnaryInfC{$\var_1 \colon\! P_1', \dots, \var_k \colon\! P_k' \vdash \la{\var}{\wnf} \colon\! [\Pair{P}{Q}]$}
      \AxiomC{$\ \vdots\,\pi''$}
      \noLine
      \UnaryInfC{$\var_1 \colon\! P_1, \dots, \var_k \colon\! P_k \vdash \anf \colon\! P$}
      \RightLabel{\footnotesize{$@$}}
      \BinaryInfC{$\var_1 \colon\! P_1 \uplus P_1', \dots, \var_k \colon\! P_k \uplus P_k' \vdash \tm \colon\! Q$}
      \DisplayProof.
    \end{equation*}    
    \qedhere
  \end{itemize}
\end{proof}

\setcounter{theoremAppendix}{\value{thm:characterize-normalizable}}
\begin{theoremAppendix}[Semantic and logic characterization of $\shufm$-normalization]
\label{thmappendix:characterize-normalizable}
\NoteState{thm:characterize-normalizable}
  Let $\tm \in \Lambda$ and let $\vec{\var} = (\var_1, \dots, \var_k)$ be a suitable list of variables for $\tm$.
  The following are equivalent:
  \begin{enumerate}
    \item\label{thmappendix:characterize-normalizable-normalizable} $\tm$ is $\shufm$-normalizable;
    \item\label{thmappendix:characterize-normalizable-equivalen-normal} $\tm \shufeq \tmtwo$ for some $\shufm$-normal $\tmtwo \in \Lambda$;
    \item\label{thmappendix:characterize-normalizable-nonempty} $\sem{\tm}{\vec{\var}} \neq \emptyset$;
    \item\label{thmappendix:characterize-normalizable-derivable} there exists a derivation $\concl{\pi}{\var_1 \colon\! P_1, \dots, \var_n \colon\! P_n}{\tm}{Q}$ for some positive types $P_1, \dots, P_n, Q$;
    \item\label{thmappendix:characterize-normalizable-strongly-normalizable} $\tm$ is strongly $\shufm$-normalizable.
  \end{enumerate}
\end{theoremAppendix}

\begin{proof}
  \begin{description}
    \item[\eqref{p:characterize-normalizable-normalizable}$\Rightarrow$\eqref{p:characterize-normalizable-equivalen-normal}:] Trivial, since $\toshufm \,\subseteq\, \toshuf \,\subseteq\, \shufeq$.
    \item[\eqref{p:characterize-normalizable-equivalen-normal}$\Rightarrow$\eqref{p:characterize-normalizable-nonempty}:] First, note that we can suppose without loss of generality that  $\vec{\var}$ is suitable also for $\tmtwo$. By \reflemma{semantics-normal}, $\sem{\tmtwo}{\vec{\var}} \neq \emptyset$.
    By invariance of relational semantics (\refthm{invariance}), $\sem{\tm}{\vec{\var}} = \sem{\tmtwo}{\vec{\var}}$.
    \item[\eqref{p:characterize-normalizable-nonempty}$\Rightarrow$\eqref{p:characterize-normalizable-derivable}] Trivial, according to \refdef{semantics}.
    \item[\eqref{p:characterize-normalizable-derivable}$\Rightarrow$\eqref{p:characterize-normalizable-strongly-normalizable}:] If there is a derivation $\concl{\pi}{\var_1 \colon\! P_1, \dots, \var_k \colon\! P_k}{\tm}{Q}$ for some positive types $P_1, \dots, P_k, Q$, then every $\shufm$-reduction sequence from $\tm$ has at most $\size{\pi} \in \nat$ $\betavm$-reduction steps by the quantitative subject reduction (\refpropps{quant-subject-reduction}{betav}{sigma}).
	As there is no $\shufm$-reduction sequence from $\tm$ with infinitely many $\betavm$-reduction steps, then every infinite $\shufm$-reduction sequence from $\tm$ would have infinitely many $\sigm$-reduction steps, but this is impossible since $\tosigm$ is strongly normalizing. 
	Therefore, there is no infinite $\shufm$-reduction sequence from $\tm$, which means that $\tm$ is strongly $\shufm$-normalizable.

    \item[\eqref{p:characterize-normalizable-strongly-normalizable}$\Rightarrow$\eqref{p:characterize-normalizable-normalizable}:] Trivial. 
    \qedhere
  \end{description}
\end{proof}

\begin{lemma}\label{l:anf-environment}
	Let $\tm \in \anfSet$.
	For all $\concl{\pi}{\var_1 \colon\! P_1, \dots, \var_k \colon\! P_k}{\tm}{Q}$ there is $1 \leq i \leq k$ such that $P_i \neq \emptymset$.
\end{lemma}

\begin{proof}
	By induction on $\tm \in \anfSet$.
	
	If $\tm = \var\tmtwo$ where $\tmtwo \in \valSet \cup \anfSet$, then $\var = \var_i$ for some $1 \leq i \leq k$, and hence
	\begin{equation*}
	\pi = 
	\AxiomC{}
	\RightLabel{\footnotesize{$\Ax$}}
	\UnaryInfC{$\var_1 \colon\! P_1', \dots, \var_i \colon\! [\Pair{P}{Q}], \dots, \var_k \colon\!P_k' \vdash \var \colon\! [\Pair{P}{Q}]$}
	\AxiomC{$\ \vdots\,\pi''$}
	\noLine
	\UnaryInfC{$\var_1 \colon\! P_1'', \dots, \var_k \colon\! P_k'' \vdash \tmtwo \colon\! P$}
	\RightLabel{\footnotesize{$@$}}
	\BinaryInfC{$\var_1 \colon\! P_1 , \dots , \var_k \colon\! P_k \vdash \tm \colon\! Q$}
	\DisplayProof
	\end{equation*}
	where $P_j = P_j' \uplus P_j''$ for all $1 \leq j \leq k$ such that $j \neq i$, and $P_i = [\Pair{P}{Q}] \uplus P_i'' \neq \emptymset$.
	
	If $\tm = \anf\wnf$ for some $\anf \in \anfSet$ and $\wnf \in \wnfSet$, then
	\begin{equation*}
	\pi = 
	\AxiomC{$\ \vdots\,\pi'$}
	\noLine
	\UnaryInfC{$\var_1 \colon\! P_1', \dots, \var_k \colon\!P_k' \vdash \anf \colon\! [\Pair{P}{Q}]$}
	\AxiomC{$\ \vdots\,\pi''$}
	\noLine
	\UnaryInfC{$\var_1 \colon\! P_1'', \dots, \var_k \colon\! P_k'' \vdash \wnf \colon\! P$}
	\RightLabel{\footnotesize{$@$}}
	\BinaryInfC{$\var_1 \colon\! P_1 , \dots , \var_k \colon\! P_k \vdash \tm \colon\! Q$}
	\DisplayProof
	\end{equation*}
	where $P_j = P_j' \uplus P_j''$ for all $1 \leq j \leq k$.
	By \ih, there is $1 \leq i \leq k$ such that $P_i' \neq \emptymset$, thus $P_i = P_i' \uplus P_i'' \neq \emptymset$.
\end{proof}

\setcounter{lemmaAppendix}{\value{l:semantic-value}}
\begin{lemmaAppendix}[Uniqueness of the derivation with empty types; Logic and semantic characterization of values]
  \label{lappendix:semantic-value}
  \NoteState{l:semantic-value}
  Let $\tm \in \Lambda$ be $\shufm$-normal.
  
  \begin{enumerate}
    \item\label{pappendix:semantic-value-uniqueness}  If $\concl{\pi}{\,}{\tm}{\emptymset}$ and $\concl{\pi'}{\Gamma}{\tm}{\emptymset}$, then $\tm \in \valSet$, $\size{\pi} = 0$, $\Dom{\Gamma} = \emptyset$ and $\pi = \pi'$.
    More precisely, $\pi$ consists 
    of a rule $\Ax$ if $\tm$ is a variable, otherwise $\tm$ is an abstraction and $\pi$ consists 
    of a 0-ary~rule~$\lambda$.  
    
    \item\label{pappendix:semantic-value-characterization} Given a list $\vec{\var} = (\var_1, \dots, \var_k)$ of variables suitable for $\tm$, the following are equivalent:
    \begin{multicols}{2}
    \begin{enumerate}
	    \item\label{pappendix:semantic-value-value} $\tm$ is a value;
	    \item\label{pappendix:semantic-value-semantic} $((\emptymset, \overset{k}{\dots\,}, \emptymset), \emptymset) \in \sem{\tm}{\vec{\var}}$;
	    \item\label{pappendix:semantic-value-logic} there 
	    is $\concl{\pi}{\,}{\tm}{\emptymset}$;
	    \item\label{pappendix:semantic-value-size} there 
	    is $\Type{\pi}{\tm}$ such that $\size{\pi} = 0$.
    \end{enumerate} 
    \end{multicols}	
  \end{enumerate}
\end{lemmaAppendix}

\begin{proof}
  \begin{enumerate}
    \item According to \refprop{syntactic-normal}, $\tm \in \wnfSet$ since $\tm$ is $\shufm$-normal, so there are only three cases.
	
    If $\tm \in \valSet$, then $\size{\pi} = 0$ by the left-to-right direction of \reflemma{value} (take $p = 0$). 
    Moreover, if $\tm$ is a variable, then $\tm = \var_i$ for some $1 \leq i \leq k$, and the only derivation with conclusion $\Gamma \vdash \tm \colon\!\emptymset$ is 
    \begin{equation*}
      \pi = 
      \AxiomC{}
      \RightLabel{\footnotesize$\Ax$}
      \UnaryInfC{$\var_i \colon\! \emptymset \vdash \var_i \colon \! \emptymset$}
      \DisplayProof
      \qquad\text{which means } \Dom{\Gamma} = \emptyset;
    \end{equation*}
    otherwise, $\tm = \la{\var}{\tmtwo}$ and the only derivation with conclusion $\Gamma \vdash \tm \colon\!\emptymset$ is 
    \begin{equation*}
      \pi = 
      \AxiomC{}
      \RightLabel{\footnotesize$\l$}
      \UnaryInfC{$\vdash \la{\var}{\tmtwo} \colon \! \emptymset$}
      \DisplayProof
      \qquad\text{\ie~a $0$-ary rule $\lambda$, which means } \Dom{\Gamma} = \emptyset.
    \end{equation*}
    
    If $\tm \in \anfSet$ then $\tm \notin \valSet$ and it is impossible that $\concl{\pi}{\,}{\tm}{\emptymset}$ by \reflemma{anf-environment}.
    
    Finally, if $\tm \in \wnfSet \smallsetminus (\valSet \cup \anfSet)$, then $\tm = (\la{\var}{\wnf})\anf$ for some $\wnf \in \wnfSet$ and $\anf \in \anfSet$.
    By necessity,
    \begin{equation*}
	    \pi = 
	    \AxiomC{$\ \vdots\,\pi'$}
	    \noLine
	    \UnaryInfC{$\vdash \la{\var}{\wnf} \colon\! [\Pair{P}{\emptymset}]$}
	    \AxiomC{$\ \vdots\,\pi''$}
	    \noLine
	    \UnaryInfC{$\vdash \anf \colon\! P$}
	    \RightLabel{\footnotesize{$@$}}
	    \BinaryInfC{$\vdash \tm \colon\! \emptymset$}
	    \DisplayProof
    \end{equation*}
    but it is impossible that $\concl{\pi''}{\,}{\anf}{P}$, according to \reflemma{anf-environment}. 
    Thus, it is impossible that $\concl{\pi}{\,}{\tm}{\emptymset}$.

    \item The equivalence \eqref{p:semantic-value-semantic}$\Leftrightarrow$\eqref{p:semantic-value-logic} follows immediately from \refdef{semantics}.
    From \reflemmap{semantic-value}{uniqueness} it follows that
    \eqref{p:semantic-value-logic}$\Rightarrow$\eqref{p:semantic-value-size}.
    By \refcoro{minimal}, the implication \eqref{p:semantic-value-value}$\Rightarrow$\eqref{p:semantic-value-logic} 
    holds.
    In order to prove that \eqref{p:semantic-value-size}$\Rightarrow$\eqref{p:semantic-value-value}, it is enough to notice that there is no instance of the rule $@$ in $\Type{\pi}{\tm}$ since $\size{\pi} = 0$, so $\tm$ is either a variable or an abstraction, \ie a value.
    \qedhere
  \end{enumerate}
\end{proof}

\setcounter{propositionAppendix}{\value{prop:semantic-valuable}}
\begin{propositionAppendix}[Logic and semantic characterization of valuability]
	\label{propappendix:semantic-valuable}
\NoteState{prop:semantic-valuable}
	Let $\tm$ be a 
	term and $\vec{\var} = (\var_1, \dots, \var_k)$ be a suitable list of variables for $\tm$.
	The following are equivalent:
		\begin{enumerate}
			\item\label{pappendix:semantic-valuable-value} \emph{Valuability:} $\tm$ is $\shufm$-normalizable and the $\shufm$-normal form of $\tm$ is a value;
			\item\label{pappendix:semantic-valuable-semantic} \emph{Empty point in the semantics:} $((\emptymset, \overset{k}{\dots}\,, \emptymset), \emptymset) \in \sem{\tm}{\vec{\var}}$;
			\item\label{pappendix:semantic-valuable-logic}\emph{Derivability with empty types:} there exists a derivation $\concl{\pi}{\,}{\tm}{\emptymset}$.
		\end{enumerate} 
\end{propositionAppendix}

\begin{proof}
	The equivalence \eqref{p:semantic-valuable-semantic}$\Leftrightarrow$\eqref{p:semantic-valuable-logic} follows immediately from \refdef{semantics}.
	In order to prove the equivalence \eqref{p:semantic-valuable-value}$\Leftrightarrow$\eqref{p:semantic-valuable-semantic}, let us consider the two possible cases:
	\begin{itemize}
	  \item either $\tm$ is not $\shufm$-normalizable, and then $\sem{\tm}{\vec{\var}} = \emptyset$ according to the semantic characterization of $\shufm$-normalization (\refthm{characterize-normalizable}), in particular $((\emptymset, \overset{k}{\dots}\,, \emptymset), \emptymset) \notin \sem{\tm}{\vec{\var}}$;

	  \item or $\tm$ is $\shufm$-normalizable; let $\tm_0$ be its $\shufm$-normal form;
	according to the semantic characterization of values (\reflemmap{semantic-value}{characterization}), $((\emptymset, \overset{k}{\dots}\,, \emptymset), \emptymset) \in \sem{\tm_0}{\vec{\var}}$ iff $\tm_0$ is a value;
	by invariance of the semantics (\refthm{invariance}), $\sem{\tm}{\vec{\var}} = \sem{\tm_0}{\vec{\var}}$;
	therefore, $((\emptymset, \overset{k}{\dots}\,, \emptymset), \emptymset) \in \sem{\tm}{\vec{\var}}$ iff $\tm_0$ is a value.
	\qedhere
	\end{itemize}
\end{proof}

\subsection{Omitted proofs and remarks of Section~\ref{sect:quantitative-semantics}}

\setcounter{lemmaAppendix}{\value{l:sizes}}
\begin{lemmaAppendix}[Relationship between sizes of normal forms and derivations]
\label{lappendix:sizes}
\NoteState{l:sizes}
	Let $\tm \in \Lambda$.
	\begin{enumerate}
		\item\label{pappendix:sizes-normal} If $\tm$ is $\shufm$-normal then $\sizeZero{\tm} = \min{\{\size{\pi} \mid \Type{\pi}{\tm}\}}$.
		\item\label{pappendix:sizes-value} If $\tm$ is a value then $\sizeZero{\tm} = \min{\{\size{\pi} \mid \Type{\pi}{\tm}\}} = 0 
		$.
	\end{enumerate}
\end{lemmaAppendix}

\begin{proof}
	\begin{enumerate}
		\item By \refprop{syntactic-normal}, since $\tm$ is $\shufm$-normal, we can proceed by induction on $\tm \in \wnfSet$. 
		Moreover, for $\tm \in \anfSet$ we prove also that, given $\vec{\var} = (\var_1, \dots, \var_k)$ suitable for $\tm$, for \emph{any} positive type $Q$ there exist positive types $P_1, \dots, P_k$ and a derivation $\concl{\pi}{\var \colon\! P_1, \dots, \var_k \colon\! P_k}{\tm}{Q}$ such that $\sizeZero{\tm} = \size{\pi}$: this stronger statement is required to handle the case where $\tm$ is a $\shufm$-normal $\beta$-redex.
		
		If $\tm$ is a value, then $\sizeZero{\tm} = 0$ by definition, and there is a derivation $\concl{\pi}{\,}{\tm}{\emptymset}$ such that $\size{\pi} = 0$, according to \refcoro{minimal}.
		Thus, $\sizeZero{\tm} = \min{\{\size{\pi} \mid \Type{\pi}{\tm}\}}$ since $\size{\pi'} \geq 0$ for any derivation $\Type{\pi'}{\tm}$.
		
		If $\tm \in \anfSet$, then there are three cases:
		\begin{itemize}
			\item $\tm = \var\val$ for some variable $\var$ and value $\val$: $\sizeZero{\tm} = 1$, $k > 0$ and $\var = \var_i$ for some $1 \leq i \leq k$.
			According to \refcor{minimal}, there exists a derivation $\concl{\pi'}{\,}{\val}{\emptymset}$ such that $\size{\pi'} = 0$, so for any positive type $Q$ there exists the derivation
			
			\small{
			\begin{equation*}
			\pi =
			\AxiomC{}
			\RightLabel{\footnotesize{$\mathsf{ax}$}}
			\UnaryInfC{$\var_i \colon\! [\Pair{\emptymset}{Q}] \vdash \var_i \colon\! [\Pair{\emptymset}{Q}]$}
			\AxiomC{$\ \vdots\,\pi'$}
			\noLine
			\UnaryInfC{$\vdash \val \colon\! \emptymset$}
			\RightLabel{\footnotesize{$@$}}
			\BinaryInfC{$\var_i \colon\! [\Pair{\emptymset}{Q}] \vdash \tm \colon\! Q$}
			\DisplayProof
			\end{equation*}
			}
			
			\noindent where $\size{\pi} = \size{\pi'} + 1 = 1 = \sizeZero{\tm}$.
			The last rule of any derivation $\Type{\pi''}{\tm}$ is by necessity $@$, thus $\size{\pi''} \geq 1$ and hence $\min{\{\size{\pi''} \mid \Type{\pi''}{\tm}\}} = 1 = \sizeZero{\tm}$\,. 
			
			\item $\tm = \var\anf$ for some variable $\var$ and $\anf \in \anfSet$:
			$\sizeZero{\tm} = \sizeZero{\anf} + 1$, $k > 0$ and $\var = \var_i$ for some $1 \leq i \leq k$.
			By \ih, $\sizeZero{\anf} = \min{\{\size{\pi'} \mid \Type{\pi'}{\anf}\}}$ and there exists a derivation $\concl{\pi'}{\var_1 \colon\! P_1, \dots, \var_k \colon\! P_k}{\anf}{\emptymset}$ for some positive types $P_1, \dots, P_k$ such that $\sizeZero{\anf} = \size{\pi'}$.
			Therefore, for any positive type $Q$ there exists a derivation
			
			\small{
			\begin{equation*}
			\pi =
			\AxiomC{}
			\RightLabel{\footnotesize{$\mathsf{ax}$}}
			\UnaryInfC{$\var_i \colon\! [\Pair{\emptymset}{Q}] \vdash \var_i \colon\! [\Pair{\emptymset}{Q}]$}
			\AxiomC{$\ \vdots\,\pi'$}
			\noLine
			\UnaryInfC{$\var_1 \colon\! P_1, \dots, \var_k \colon\! P_k \vdash \anf \colon\! [\,]$}
			\RightLabel{\footnotesize{$@$}}
			\BinaryInfC{$\var_1 \colon\! P_1, \dots, \var_i \colon\! [\Pair{[\,]}{Q}] \uplus P_i, \dots, \var_k \colon\! P_k \vdash \tm \colon\! Q$}
			\DisplayProof
			\end{equation*}
			}
		
			\noindent where $\size{\pi} = \size{\pi'} + 1 = \sizeZero{\anf} + 1 = \sizeZero{\tm}$.
			The last rule of any derivation $\Type{\pi''}{\tm}$ is by necessity $@$, having a derivation typing $\anf$ as a premise, thus $\size{\pi''} \geq \sizeZero{\anf} + 1$ and so $\min{\{\size{\pi''} \mid \Type{\pi''}{\tm}\}} = \sizeZero{\anf} + 1 = \sizeZero{\tm}$\,. 
			
			\item $\tm = \anf\wnf$ with $\anf \in \anfSet$ and $\wnf \in \wnfSet$:
			by \ih applied to $\wnf$, $\min{\{\size{\pi''} \mid \Type{\pi''\allowbreak}{\wnf}\}} \allowbreak= \sizeZero{\wnf}$, in particular there is a derivation $\concl{\pi''}{\var_1 \colon\! P_1, \dots, \var_k \colon\! P_k}{\wnf}{P}$ for some positive types $P, P_1, \dots, P_k$ such that $\size{\pi} = \sizeZero{\wnf}$.
			For any positive type $Q$, by \ih~applied to $\anf$, there are positive types $P_1', \dots, P_k'$ and a derivation $\concl{\pi'}{\allowbreak\var_1 \colon\! P_1', \dots, \var_k \colon\! P_k'}{\anf}{[\Pair{P}{Q}]}$ such that $\sizeZero{\anf} = \size{\pi'} = \min\{\size{\pi'} \mid \Type{\pi'}{\anf}\}$.
			So, there is a derivation
			
			\begin{equation*}
				\pi =
				\AxiomC{$\ \vdots\,\pi'$}
				\noLine
				\UnaryInfC{$\var_1 \colon\! P_1', \dots, \var_k \colon\! P_k' \vdash \anf \colon\! [\Pair{P}{Q}]$}
				\AxiomC{$\ \vdots\,\pi''$}
				\noLine
				\UnaryInfC{$\var_1 \colon\! P_1, \dots, \var_k \colon\! P_k \vdash \wnf \colon\! P$}
				\RightLabel{\footnotesize{$@$}}
				\BinaryInfC{$\var_1 \colon\! P_1 \uplus P_1', \dots, \var_k \colon\! P_k \uplus P_k' \vdash \tm \colon\! Q$}
				\DisplayProof
			\end{equation*}
			
			\noindent where $\size{\pi} = \size{\pi'} + \size{\pi''} + 1 = \sizeZero{\anf} + \sizeZero{\wnf} + 1 = \sizeZero{\tm}$ (the last equation holds since $\anf$ is not an abstraction, see \refprop{syntactic-normal}).
			The last rule of any derivation $\Type{\pi'''}{\tm}$ is by necessity $@$, having derivations typing $\anf$ and $\wnf$ as premises, thus $\size{\pi'''} \geq \sizeZero{\anf} + \sizeZero{\wnf} + 1$ and hence $\min{\{\size{\pi'''} \mid \Type{\pi'''}{\tm}\}} = \sizeZero{\anf} + \sizeZero{\wnf} + 1 = \sizeZero{\tm}$\,. 
		\end{itemize}
		
		Finally, if $\tm = (\la{\var}{\wnf})\anf$ for some $\anf \in \anfSet$ and $\wnf \in \wnfSet$, then $\sizeZero{\tm} = \sizeZero{\wnf} + \sizeZero{\anf} + 1$ by definition.
		By \ih, applied to $\wnf$, there is a derivation $\concl{\pi'}{\var_1 \colon\! P_1', \dots, \var_k \colon\! P_k', \var \colon\! P}{\wnf}{Q}$ for some positive types $P_1', \dots, P_k', P, Q$ such that $\sizeZero{\wnf} = \size{\pi'} = \min{\{\size{\pi'} \mid \Type{\pi'}{\wnf}\}}$.
		By \ih applied to $\anf$, there exists a derivation $\concl{\pi''}{\var_1 \colon\! P_1, \dots, \var_k \colon\! P_k}{\anf}{P}$ for some positive types $P_1, \dots, P_k$ such that $\sizeZero{\anf} = \size{\pi''} = \min{\{\size{\pi''} \mid \Type{\pi''}{\anf}\}}$.
		Therefore, there is a derivation
		
		\begin{equation*}
			\pi =
			\AxiomC{$\ \vdots\,\pi'$}
			\noLine
			\UnaryInfC{$\var_1 \colon\! P_1', \dots, \var_k \colon\! P_k', \var \colon\! P \vdash \wnf \colon\! Q$}
			\RightLabel{\footnotesize{$\lambda$}}
			\UnaryInfC{$\var_1 \colon\! P_1', \dots, \var_k \colon\! P_k' \vdash \la{\var}{\wnf} \colon\! [\Pair{P}{Q}]$}
			\AxiomC{$\ \vdots\,\pi''$}
			\noLine
			\UnaryInfC{$\var_1 \colon\! P_1, \dots, \var_k \colon\! P_k \vdash \anf \colon\! P$}
			\RightLabel{\footnotesize{$@$}}
			\BinaryInfC{$\var_1 \colon\! P_1 \uplus P_1', \dots, \var_k \colon\! P_k \uplus P_k' \vdash \tm \colon\! Q$}
			\DisplayProof.
		\end{equation*}
		 
		\noindent where $\size{\pi} = \size{\pi'} + \size{\pi''} + 1 = \sizeZero{\wnf} + \sizeZero{\anf} + 1 = \sizeZero{\tm}$\,. 
		Given a derivation $\Type{\pi'''}{\tm}$, its last rule is by necessity $@$, having derivations typing $\anf$ and $\la{\var}{\wnf}$ as premises (the last rule of the latter is necessarily $\lambda$ having a derivation typing $\wnf$ as unique premise), thus $\size{\pi'''} \geq \sizeZero{\anf} + \sizeZero{\wnf} + 1$ and hence $\min{\{\size{\pi'''} \mid \Type{\pi'''}{\tm}\}} = \sizeZero{\anf} + \sizeZero{\wnf} + 1 = \sizeZero{\tm}$\,. 
		
		\item By definition, $\sizeZero{\tm} = 0$.
		According to \refcoro{minimal}, there is a derivation $\concl{\pi}{\,}{\tm}{\emptymset}$ such that $\size{\pi} = 0$, hence $\sizeZero{\tm} = \inf{\{\size{\pi} \mid \Type{\pi}{\tm}\}}$ since $\size{\pi'} \geq 0$ for any derivation $\Type{\pi'}{\tm}$.
		\qedhere
	\end{enumerate}
\end{proof}

\setcounter{propositionAppendix}{\value{prop:number-steps}}
\setcounter{equation}{1}
\begin{propositionAppendix}[Exact number of $\betavm$-steps]
\label{propappendix:number-steps}
\NoteState{prop:number-steps}
	Let $\tm$ be a $\shufm$-normalizable term and $\tm_0$ be its $\shufm$-normal form. 
	For every reduction sequence $\deriv \colon \tm \toshufm^* \tm_0$ and every $\Type{\pi}{\tm}$ and $\Type{\pi_0}{\tm_0}$ such that $\size{\pi} = \min \{\size{\pi'} \mid \Type{\pi'}{\tm}\}$ and $\size{\pi_0} = \min \{\size{\pi_0'} \mid \Type{\pi_0'}{\tm_0} \}$, one has
	\begin{equation}\label{eqAppendix:number-steps}
		\Lengbv{\deriv} = \size{\pi} - \sizeZero{\tm_0} = \size{\pi} - \size{\pi_0}\,.
	\end{equation}
	If moreover $\tm_0$ is a value, then $\Lengbv{\deriv} = \size{\pi}$.
\end{propositionAppendix}

\begin{proof}
	The statement concerning the case where $\tm_0$ is a value is an immediate consequence of Eq.~\refeq{number-steps} and \reflemmap{sizes}{value}.
	The second identity in Eq.~\refeq{number-steps} follows immediately from \reflemmap{sizes}{normal}.
	We prove the first identity in Eq.~\refeq{number-steps} by induction on $k = \Leng{\deriv} \in \nat$.
	
	If $k = 0$ then $\Lengbv{\deriv} = 0$ and $\tm = \tm_0$ (thus $\pi = \pi_0$), hence $\size{\pi} = \sizeZero{\tm} = \sizeZero{\tm_0}$ according to \reflemmap{sizes}{normal}, and therefore 
	$\Lengbv{\deriv} = 0 = \size{\pi} - \sizeZero{\tm_0}$.
	
	If $k > 0$ then $\deriv$ has the form $\tm \toshufm \tmp \toshufm^{n-1} \tm_0$ for some term $\tmp$; 
	let $\derivp$ be the sub-reduction sequence $\tmp \toshufm^{n-1} \tm_0$ in $\deriv$.
	There are two cases:
	\begin{itemize}
		\item \emph{A $\betav$-step} at the beginning of $\deriv$, \ie $\tm \tobvm \tmp$:
		according to the quantitative subject reduction for $\tobvm$ (\refpropp{quant-subject-reduction}{betav}), there is a derivation $\Type{\pi'}{\tmp}$ such that $\size{\pi'} = \size{\pi} - 2$.
		According to the quantitative subject expansion for $\tobvm$ (\refpropp{quant-subject-expansion}{betav}), for any derivation $\Type{\pi''}{\tmp}$ there exists a derivation $\Type{\pi'''}{\tm}$ such that $\size{\pi''} = \size{\pi'''} - 2$.
		Therefore, from the minimality of $\size{\pi}$ follows the minimality of $\size{\pi'}$ (among the derivations $\Type{\pi''}{\tmp}$).
		We can then apply the \ih to $\derivp$, so that
			$\Lengbv{\derivp} = \size{\pi'} - \sizeZero{\tm_0} = \size{\pi} - 1 - \sizeZero{\tm_0}$ 
		and hence $\Lengbv{\deriv} = \Lengbv{\derivp} + 1 = \size{\pi} - \sizeZero{\tm_0}$\,.
		\item \emph{A $\sigma$-step} at the beginning of $\deriv$, \ie $\tm \tosigm \tmp$:
		according to the quantitative subject reduction for $\tosigm$ (\refpropp{quant-subject-reduction}{sigma}), there is a derivation $\Type{\pi'}{\tmp}$ such that $\size{\pi'} = \size{\pi}$.
		By the quantitative subject expansion for $\tosigm$ (\refpropp{quant-subject-expansion}{sigma}), for any derivation $\Type{\pi''}{\tmp}$ there exists a derivation $\Type{\pi'''}{\tm}$ such that $\size{\pi''} = \size{\pi'''}$.
		Therefore, from the minimality of $\size{\pi}$ follows the minimality of $\size{\pi'}$ (among the derivations $\Type{\pi''}{\tmp}$).
		We can then apply the \ih to $\derivp$, so that
		$\Lengbv{\deriv} = \Lengbv{\derivp} = \size{\pi'} - \sizeZero{\tm_0} = \size{\pi} - \sizeZero{\tm_0}\,.$
		\qedhere
	\end{itemize}
\end{proof}

\begin{lemma}[Minimal derivation]
	\label{l:minimal}
	Let $\tm \in \Lambda$.
	If $\concl{\pi_0}{\,}{\tm}{\emptymset}$ then $\size{\pi_0} = \min{\{\size{\pi} \mid \Type{\pi}{\tm} \}}$.
\end{lemma}

\begin{proof}
	By \refprop{semantic-valuable}, from $\concl{\pi_0}{\,}{\tm}{\emptymset}$ it follows that there exists a reduction sequence $\deriv \colon \tm \toshufm^* \val$ for some $\val \in \valSet$ (which is $\shufm$-normal).
	We proceed by induction on the number $\Leng{\deriv} \in \Nat$ of $\shufm$-steps in $\deriv$.
	
	If $\Leng{\deriv} = 0$ then $\tm = \val$ and hence $\size{\pi_0} = 0 = \min{\{\size{\pi} \mid \Type{\pi}{\tm}\}}$ according to \reflemmap{semantic-value}{uniqueness}.
	
	Otherwise $\Leng{\deriv} > 0$ and hence $\deriv$ is the concatenation of $\tm \toshufm \tm'$ and a reduction sequence $\deriv' \colon \tm' \toshufm^* \val$ (for some term $\tm'$), so that $\Leng{\deriv} = \Leng{\deriv'} + 1$.
	There are two cases:
	\begin{itemize}
		\item either $\tm \tobvm \tm'$ and then, by the quantitative subject reduction for $\tobvm$ (\refpropp{quant-subject-reduction}{betav}), there exists $\concl{\pi'_0}{\,}{\tm'}{\emptymset}$ with $\size{\pi_0} = \size{\pi'_0} + 1$; 
		by \ih, $\size{\pi'_0} = \min{\{\size{\pi'} \mid \Type{\pi'}{\tm'}\}}$;
		suppose by absurd that $\size{\pi_0} \neq \min{\{\size{\pi} \mid \Type{\pi}{\tm}\}}$: so, there would be $\Type{\pi_{\min{\!}}}{\tm}$ such that $\size{\pi_{\min{}}} < \size{\pi_0}$ and hence, by quantitative 
		subject reduction (\refpropp{quant-subject-reduction}{betav}), there would be $\Type{\pi_{\min{\!}}'}{\tm'}$ such that $\size{\pi_{\min{}}'} = \size{\pi_{\min{}}} - 1 < \size{\pi_0} - 1 = \size{\pi_0'}$, which is impossible;
		therefore, $\size{\pi_0} = \min{\{\size{\pi} \mid \Type{\pi}{\tm}\}}$.
		
		\item or $\tm \tosigm \tm'$ and then, by the quantitative subject reduction for $\tosigm$ (\refpropp{quant-subject-reduction}{sigma}), there is a derivation $\concl{\pi'_0}{\,}{\tm'}{\emptymset}$ with $\size{\pi_0} = \size{\pi'_0}$; 
		by \ih, $\size{\pi'_0} = \min{\{\size{\pi'} \mid \Type{\pi'}{\tm'}\}}$;
		suppose by absurd that $\size{\pi_0} \neq \min{\{\size{\pi} \mid \Type{\pi}{\tm}\}}$: so, there would be $\Type{\pi_{\min{\!}}}{\tm}$ such that $\size{\pi_{\min{}}} < \size{\pi_0}$ and hence, by quantitative 
		subject reduction (\refpropp{quant-subject-reduction}{sigma}), there would be $\Type{\pi_{\min{\!}}'}{\tm'}$ such that $\size{\pi_{\min{}}'} = \size{\pi_{\min{}}}  < \size{\pi_0} = \size{\pi_0'}$, which is impossible;
		therefore, $\size{\pi_0} = \min{\{\size{\pi} \mid \Type{\pi}{\tm}\}}$.
		\qedhere
	\end{itemize}
\end{proof}

\setcounter{theoremAppendix}{\value{thm:number-steps-value}}
\begin{theoremAppendix}[Exact number of $\betavm$-steps for valuables]
  \label{thmappendix:number-steps-value}
  \NoteState{thm:number-steps-value}
  If $\tm \toshufm^*\! \val \in \valSet$ then  $\Lengbv{\tm} = \size{\pi}$
for $\concl{\pi}{\,}{\tm}{\emptymset}$.
\end{theoremAppendix}

\begin{proof}
According to \refprop{syntactic-normal}, $\tm$ is $\shufm$-normalizable; also, there exists $\concl{\pi}{\,}{\tm}{\emptymset}$ by \refprop{semantic-valuable}, since there is $\deriv \colon \tm \toshufm \val \in \valSet$.
By \reflemma{minimal},  $\size{\pi} = \min{\{\size{\pi'} \mid \Type{\pi'}{\tm}\}}$.
By \refprop{number-steps}, $\Lengbv{\tm} = \Lengbv{\deriv} = \size{\pi}$.
%
%
%
\end{proof}

\subsection{Omitted proofs and remarks of \refsect{conclusions}}

\setcounter{theoremAppendix}{\value{thm:characterize-normalizable-Plotkin}}
\begin{theoremAppendix}[Semantic and logical characterization of $\betavm$-normalization in the closed case]
	\label{thmappendix:characterize-normalizable-Plotkin}
	\NoteState{thm:characterize-normalizable-Plotkin}
	Let $\tm$ be a closed term.
	The following are equivalent:
	\begin{enumerate}
		\item\label{pappendix:characterize-normalizable-normalizable-Plotkin} \emph{Normalizability:} $\tm$ is $\betavm$-normalizable;
		\item\label{pappendix:characterize-normalizable-valuable-Plotkin} \emph{Valuability:} $\tm \tobvm^* \val$ for some closed value $\val$;

		\item\label{pappendix:characterize-normalizable-equivalen-normal-Plotkin} \emph{Completeness:} $\tm \betaveq \val$ for some closed  value $\val$;
		
		\item\label{pappendix:characterize-normalizable-nonempty-Plotkin}\emph{Adequacy:} $\sem{\tm}{\vec{\var}} \neq \emptyset$ for any list $\vec{\var} = (\var_1, \dots, \var_k)$ (with $k \in \nat$) of pairwise distinct variables;
		\item\label{pappendix:characterize-normalizable-semantic-Plotkin} \emph{Empty point:} $((\emptymset, \overset{k}{\dots}\,, \emptymset), \emptymset) \in \sem{\tm}{\vec{\var}}$ for any list $\vec{\var} = (\var_1, \dots, \var_k)$ ($k \in \nat$) of pairwise distinct variables;
		
		\item\label{pappendix:charachterize-normalizable-logic-Plotkin}\emph{Derivability with empty types:} there exists a derivation $\concl{\pi}{\,}{\tm}{\emptymset}$.
		
		\item\label{pappendix:characterize-normalizable-derivable-Plotkin}\emph{Derivability:} there exists a derivation $\concl{\pi}{\,}{\tm}{Q}$ for some positive type $ Q$;
		\item\label{pappendix:characterize-normalizable-strongly-normalizable-Plotkin}\emph{Strong normalizabilty:} $\tm$ is strongly $\betavm$-normalizable.
	\end{enumerate}
\end{theoremAppendix}

\begin{proof}
	\begin{description}
		\item[\eqref{p:characterize-normalizable-normalizable-Plotkin}$\Rightarrow$\eqref{p:characterize-normalizable-valuable-Plotkin}:] $\betavm$-normalizability of $\tm$ means that $\tm \tobvm^* \tmtwo$ for some $\betavm$-normal term $\tmtwo$.  As $\tm$ is closed, $\tmtwo$ is so, and hence $\tmtwo$ is a value according to \refcoro{syntactic-normal-closed}.
		
		\item[\eqref{p:characterize-normalizable-valuable-Plotkin}$\Rightarrow$\eqref{p:characterize-normalizable-equivalen-normal-Plotkin}:] Trivial, since $\tobvm \,\subseteq\, \tobv \,\subseteq\, \betaveq$.
		
		\item[\eqref{p:characterize-normalizable-equivalen-normal-Plotkin}$\Rightarrow$\eqref{p:characterize-normalizable-nonempty-Plotkin}:]
		First, note that any list of variables is suitable for $\tm$ and $\tmtwo$, since $\tm$ and $\tmtwo$ are closed.
		According to \refcoro{syntactic-normal-closed}, $\val$ is $\shufm$-normal;
		moreover $\tm \shufeq \val$ because $\betaveq \,\subseteq\, \shufeq$ (as $\tobv \,\subseteq\, \toshuf$).
		By the implication \eqref{p:characterize-normalizable-equivalen-normal}$\Rightarrow$\eqref{p:characterize-normalizable-nonempty} of \refthm{characterize-normalizable}, $\sem{\tm}{\vec{\var}} \neq \emptyset$.
		
		\item[\eqref{p:characterize-normalizable-nonempty-Plotkin}$\Rightarrow$\eqref{p:characterize-normalizable-semantic-Plotkin}:]
		Let $\vec{\var} = (\var_1, \dots, \var_k)$ (with $k \in \nat$) be a list of pairwise distinct variables: it is suitable for $\tm$ because $\tm$ is closed. 
		By the implication \eqref{p:characterize-normalizable-nonempty}$\Rightarrow$\eqref{p:characterize-normalizable-normalizable} of \refthm{characterize-normalizable}, from $\sem{\tm}{\vec{\var}} \neq \emptyset$ it follows that $\tm$ is $\shufm$-normalizable; moreover, the $\shufm$-normal form $\tmtwo$ of $\tm$ is closed (as $\tm$ is so) and hence $\tmtwo$ is a value by \refcoro{syntactic-normal-closed}. 
		By the implication \eqref{p:semantic-valuable-value}$\Rightarrow$\eqref{p:semantic-valuable-semantic} of \refprop{semantic-valuable}, $((\emptymset, \overset{k}{\dots}\,, \emptymset), \emptymset) \in \sem{\tm}{\vec{\var}}$.

		\item[\eqref{p:characterize-normalizable-semantic-Plotkin}$\Rightarrow$\eqref{p:characterize-normalizable-logic-Plotkin}:] Trivial, according to the definition of relational semantics (\refdef{semantics}).
		
		\item[\eqref{p:characterize-normalizable-logic-Plotkin}$\Rightarrow$\eqref{p:characterize-normalizable-derivable-Plotkin}]: Trivial.
		
		
		\item[\eqref{p:characterize-normalizable-derivable-Plotkin}$\Rightarrow$\eqref{p:characterize-normalizable-strongly-normalizable-Plotkin}:] 
		By the implication \eqref{p:characterize-normalizable-derivable}$\Rightarrow$\eqref{p:characterize-normalizable-strongly-normalizable} of \refthm{characterize-normalizable}, $\tm$ is strongly $\shufm$-normalizable, which implies that $\tm$ is strongly $\betavm$-normalizable since $\tobvm \,\subseteq\, \toshufm$.
		
		\item[\eqref{p:characterize-normalizable-strongly-normalizable-Plotkin}$\Rightarrow$\eqref{p:characterize-normalizable-normalizable-Plotkin}:] Trivial. 
		\qedhere
	\end{description}	
\end{proof}

\setcounter{corollaryAppendix}{\value{coro:same-number-Plotkin}}
\begin{corollaryAppendix}[Same number of $\betavm$-steps]
	\label{coroappendix:same-number-Plotkin}
	\NoteState{coro:same-number-Plotkin}
	Let $\tm$ be a closed $\betavm$-normalizable term and $\tm_0$ be its $\betavm$-normal form.
	For all reduction sequences $\deriv \colon \tm \tobvm^* \tm_0$ and $\derivp \colon \tm \tobvm^* \tm_0$, one has $\Lengbv{\deriv} = \Lengbv{\derivp}$.
\end{corollaryAppendix}

\begin{proof}
	Since $\tm$ is closed, its $\betavm$-normal form $\tm_0$ is closed as well and hence is $\shufm$-normal according to \refcoro{syntactic-normal-closed}.
	As $\tobvm \,\subseteq\, \toshufm$, one has $\deriv \colon \tm \toshufm^* \tm_0$ and $\deriv' \colon \tm \toshufm^* \tm_0$.
	Therefore, $\Lengbv{\deriv} = \Lengbv{\derivp}$ by \refcoro{same-number}.
\end{proof}

\setcounter{theoremAppendix}{\value{thm:number-steps-value-Plotkin}}
\begin{theoremAppendix}[
	Number of $\betavm$-steps]
	\label{thmappendix:number-steps-value-Plotkin}
	\NoteState{thm:number-steps-value-Plotkin}
	If $\tm $ is closed and $\betavm$-normalizable, then $\Lengbv{\tm} = \size{\pi}$ for $\concl{\pi}{\,}{\tm}{\emptymset}$.
\end{theoremAppendix}

\begin{proof}
	Since $\tm$ is closed, its $\betavm$-normal form $\tm_0$ is closed as well and hence is a value according to \refcoro{syntactic-normal-closed}.
	As $\tobvm \,\subseteq\, \toshufm$, one has $\tm \toshufm^* \tm_0$.
	By \refthm{number-steps-value}, $\Lengbv{\tm} = \size{\pi}$ where $\pi$ is the derivation $\concl{\pi}{\,}{\tm}{\emptymset}$. 
\end{proof}

\end{document}